\documentclass[aps,pra,10pt,letterpaper,twocolumn,tightenlines,superscriptaddress,notitlepage,longbibliography]{revtex4-2}

\usepackage[utf8]{inputenc}
\usepackage[english]{babel}
\usepackage[T1]{fontenc}
\usepackage{amsmath}
\makeatletter
\newcommand{\pushright}[1]{\ifmeasuring@#1\else\omit\hfill$\displaystyle#1$\fi\ignorespaces}
\newcommand{\pushleft}[1]{\ifmeasuring@#1\else\omit$\displaystyle#1$\hfill\fi\ignorespaces}
\makeatother

\allowdisplaybreaks

\usepackage{tikz}
\usepackage{lipsum}
\usepackage{amssymb}
\usepackage{color}
\usepackage{mathrsfs}
\usepackage{amsbsy}
\usepackage{amsthm}
\usepackage{array}
\newcolumntype{M}[1]{>{\centering\arraybackslash}m{#1}}
\newcolumntype{N}{@{}m{0pt}@{}}
\usepackage{balance}
\usepackage{thmtools,thm-restate}
\usepackage{graphicx}
\usepackage{tikz}

\usepackage{comment}

\usepackage{bbm}
\usepackage{bm}
\usepackage{epsfig}
\usepackage{xfrac}
\usepackage{xcolor}
\usepackage{enumerate}
\usepackage[shortlabels]{enumitem}
\usepackage{graphicx}
\usepackage{bm}
\usepackage{hyperref}
\usepackage{subfigure}
\hypersetup{
colorlinks=true,
linkcolor=blue,
filecolor=blue,
citecolor=blue,  
urlcolor=blue,
}

\newtheorem{definition}{Definition}
\newtheorem{proposition}{Proposition}
\newtheorem{theorem}{Theorem}
\newtheorem{theorem-main}{Theorem}

\newtheorem{lemma}{Lemma}

\usepackage{stackengine}
\newcommand\barbelow[1]{\stackunder[0.2pt]{$#1$}{\rule{1.1ex}{.1ex}}}
\newcommand\barbelowsmall[1]{\stackunder[0.1pt]{{\tiny $#1$}}{\rule{0.55ex}{.05ex}}}
\newcommand\barbelowmedium[1]{\stackunder[0.15pt]{{\scriptsize $#1$}}{\rule{0.75ex}{.075ex}}}

\usepackage{braket}
\renewcommand{\v}[1]{\ensuremath{\mathbf{#1}}} 
\newcommand{\gv}[1]{\ensuremath{\text{\boldmath$ #1 $}}}
\newcommand{\abs}[1]{\left| #1 \right|}
\newcommand{\norm}[1]{\left\| #1 \right\|} 

\newcommand{\trace}{\mathrm{Tr}}

\newcommand{\sM}{{\textsf{M}}}
\newcommand{\sU}{{\textsf{U}}}

\newcommand{\tp}{{\tilde p}}

\newcommand{\mS}{{\mathcal{S}}}

\newcommand{\mA}{{\mathcal{A}}}
\newcommand{\mC}{{\mathcal{C}}}

\newcommand{\mU}{{\mathcal{U}}}
\newcommand{\mN}{{\mathcal{N}}}
\newcommand{\mM}{{\mathcal{M}}}
\newcommand{\mR}{{\mathcal{R}}}

\newcommand{\frakm}{{\mathfrak{m}}}

\newcommand{\tK}{{\tilde K}}
\newcommand{\ubeta}{\barbelow{\beta}}
\newcommand{\ub}{\barbelow{b}}
\newcommand{\ugamma}{\barbelow{\gamma}}

\newcommand{\mE}{{\mathcal{E}}}

\newcommand{\id}{{\mathbbm{1}}}

\newcommand{\frakh}{{\mathfrak{h}}}

\newcommand{\vh}{{\v{h}}}

\newcommand{\vi}{{\gv{i}}}
\newcommand{\vm}{{\gv{m}}}

\newcommand{\ve}{{\gv{e}}}

\newcommand{\vv}{{\gv{v}}}
\newcommand{\vt}{{\gv{t}}}
\newcommand{\vw}{{\gv{w}}}
\newcommand{\vs}{{\gv{s}}}
\newcommand{\va}{{\gv{a}}}
\newcommand{\vb}{{\gv{b}}}
\newcommand{\vsig}{{\gv{\sigma}}}
\newcommand{\vvsig}{{\gv{\varsigma}}}

\newcommand{\vK}{{\v{K}}}

\newcommand{\bR}{{\mathbb{R}}}

\newcommand{\bE}{{\mathbb{E}}}
\newcommand{\bP}{{\mathbb{P}}}
\newcommand{\bC}{{\mathbb{C}}}

\renewcommand{\Re}{{\mathrm{Re}}}
\renewcommand{\Im}{{\mathrm{Im}}}

\renewcommand{\epsilon}{\varepsilon}
\newcommand{\appropto}{\mathrel{\vcenter{
  \offinterlineskip\halign{\hfil$##$\cr
    \propto\cr\noalign{\kern2pt}\sim\cr\noalign{\kern-2pt}}}}}

\definecolor{fluorescentpink}{rgb}{1.0, 0.08, 0.58}

\let\baraccent=\= 
\renewcommand{\=}[1]{\stackrel{#1}{=}} 

\newcommand{\thmref}[1]{\hyperref[#1]{Theorem~\ref{#1}}}
\newcommand{\lemmaref}[1]{\hyperref[#1]{Lemma~\ref{#1}}}
\newcommand{\propref}[1]{\hyperref[#1]{Proposition~\ref{#1}}}
\newcommand{\corollaryref}[1]{\hyperref[#1]{Corollary~\ref{#1}}}
\newcommand{\figref}[1]{\hyperref[#1]{Fig.~\ref{#1}}}
\newcommand{\tabref}[1]{\hyperref[#1]{Table~\ref{#1}}}
\newcommand{\figaref}[1]{\hyperref[#1]{Fig.~\ref{#1}(a)}}
\newcommand{\figbref}[1]{\hyperref[#1]{Fig.~\ref{#1}(b)}}
\newcommand{\figcref}[1]{\hyperref[#1]{Fig.~\ref{#1}(c)}}
\newcommand{\figdref}[1]{\hyperref[#1]{Fig.~\ref{#1}(d)}}
\newcommand{\figeref}[1]{\hyperref[#1]{Fig.~\ref{#1}(e)}}
\newcommand{\figfref}[1]{\hyperref[#1]{Fig.~\ref{#1}(f)}}
\renewcommand{\eqref}[1]{\hyperref[#1]{Eq.~(\ref{#1})}}
\newcommand{\secref}[1]{\hyperref[#1]{Sec.~\ref{#1}}}
\newcommand{\eqsref}[2]{\hyperref[#1]{Eqs.~(\ref{#1})-(\ref{#2})}}
\newcommand{\appref}[1]{\hyperref[#1]{Appx.~\ref{#1}}}

\begin{document}

\title{Limits of noisy quantum metrology with restricted quantum controls}

\author{Sisi Zhou}\email{sisi.zhou26@gmail.com}
\affiliation{Perimeter Institute for Theoretical Physics, Waterloo, Ontario N2L 2Y5, Canada}
\affiliation{Institute for Quantum Information and Matter, California Institute of Technology, Pasadena, CA 91125, USA}
\affiliation{Department of Physics and Astronomy and Institute for Quantum Computing, University of Waterloo, Ontario N2L 3G1, Canada}

\date{\today}

\begin{abstract}
The Heisenberg limit (HL, with estimation error scales as $1/n$) and the standard quantum limit (SQL, $\propto 1/\sqrt{n}$) are two fundamental limits in estimating an unknown parameter in $n$ copies of quantum channels and are achievable with full quantum controls, e.g., quantum error correction (QEC). It is unknown though, whether these limits are still achievable in restricted quantum devices when QEC is unavailable, e.g., with only unitary controls or bounded system sizes. In this Letter, we discover various new limits for estimating qubit channels under restrictive controls. The HL is shown to be unachievable in various cases, indicating the necessity of QEC in achieving the HL. Furthermore, a necessary and sufficient condition to achieve the SQL is determined, where a single-qubit unitary control protocol is identified to achieve the SQL for certain types of noisy channels, and for other cases a constant floor on the estimation error is proven. A practical example of the unitary protocol is provided. 
\end{abstract}

\maketitle

\emph{Introduction.---}
Quantum metrology studies parameter estimation of small unknown parameters in quantum systems~\cite{giovannetti2011advances,degen2017quantum,pezze2018quantum,pirandola2018advances}. It has a wide range of application scenarios, e.g., graviational wave detection~\cite{caves1981quantum,yurke19862,ligo2011gravitational,ligo2013enhanced}, quantum imaging~\cite{le2013optical,lemos2014quantum,tsang2016quantum,abobeih2019atomic}, quantum magnetometry~\cite{wineland1992spin,bollinger1996optimal,leibfried2004toward,taylor2008high,zhou2020quantum}, atomic clocks~\cite{rosenband2008frequency,appel2009mesoscopic,ludlow2015optical,kaubruegger2021quantum,marciniak2022optimal}, etc. 
A typical task in quantum metrology is to estimate an unknown parameter encoded in a one-parameter quantum channel given $n$ copies of it. The Heisenberg limit (HL) is the ultimate estimation limit~\cite{giovannetti2006quantum}, stated as $\Delta \hat\theta \propto 1/n$, where $\Delta \hat\theta$ is the estimation precision of an unknown parameter $\theta$. Chasing the HL is a central goal in quantum enhanced metrology. 

To achieve the HL for unitary channels, one can either prepare an entangled state in a large multi-probe system (e.g., the GHZ state~\cite{giovannetti2004quantum,giovannetti2006quantum} or the spin-squeezed state~\cite{wineland1992spin,kitagawa1993squeezed}) or apply quantum channels sequentially on a single probe and measure after a long time of evolution~\cite{ramsey1950molecular}. However, in the presence of noise, the standard quantum limit (SQL), $\Delta \hat\theta \propto 1/\sqrt{n}$ sometimes becomes the optimal limit one can achieve~\cite{huelga1997improvement,ulam2001spin,escher2011general,demkowicz2012elusive,demkowicz2014using}. The SQL is considered more accessible than the HL and can be achieved using only product states or repeated measurements. Recent works~\cite{sekatski2017quantum,demkowicz2017adaptive,zhou2018achieving,zhou2021asymptotic} revealed a necessary and sufficient condition to achieve of the HL using quantum controls, namely, the ``Hamiltonian-not-in-Kraus-span'' (HNKS) condition~\cite{sekatski2017quantum,demkowicz2017adaptive,zhou2018achieving,zhou2021asymptotic}. When satisfied, there exist quantum error correction (QEC) protocols achieving the HL~\cite{kessler2014quantum,arrad2014increasing,unden2016quantum,dur2014improved,zhou2021asymptotic}. 
These protocols, however, are resource demanding in that it requires a real-time measurement and feed-forward QEC procedure, along with either an noiseless ancillary system or a long-range entangled multi-probe system, making experimental demonstration particularly difficult in practice. 

In this Letter, we investigate the metrological limits for noisy quantum metrology with less resource-demanding, restricted quantum controls. We consider estimating one-parameter qubit channels when the controls are unital and show that for any noisy qubit channels (i.e., non-unitary channels), the metrological limit cannot surpass the SQL. We also consider another scenario with full controls but no noiseless ancillary qubits and again find the HL to be unattainable. 

Besides the unachievability of the HL, we also resolve the achievability of the SQL. Remarkably, we find a necessary and sufficient condition called the RGNKS condition, which is similar but less stringent than the HNKS condition, for achieving the SQL using restricted quantum controls. When it holds, we show how to achieve the SQL using single-qubit unitary controls; when violated, we show the estimation precision has a constant floor with only unital controls. We finish with an example demonstrating our SQL-achieving unitary control protocol, which achieves better estimation precision than the traditional repeated measurement protocol without controls, even under state preparation and measurement (SPAM) noise. 

\emph{Setting.---}
Consider a one-parameter quantum state $\rho_\theta$, the estimation error $\Delta \hat\theta$ of an unbiased estimator $\hat\theta$, defined to be the standard deviation of $\hat \theta$ at the true value $\theta$ (assumed to be $\theta = 0$ in this Letter for simplicity~\footnote{We will use $\dot{f_\theta}$ to denote the derivative of any function $f_\theta$ with respect to $\theta$ at $\theta = 0$. We will also implicitly assume functions of $\theta$ are taken at its true value $\theta = 0$ and sometimes drop the subscript $_\theta$ when there is no ambiguity.}), satisfies the quantum Cram\'er--Rao bound~\cite{holevo2011probabilistic,helstrom1976quantum} $\Delta \hat\theta \geq {1}/{\sqrt{N_{\rm expr} F(\rho_\theta)}}$, where $N_{\rm expr}$ is the number of repeated experiments, and $F(\rho_\theta)$ is the quantum Fisher information (QFI)~\cite{braunstein1994statistical,paris2009quantum,petz1996geometries}. The bound is often saturable as $N_{\rm expr} \rightarrow \infty$ using maximum likelihood estimators~\cite{barndorff2000fisher,kay1993fundamentals,casella2002statistical,kobayashi2011probability} and an optimal choice of positive operator-valued measure (POVM) on $\rho_\theta$~\cite{braunstein1994statistical}, making QFI a suitable measure of the sensitivity of a quantum sensor. For detailed preliminaries on QFI and related theories, see~\cite{SM}. 
\nocite{kolodynski2013efficient,kolodynski2014precision,katariya2020geometric,helstrom1967minimum,yuan2017fidelity,perez2006contractivity}

\begin{figure}[tbp]
    \centering
    \includegraphics[width=0.47\textwidth]{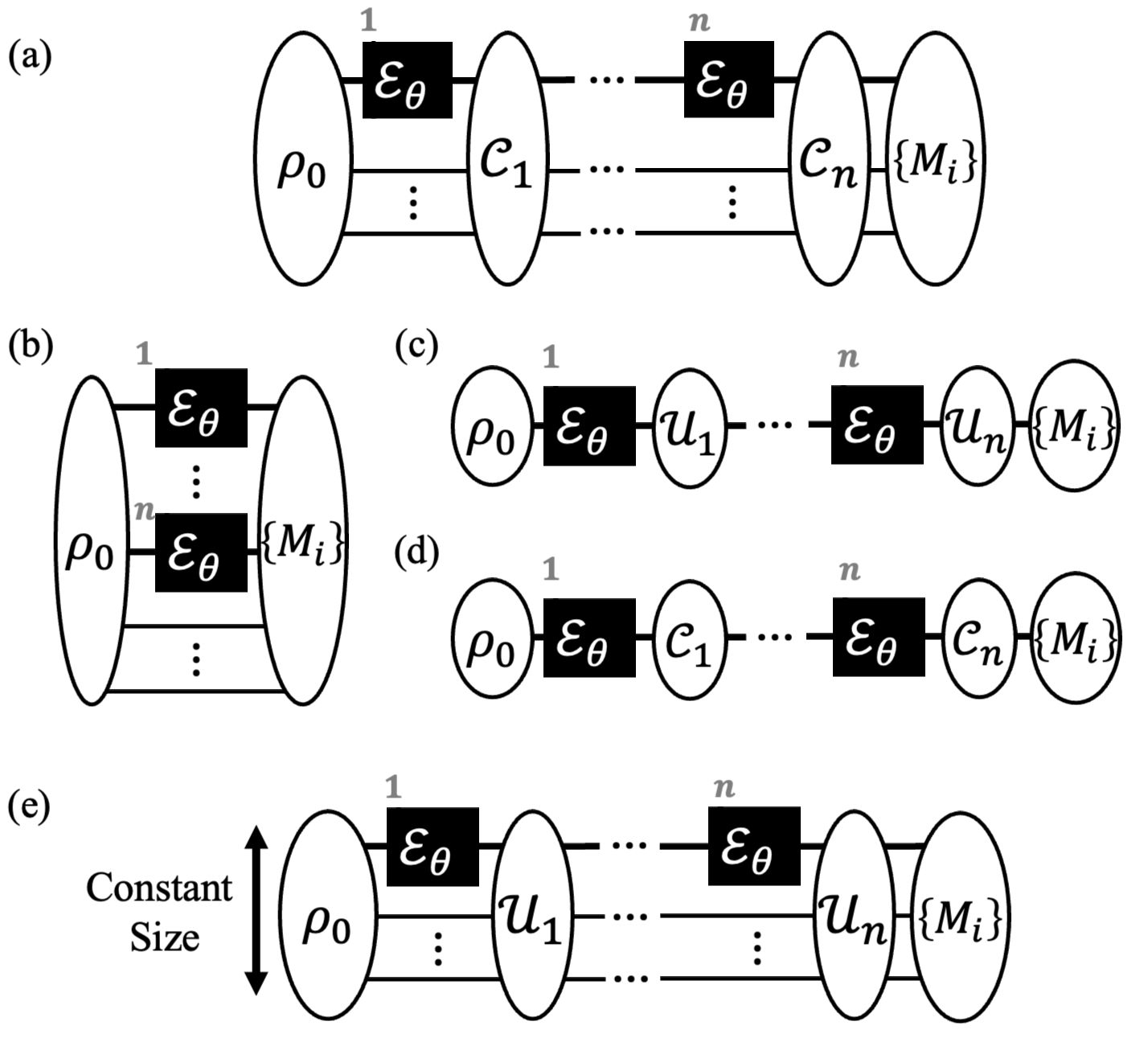}
    \caption{Different quantum metrological strategies for channel estimation. (a) Ancilla-assisted sequential strategy. (b) Parallel strategy. (c) Ancilla-free sequential strategy with unital controls. (d) Ancilla-free sequential strategy. (e) Bounded-ancilla-assisted sequential strategy with unital controls. The inclusion relations of them are: (a) $\supseteq$ (b), and (a) $\supseteq$ (d) (or (e)) $\supseteq$ (c), where A $\supseteq$ B means strategy A can simulate strategy B. Here $\mE_\theta$ is the one-parameter quantum channel to be estimated. $\rho_0$, $\{M_i\}$, $\mC_k$ and $\mU_k$ can be arbitrary quantum states, POVM, channels and unital channels, which function as quantum controls to enhance the estimation precision. The size of the noiseless ancilla is unbounded except in (e).}
    \label{fig:strategies}
\end{figure}

Given multiple copies of one-parameter quantum channels (i.e., completely positive and trace preserving (CPTP) maps~\cite{nielsen2002quantum,watrous2018theory}) $\mE_\theta$, it is a central topic to study how to maximize the QFI of output states with respect to $n$, the number of times $\mE_\theta$ is applied in the system. 
The most powerful strategy we consider is the \emph{ancilla-assisted sequential strategy} (see \figaref{fig:strategies}), where an unbounded noiseless ancilla~\footnote{Later on, when we say \emph{ancilla} we usually mean the noiseless ancillary system that extends the probe system to a larger Hilbert space. We do not explicitly discuss the ancillary qubits used in practice to implement CPTP controls, e.g., through Stinespring dilation.} is available and arbitrary quantum channels, as quantum controls, are applied sequentially. 
It includes the \emph{parallel strategy} (see \figbref{fig:strategies}), where the one-parameter channel acts once on each probe of the arbitrarily large multi-probe system in parallel.

\begin{table*}[htbp]
\centering
\includegraphics[width=0.75\linewidth]{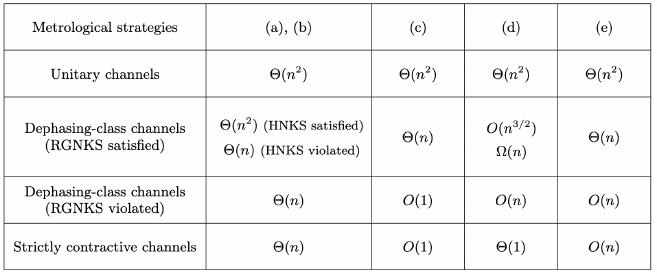}
\caption{Summary: Scalings of the QFI for different metrological strategies and classes of one-parameter qubit channels. }
\label{table}
\end{table*}

Previous works~\cite{sekatski2017quantum,demkowicz2017adaptive,zhou2018achieving,zhou2021asymptotic} studied strategies (a) and (b) comprehensively. The optimal QFI achievable using these strategies follows either the HL (QFI $= \Theta(n^2)$~\footnote{We use the big-O notation for non-negative functions where $f(n) = O(g(n))$ means $f(n) \leq c g(n)$ for some constant $c$ and large enough $n$. Similarly, $\Omega$, $\Theta$ and $o$ are also used in a conventional manner.}), or the SQL (QFI = $\Theta(n)$). Given a quantum channel $\mE_\theta = i \sum_i K_{\theta,i}(\cdot)K_{\theta,i}^\dagger$ represented using Kraus operators, the HL is achievable if and only if the HNKS condition holds, i.e., $H \notin \mS$, where $H := i\sum_{j} K_{\theta,j}^\dagger \dot K_{\theta,j}$ is the so-called Hamiltonian representing the total quantum signal, and $\mS := {\rm span}\{K_{\theta,i}^\dagger K_{\theta,j},\forall i,j\}$ is the Kraus span representing the quantum noise, where ${\rm span}\{\cdot\}$ denotes all Hermitian operators spanned by operators in $\{\cdot\}$~\footnote{Note that although $H$ is not uniquely defined given a quantum channel, $H\notin\mS$ is a well-defined condition with no ambiguity.}. However, to achieve the HL when the HNKS condition holds, perfect ancilla-assisted or multi-probe QEC is required, making it difficult to realize in practice. 

In this Letter, we comprehensively study the metrological limits (i.e., the optimally achievable scalings of the output QFI with respect to $n$) with less demanding restrictive quantum controls for estimating one-parameter \emph{qubit} quantum channels, which is one of the most commonly encountered types of metrological scenarios. \figcref{fig:strategies} shows the \emph{ancilla-free sequential strategy with unital controls}, which is the most restrictive strategy we consider where no noiseless ancilla is allowed and only unital quantum controls are allowed. \figdref{fig:strategies} shows the \emph{ancilla-free sequential strategy} that allows general CPTP controls compared to the previous one. \figeref{fig:strategies} shows the \emph{bounded-ancilla-assisted sequential strategy with unital controls} which allows a bounded noiseless ancilla but only unital controls.

These strategies, with minimum requirements on control types or system size, characterize practical limitations of quantum devices. Strategy (c) can characterize devices where two-qubit gates are too noisy and only single-qubit gates are available. Strategy (d) can characterize devices where CPTP controls are implementable using two-qubit gates and short-lived ancillary qubits (e.g. through Stinepring dilation), while long-lived (noiseless) ancillary qubits are unavailable. Strategy (e) can simulate an important multi-probe sensing scenario where $\mE_\theta^{\otimes k}$ is applied on a $k$-probe system for $n/k$ steps with interleaved unital controls, where $k$ is a constant.  

It is challenging to identify the metrological limits achievable using these strategies. First, previous works~\cite{escher2011general,demkowicz2012elusive,demkowicz2014using,demkowicz2017adaptive,zhou2018achieving,zhou2021asymptotic,wan2022bounds,kurdzialek2022using,altherr2021quantum,liu2023optimal} on deriving the metrological limits for strategies (a) and (b) relied heavily on the channel extension method~\cite{fujiwara2008fibre}, which requires extending the quantum system using noiseless ancilla and is thus not directly applicable in our case. Second, we do not allow mid-circuit measurements in these restricted strategies, and even the SQL that is considered trivially achievable using repeated measurements may not be achievable in our case. 

\emph{Results.---}
Here we present the results of our work, as summarized in \tabref{table}, which shows the scalings of the QFI of metrological strategies with $n$ one-parameter qubit channels.

We first explain the classification of qubit channels. Any qubit channel $\mE$ belongs to one of the three classes~\cite{king2001minimal,SM}: (1)~\emph{Unitary channel} that is described by a unitary operator acting on the system; (2)~\emph{Dephasing-class channel} that is equivalent to a dephasing channel up to unitary rotations, e.g., $V \mE (U \rho U^\dagger)V^\dagger = (1-p) \rho + p Z \rho Z$ for some unitaries $U$ and $V$ where $Z$ is the Pauli-Z operator and $p \in (0,1/2]$ is the noise probability; (3)~\emph{Strictly contractive channel} that satisfies $\norm{\mE(A)}_1 < \norm{A}_1$ for any traceless Hermitian operator, where $\norm{\cdot}_1$ is the trace norm. In particular, any qubit channel $\mE$ can be represented through its action on $\rho = \frac{1}{2}(\id + \vw\cdot\vsig)$ ($\vw \in \bR^3$), that is, 
\begin{equation}
\label{eq:qubit-main}
V \mE (U \rho U^\dagger)V^\dagger = \frac{1}{2}(\id + (\vs + \Lambda\vw)\cdot\vsig),
\end{equation}
where $\vsig = (X,Y,Z)$ is the vector of Pauli matrices and $\Lambda$ is a diagonal matrix~\cite{king2001minimal}. Unitary, dephasing-class, and strictly contractive channels correspond to the cases where the singular values of $\Lambda$ are all equal to $1$, have one element equal to $1$, and are all smaller than $1$.

For one-parameter qubit channels $\mE_\theta$, we further divide them based on the dependence of $\mE_\theta$ on $\theta$. We define the \emph{``rotation-generators-not-in-Kraus-span''} (RGNKS) condition to be that there exists a choice of unitary rotations $U_\theta$ and $V_\theta$ in \eqref{eq:qubit-main} such that either $H_0:= -iU_\theta^\dagger \dot U_\theta \notin \mS$ or $H_1:= -iV_\theta^\dagger \dot V_\theta \notin \mS$. $H_{0,1}$ can be interpreted as the quantum signals before and after the noise channel. (Note that in this Letter we assume regularity conditions that the chosen $U_\theta,V_\theta,\vs_\theta,\Lambda_\theta$ are differentiable; and $\mE_\theta$ remains the same class of qubit channels in the neighborhood of $\theta = 0$.)

The RGNKS condition is a necessary condition of the HNKS condition and they behave similarly. Both conditions are satisfied for non-trivial unitary channels, and are violated for strictly contractive channels (as we will see later). For dephasing-class channels, $\mS = {\rm span}\{\id,UZU^\dagger\}$ and $H = U ((1-p_\theta)H_1+p_\theta ZH_1Z)U^\dagger + H_0$. The RGNKS does not imply the HNKS. An example is 
\begin{equation}
\label{eq:example}
    \mE_\theta(\cdot) = e^{-i\theta X} \big( (1-p)(\cdot) + p Z(\cdot)Z \big) e^{i\theta X},
\end{equation}
where $p$ does not depend on $\theta$, which is a composition of dephasing noise and Pauli-X rotation (see also e.g., \cite{kessler2014quantum,arrad2014increasing,dur2014improved,unden2016quantum}). The HNKS and RGNKS both hold when $p < 1/2$, but when $p = 1/2$ the RGNKS condition holds while the HNKS fails. Note that if $e^{-i\theta X}$ in \eqref{eq:example} is replaced by Pauli-Z rotation $e^{-i\theta Z}$, $\mE_\theta$ then violates both the HNKS and RGNKS. 

We have explained above the classification of qubit channels, as shown in \tabref{table}. It consists of two essential groups of results that we present below: (1)~For all noisy qubit channels, the HL is not achievable under restricted strategies. (2)~The RGNKS condition determines the achievability of the SQL under restricted strategies. 

\emph{Unachievability of the HL.---}
Here we summarize the results on the unachievability of the HL for noisy qubit channels, with detailed proofs in~\cite{SM}. 
\begin{theorem-main}
\label{thm:HL-contractive}
For strictly contractive qubit channels, the HNKS condition is violated, and it implies a QFI upper bound of $O(n)$ using all strategies in \figref{fig:strategies}. 
\end{theorem-main}
\begin{theorem-main}
\label{thm:HL-dephasing}
For dephasing-class channels, the QFI has an upper bound of $O(n)$ using strategies (c) and (e); and an upper bound of $O(n^{3/2})$ using strategies (d). 
\end{theorem-main}

\thmref{thm:HL-contractive} and \thmref{thm:HL-dephasing} establish the unachievability of the HL for all noisy qubit channels. The SQL cannot be surpassed in most cases, but note that when applying strategy (d) to dephasing-class channels satisfying the RGNKS, the possibility of a QFI up to $O(n^{3/2})$ has not been ruled out, according to \thmref{thm:HL-dephasing}.

To prove \thmref{thm:HL-contractive}, we decompose the Kraus operators of the quantum channel into linear combinations of Pauli operators, and prove the Kraus span $\mS$ contains all Hermitian operators, which implies the violation of the HNKS condition. Note that a similar technique was used in~\cite{sekatski2017quantum} for a restrictive scenario. 

To prove \thmref{thm:HL-dephasing}, we further develop and refine the channel extension method~\cite{fujiwara2008fibre,escher2011general,demkowicz2012elusive}. Taking the channel in \eqref{eq:example} as an example, for strategies (c), we show the output state $\rho_\theta$ satisfies 
\begin{equation}
\label{eq:upper-main}
    F(\rho_\theta) \leq \sum_{k=1}^n 4 \trace(\alpha_k) + \sum_{k=1}^{n-1} 8 \trace(\ugamma_k \beta_{k+1}) = O(n),
\end{equation}
where $\alpha_k = \alpha = \id$, $\beta_k = \beta = (1-2p) X$ and $\ugamma_k$ is defined recursively via $\ugamma_k = \mU_k\circ((1-p)\ugamma_{k} + p Z \ugamma_{k} Z + X)$ and $\ugamma_0 = 0$. (For general qubit channels, $\alpha_k,\beta_k$ and $\ugamma_k$ are all functions of the Kraus operators of $\mU_k \circ \mE_\theta$.) Compared to the traditional channel extension method which gives~\cite{kurdzialek2024quantum} 
\begin{equation}
    F(\rho_\theta) \leq 4 n \norm{\alpha} + 4 n(n-1) (\norm{\beta}^2 + o(1)) = O(n^2), 
\end{equation}
where $\norm{\cdot}$ is the operator norm, our method utilizes the algebraic structure (orthogonality of Pauli operators) in the recurrence relation of $\ugamma_k$ which allows the second term in the bound to be tightened to $O(n)$. The SQL upper bound is proven in a similar manner for stratgies (e). For strategies (d), however, $\trace(\alpha_k)$ in \eqref{eq:upper-main} is modified to $\trace(\iota_{k-1}\alpha_k)$, 
where $\iota_{k} = \mC_{k} \circ \mE_{\theta} \circ \cdots \circ \mC_{1} \circ \mE_{\theta} (\id)$ is no longer necessarily equal to $\id$, and the recurrence relation of $\ugamma_k$ also depends on $\iota_{k-1}$. As a result, we can only show an upper bound of $O(n^{3/2})$. The proof consists of two steps: (1)~We first prove an upper bound of $O(n^{3/2})$ when controls are not too non-unital, i.e., $\abs{\trace(\iota_k Z)} \leq 1$ for all $k$, by proving $\norm{\ugamma_k} = O(k^{1/2})$; (2)~In general cases, we partition the channel into segments of not-too-non-unital channels and then use a second layer of the channel extension method to bound the QFI of the entire channel, viewing each segment as individual channels.

\emph{Achievability of the SQL.---}
It remains to answer whether the above upper bounds are achievable using restricted strategies. The results are summarized below with detailed proofs in~\cite{SM}. 
\begin{theorem-main}
\label{thm:SQL-1}
When the RGNKS condition is satisfied, there exists a single-qubit gate as the control operations in strategies (c) that achieves a QFI of $\Theta(n)$.
\end{theorem-main}
\begin{theorem-main}
\label{thm:SQL-2}
For dephasing-class channels that violates the RGNKS condition, the QFI is at most $O(1)$ using strategies (c).
\end{theorem-main}
\begin{theorem-main}
\label{thm:SQL-3}
For strictly contractive qubit channels, the QFI is at most $O(1)$ using strategies (c) and (d).
\end{theorem-main}

\thmref{thm:SQL-1} provides a single-qubit unitary control that achieves the SQL when the RGNKS is satisfied, and this protocol will be illustrated in an example later. \thmref{thm:SQL-2} and \thmref{thm:SQL-3} are no-go results showing the constant QFI upper bounds when the RGNKS condition is violated. In the case of \thmref{thm:SQL-2}, the channel, up to unitary rotations, must be a composition of Pauli-Z rotation and dephasing noise, and the constant QFI bound is proven through the contraction relation of the norm of the Bloch vector at each step. 

\thmref{thm:SQL-3} applies more generally to CPTP controls. To prove it, we use the relation between the QFI and the Bures distance~\cite{zhou2019exact,hubner1992explicit}, 
$F(\rho_\theta) = \lim_{d\theta\rightarrow 0}\frac{d_B^2(\rho_{\theta+d\theta},\rho_{\theta-d\theta})}{(d\theta)^2}$, to prove for any (qudit) channel $\mE_\theta$, 
\begin{equation}
\label{eq:contraction-main}
    F(\mC_n\circ\mE_\theta\circ\cdots \circ \mC_1\circ\mE_\theta(\rho_0)) \leq \sum_{k=0}^{n-1}\eta(\mE)^{n-k-1} F(\mE_\theta),
\end{equation}
where $\eta(\cdot)$ is the contraction coefficient with respect to QFI, defined by $\eta(\mN) := \sup_{\sigma_\theta} \frac{F(\mN(\sigma_\theta))}{F(\sigma_\theta)}$, viewing $\mN$ as a parameter-independent channel, and $F(\mE_\theta) := \sup_{\sigma} F(\mE_\theta(\sigma))$. For all strictly contractive qubit channels, we prove $\eta(\mE) < 1$, which leads to the constant upper bound of the QFI. Note that $\eta(\mE)$ is usually difficult to compute~\cite{hiai2016contraction}, but we manage to bound it using the notion of quantum preprocessing-optimized FI~\cite{len2021quantum,zhou2023optimal}. This technique might find applications in other metrological scenarios.  

\emph{Example.---} Here we analyze in detail an example with dephasing noise and Pauli-X rotation, where $\mE_\theta$ is \eqref{eq:example}. It satisfies the HNKS condition, i.e., $H = (1-2p) X \notin \mS = {\rm span}\{\id,Z\}$ when $p < 1/2$. A QFI of $4(1-2p)^2n^2$ is achievable through QEC. In strategy (a), one can use the two-qubit repetition code with input state $\frac{1}{\sqrt{2}}(\ket{+}\ket{0}_A+\ket{-}\ket{1}_A)$~\cite{kessler2014quantum,arrad2014increasing,unden2016quantum}, where $\ket{\pm} = \frac{1}{\sqrt{2}}(\ket{0}\pm\ket{1})$ and $\ket{\cdot}_A$ is a noiseless ancillary qubit. Each step Pauli-Z errors are corrected by measuring $X \otimes Z_A$ and applying $Z \otimes \id_A$ whenever the syndrome is $-1$. 
Similarly, one can also use the multi-probe repetition code $\frac{1}{\sqrt{2}}(\ket{+}^{\otimes n}+\ket{-}^{\otimes n})$ in strategy (b) to achieve the HL~\cite{dur2014improved}. 

\begin{figure}[tbp]
    \centering
    \includegraphics[width=0.47\textwidth]{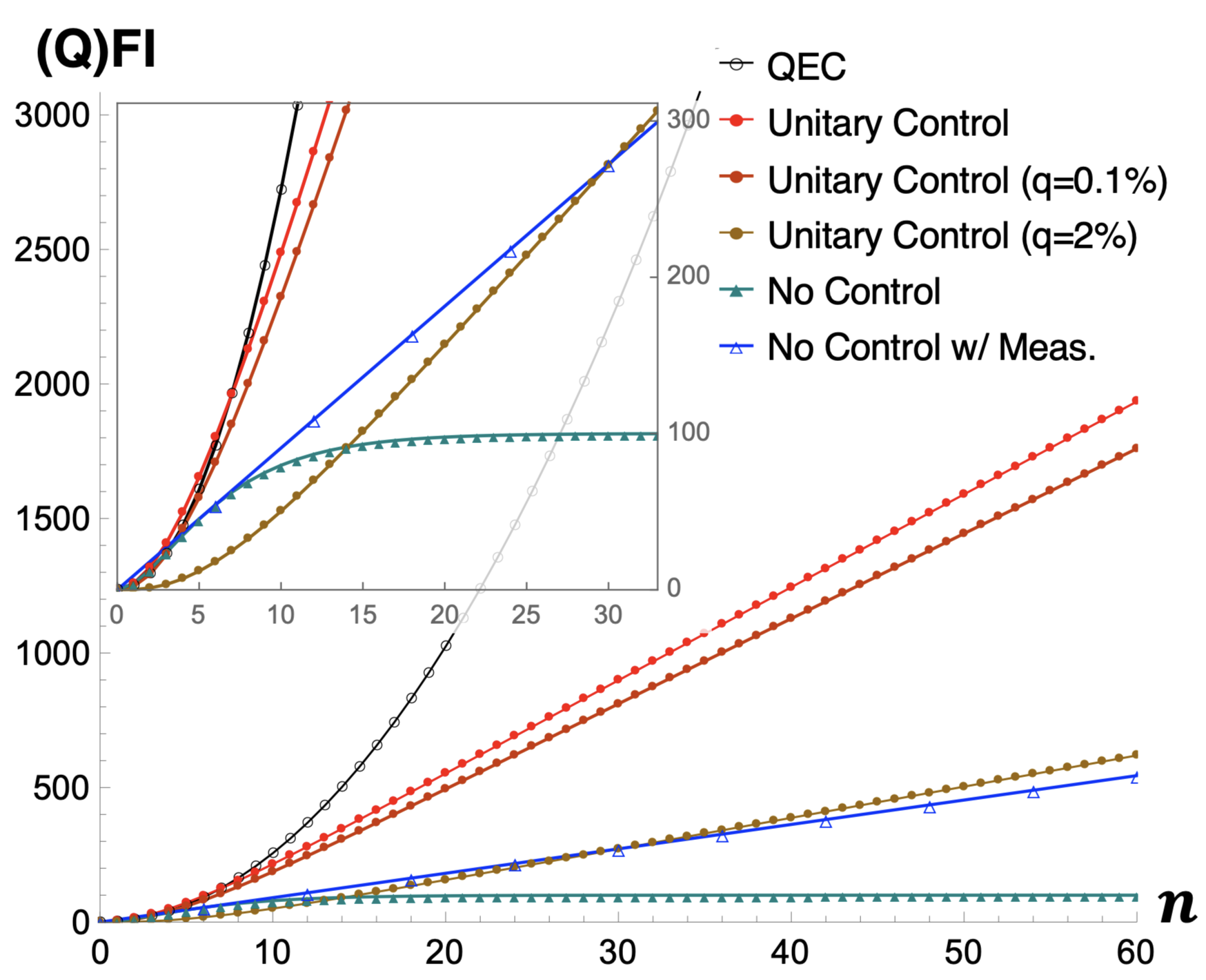}
    \caption{QFI (or FI) of different metrological strategies using $n$ copies of channel \eqref{eq:example} when $p = 0.1$. From top to bottom, the strategies are QEC (empty circle, black), unitary control protocols ($w = 0.01$) with SPAM noise rates $q = 0$, $0.1\%$ and $2\%$ (circle, red/dark red/brown), repeated measurement protocol without control (empty triangle, blue) and the naive protocol without control or intermediate measurement (triangle, teal). The QEC protocol achieves the HL, the unitary control protocols and the repeated measurement protocol achieve the SQL, and the protocol without control or intermediate measurement achieves only a constant QFI. The inset plot zooms in and shows the crossovers of curves when the (Q)FIs are small. } 
    \label{fig:example}
\end{figure}

With restricted strategies, the HL is unachievable, but the SQL is achievable using single-qubit unitary controls. Here, we let the input state be $\ket{0}$, and apply the same $U = \exp({-i \frac{1}{2}\sqrt{\frac{w}{n}} X})$ each step after $\mE_\theta$. (Here $w$ can be any positive constant, and the asymptotic behavior of the QFI is optimal when $w \ll p$.) 
Consider the evolution of the Bloch vector $\vv$ and its derivative $\dot\vv$. After a sufficiently long evolution, the Y-component of $\vv$ is $\Theta(\frac{1}{\sqrt{n}})$, and each step the Z-component of $\dot\vv$ is increased by $\Theta(\frac{1}{\sqrt{n}})$, robust under dephasing noise. After $n$ steps, the derivative $\dot\vv$ accumulates to $\Theta(\sqrt{n})$. 
When $w \ll p$, the output QFI $\propto \frac{4(1-p)}{p} n$. 
\figref{fig:example} shows the linear scaling emerges at a small $n$ using our unitary control, but without control the QFI approaches only a constant for large $n$. 

Our unitary protocol applies measurement only in the end. If instead, we are allowed to measure the state at any intermediate steps, the SQL is achievable through repeated measurements. The blue curve in \figref{fig:example} shows a repeated measurement protocol without control where the state is measured using the optimal basis and then reset to $\ket{0}$ every $6$ steps ($6$ is the optimal choice here). The QFI is equal to the number of intervals multiplied by the QFI of each interval. Remarkably, our unitary protocol significantly outperforms the repeated measurement protocol. 
We also plot the Fisher information (FI) of our unitary protocol with noisy state preparation and measurement (SPAM). We assume the input state is $(1-q)\ket{0}\bra{0} + q \ket{1}\bra{1}$, fix the POVM to be $\{M = (1-q)\ket{0}\bra{0} + q \ket{1}\bra{1},\id-M\}$ and take $q = 0.1\%$ and $2\%$. Even with 2\% SPAM noise, our protocol outperforms the noiseless repeated measurement protocol at large $n$. 

In practice, the readout noise is usually the dominant source of noise (compared to gate noise) in quantum devices~\cite{krantz2019quantum,bruzewicz2019trapped,dutt2007quantum,jiang2009repetitive}, and our simulation shows the advantage of our protocol at 2\% readout noise, a value achievable using current quantum devices~\cite{ibm2024}. 
In addition, our unitary control protocol may have the advantage over the repeated measurement protocol in terms of saving the probing time. For example, for the IBMQ device \textsf{ibm\_sherbrooke}~\cite{ibm2024}, the readout time $\sim$1us is significantly larger than the single-qubit gate time $\sim$10ns.

\emph{Discussion.---}
We studied the metrological limits for qubit channels under restricted quantum controls, where a new hierarchy is formed. It would be interesting to consider generalization to qudit quantum channels where the simple classification is no longer available. However, we expect some of our techniques, e.g., the channel extension method and the contraction coefficient, apply to at least some types of qudit channels, and might be useful for e.g., exploring the achievability of the HL~\cite{liu2024heisenberg}. Another generalization is to include ancilla-free measurement and feed-forward controls in our setting--it is unknown whether the HL is achievable for some types of noisy qubit channels in this case. Finally, it is of practical interest to study how the QFI varies as the size of the ancilla changes~\cite{kurdzialek2024quantum,liu2024efficient}. 

Finally, it is also worth noting the potential connections of our results to quantum channel learning theory (with and without ancilla~\cite{chen2022quantum,chen2023tight,chen2023efficient} or memory~\cite{aharonov2022quantum,huang2022quantum,chen2022exponential}) and to recent no-go results on quantum computation without QEC~\cite{chen2023complexity,tsubouchi2023universal,takagi2023universal,quek2022exponentially}.


\emph{Acknowledgments.---}
S.Z. thanks Senrui Chen for helpful discussions. S.Z. acknowledges funding provided by Perimeter Institute for Theoretical Physics, a research institute supported in part by the Government of Canada through the Department of Innovation, Science and Economic Development Canada and by the Province of Ontario through the Ministry of Colleges and Universities, and the Institute for Quantum Information and Matter, an NSF Physics Frontiers Center (NSF Grant No. PHY2317110).

\bibliography{refs-small}

\begin{thebibliography}{90}%
\makeatletter
\providecommand \@ifxundefined [1]{%
 \@ifx{#1\undefined}
}%
\providecommand \@ifnum [1]{%
 \ifnum #1\expandafter \@firstoftwo
 \else \expandafter \@secondoftwo
 \fi
}%
\providecommand \@ifx [1]{%
 \ifx #1\expandafter \@firstoftwo
 \else \expandafter \@secondoftwo
 \fi
}%
\providecommand \natexlab [1]{#1}%
\providecommand \enquote  [1]{``#1''}%
\providecommand \bibnamefont  [1]{#1}%
\providecommand \bibfnamefont [1]{#1}%
\providecommand \citenamefont [1]{#1}%
\providecommand \href@noop [0]{\@secondoftwo}%
\providecommand \href [0]{\begingroup \@sanitize@url \@href}%
\providecommand \@href[1]{\@@startlink{#1}\@@href}%
\providecommand \@@href[1]{\endgroup#1\@@endlink}%
\providecommand \@sanitize@url [0]{\catcode `\\12\catcode `\$12\catcode
  `\&12\catcode `\#12\catcode `\^12\catcode `\_12\catcode `\%12\relax}%
\providecommand \@@startlink[1]{}%
\providecommand \@@endlink[0]{}%
\providecommand \url  [0]{\begingroup\@sanitize@url \@url }%
\providecommand \@url [1]{\endgroup\@href {#1}{\urlprefix }}%
\providecommand \urlprefix  [0]{URL }%
\providecommand \Eprint [0]{\href }%
\providecommand \doibase [0]{https://doi.org/}%
\providecommand \selectlanguage [0]{\@gobble}%
\providecommand \bibinfo  [0]{\@secondoftwo}%
\providecommand \bibfield  [0]{\@secondoftwo}%
\providecommand \translation [1]{[#1]}%
\providecommand \BibitemOpen [0]{}%
\providecommand \bibitemStop [0]{}%
\providecommand \bibitemNoStop [0]{.\EOS\space}%
\providecommand \EOS [0]{\spacefactor3000\relax}%
\providecommand \BibitemShut  [1]{\csname bibitem#1\endcsname}%
\let\auto@bib@innerbib\@empty
\bibitem [{\citenamefont {Giovannetti}\ \emph {et~al.}(2011)\citenamefont
  {Giovannetti}, \citenamefont {Lloyd},\ and\ \citenamefont
  {Maccone}}]{giovannetti2011advances}%
  \BibitemOpen
  \bibfield  {author} {\bibinfo {author} {\bibfnamefont {V.}~\bibnamefont
  {Giovannetti}}, \bibinfo {author} {\bibfnamefont {S.}~\bibnamefont {Lloyd}},\
  and\ \bibinfo {author} {\bibfnamefont {L.}~\bibnamefont {Maccone}},\
  }\bibfield  {title} {\bibinfo {title} {Advances in quantum metrology},\
  }\href@noop {} {\bibfield  {journal} {\bibinfo  {journal} {Nat. Photonics.}\
  }\textbf {\bibinfo {volume} {5}},\ \bibinfo {pages} {222} (\bibinfo {year}
  {2011})}\BibitemShut {NoStop}%
\bibitem [{\citenamefont {Degen}\ \emph {et~al.}(2017)\citenamefont {Degen},
  \citenamefont {Reinhard},\ and\ \citenamefont
  {Cappellaro}}]{degen2017quantum}%
  \BibitemOpen
  \bibfield  {author} {\bibinfo {author} {\bibfnamefont {C.~L.}\ \bibnamefont
  {Degen}}, \bibinfo {author} {\bibfnamefont {F.}~\bibnamefont {Reinhard}},\
  and\ \bibinfo {author} {\bibfnamefont {P.}~\bibnamefont {Cappellaro}},\
  }\bibfield  {title} {\bibinfo {title} {Quantum sensing},\ }\href@noop {}
  {\bibfield  {journal} {\bibinfo  {journal} {Rev. Mod. Phys.}\ }\textbf
  {\bibinfo {volume} {89}},\ \bibinfo {pages} {035002} (\bibinfo {year}
  {2017})}\BibitemShut {NoStop}%
\bibitem [{\citenamefont {Pezz\`e}\ \emph {et~al.}(2018)\citenamefont
  {Pezz\`e}, \citenamefont {Smerzi}, \citenamefont {Oberthaler}, \citenamefont
  {Schmied},\ and\ \citenamefont {Treutlein}}]{pezze2018quantum}%
  \BibitemOpen
  \bibfield  {author} {\bibinfo {author} {\bibfnamefont {L.}~\bibnamefont
  {Pezz\`e}}, \bibinfo {author} {\bibfnamefont {A.}~\bibnamefont {Smerzi}},
  \bibinfo {author} {\bibfnamefont {M.~K.}\ \bibnamefont {Oberthaler}},
  \bibinfo {author} {\bibfnamefont {R.}~\bibnamefont {Schmied}},\ and\ \bibinfo
  {author} {\bibfnamefont {P.}~\bibnamefont {Treutlein}},\ }\bibfield  {title}
  {\bibinfo {title} {Quantum metrology with nonclassical states of atomic
  ensembles},\ }\href@noop {} {\bibfield  {journal} {\bibinfo  {journal} {Rev.
  Mod. Phys.}\ }\textbf {\bibinfo {volume} {90}},\ \bibinfo {pages} {035005}
  (\bibinfo {year} {2018})}\BibitemShut {NoStop}%
\bibitem [{\citenamefont {Pirandola}\ \emph {et~al.}(2018)\citenamefont
  {Pirandola}, \citenamefont {Bardhan}, \citenamefont {Gehring}, \citenamefont
  {Weedbrook},\ and\ \citenamefont {Lloyd}}]{pirandola2018advances}%
  \BibitemOpen
  \bibfield  {author} {\bibinfo {author} {\bibfnamefont {S.}~\bibnamefont
  {Pirandola}}, \bibinfo {author} {\bibfnamefont {B.~R.}\ \bibnamefont
  {Bardhan}}, \bibinfo {author} {\bibfnamefont {T.}~\bibnamefont {Gehring}},
  \bibinfo {author} {\bibfnamefont {C.}~\bibnamefont {Weedbrook}},\ and\
  \bibinfo {author} {\bibfnamefont {S.}~\bibnamefont {Lloyd}},\ }\bibfield
  {title} {\bibinfo {title} {Advances in photonic quantum sensing},\
  }\href@noop {} {\bibfield  {journal} {\bibinfo  {journal} {Nat. Photonics.}\
  }\textbf {\bibinfo {volume} {12}},\ \bibinfo {pages} {724} (\bibinfo {year}
  {2018})}\BibitemShut {NoStop}%
\bibitem [{\citenamefont {Caves}(1981)}]{caves1981quantum}%
  \BibitemOpen
  \bibfield  {author} {\bibinfo {author} {\bibfnamefont {C.~M.}\ \bibnamefont
  {Caves}},\ }\bibfield  {title} {\bibinfo {title} {Quantum-mechanical noise in
  an interferometer},\ }\href@noop {} {\bibfield  {journal} {\bibinfo
  {journal} {Phys. Rev. D}\ }\textbf {\bibinfo {volume} {23}},\ \bibinfo
  {pages} {1693} (\bibinfo {year} {1981})}\BibitemShut {NoStop}%
\bibitem [{\citenamefont {Yurke}\ \emph {et~al.}(1986)\citenamefont {Yurke},
  \citenamefont {McCall},\ and\ \citenamefont {Klauder}}]{yurke19862}%
  \BibitemOpen
  \bibfield  {author} {\bibinfo {author} {\bibfnamefont {B.}~\bibnamefont
  {Yurke}}, \bibinfo {author} {\bibfnamefont {S.~L.}\ \bibnamefont {McCall}},\
  and\ \bibinfo {author} {\bibfnamefont {J.~R.}\ \bibnamefont {Klauder}},\
  }\bibfield  {title} {\bibinfo {title} {Su(2) and su(1,1) interferometers},\
  }\href@noop {} {\bibfield  {journal} {\bibinfo  {journal} {Phys. Rev. A}\
  }\textbf {\bibinfo {volume} {33}},\ \bibinfo {pages} {4033} (\bibinfo {year}
  {1986})}\BibitemShut {NoStop}%
\bibitem [{\citenamefont {{LIGO Collaboration}}(2011)}]{ligo2011gravitational}%
  \BibitemOpen
  \bibfield  {author} {\bibinfo {author} {\bibnamefont {{LIGO
  Collaboration}}},\ }\bibfield  {title} {\bibinfo {title} {A gravitational
  wave observatory operating beyond the quantum shot-noise limit},\ }\href@noop
  {} {\bibfield  {journal} {\bibinfo  {journal} {Nat. Phys.}\ }\textbf
  {\bibinfo {volume} {7}},\ \bibinfo {pages} {962} (\bibinfo {year}
  {2011})}\BibitemShut {NoStop}%
\bibitem [{\citenamefont {{LIGO Collaboration}}(2013)}]{ligo2013enhanced}%
  \BibitemOpen
  \bibfield  {author} {\bibinfo {author} {\bibnamefont {{LIGO
  Collaboration}}},\ }\bibfield  {title} {\bibinfo {title} {Enhanced
  sensitivity of the {LIGO} gravitational wave detector by using squeezed
  states of light},\ }\href@noop {} {\bibfield  {journal} {\bibinfo  {journal}
  {Nat. Photonics.}\ }\textbf {\bibinfo {volume} {7}},\ \bibinfo {pages} {613}
  (\bibinfo {year} {2013})}\BibitemShut {NoStop}%
\bibitem [{\citenamefont {Le~Sage}\ \emph {et~al.}(2013)\citenamefont
  {Le~Sage}, \citenamefont {Arai}, \citenamefont {Glenn}, \citenamefont
  {DeVience}, \citenamefont {Pham}, \citenamefont {Rahn-Lee}, \citenamefont
  {Lukin}, \citenamefont {Yacoby}, \citenamefont {Komeili},\ and\ \citenamefont
  {Walsworth}}]{le2013optical}%
  \BibitemOpen
  \bibfield  {author} {\bibinfo {author} {\bibfnamefont {D.}~\bibnamefont
  {Le~Sage}}, \bibinfo {author} {\bibfnamefont {K.}~\bibnamefont {Arai}},
  \bibinfo {author} {\bibfnamefont {D.~R.}\ \bibnamefont {Glenn}}, \bibinfo
  {author} {\bibfnamefont {S.~J.}\ \bibnamefont {DeVience}}, \bibinfo {author}
  {\bibfnamefont {L.~M.}\ \bibnamefont {Pham}}, \bibinfo {author}
  {\bibfnamefont {L.}~\bibnamefont {Rahn-Lee}}, \bibinfo {author}
  {\bibfnamefont {M.~D.}\ \bibnamefont {Lukin}}, \bibinfo {author}
  {\bibfnamefont {A.}~\bibnamefont {Yacoby}}, \bibinfo {author} {\bibfnamefont
  {A.}~\bibnamefont {Komeili}},\ and\ \bibinfo {author} {\bibfnamefont {R.~L.}\
  \bibnamefont {Walsworth}},\ }\bibfield  {title} {\bibinfo {title} {Optical
  magnetic imaging of living cells},\ }\href@noop {} {\bibfield  {journal}
  {\bibinfo  {journal} {Nature}\ }\textbf {\bibinfo {volume} {496}},\ \bibinfo
  {pages} {486} (\bibinfo {year} {2013})}\BibitemShut {NoStop}%
\bibitem [{\citenamefont {Lemos}\ \emph {et~al.}(2014)\citenamefont {Lemos},
  \citenamefont {Borish}, \citenamefont {Cole}, \citenamefont {Ramelow},
  \citenamefont {Lapkiewicz},\ and\ \citenamefont
  {Zeilinger}}]{lemos2014quantum}%
  \BibitemOpen
  \bibfield  {author} {\bibinfo {author} {\bibfnamefont {G.~B.}\ \bibnamefont
  {Lemos}}, \bibinfo {author} {\bibfnamefont {V.}~\bibnamefont {Borish}},
  \bibinfo {author} {\bibfnamefont {G.~D.}\ \bibnamefont {Cole}}, \bibinfo
  {author} {\bibfnamefont {S.}~\bibnamefont {Ramelow}}, \bibinfo {author}
  {\bibfnamefont {R.}~\bibnamefont {Lapkiewicz}},\ and\ \bibinfo {author}
  {\bibfnamefont {A.}~\bibnamefont {Zeilinger}},\ }\bibfield  {title} {\bibinfo
  {title} {Quantum imaging with undetected photons},\ }\href@noop {} {\bibfield
   {journal} {\bibinfo  {journal} {Nature}\ }\textbf {\bibinfo {volume}
  {512}},\ \bibinfo {pages} {409} (\bibinfo {year} {2014})}\BibitemShut
  {NoStop}%
\bibitem [{\citenamefont {Tsang}\ \emph {et~al.}(2016)\citenamefont {Tsang},
  \citenamefont {Nair},\ and\ \citenamefont {Lu}}]{tsang2016quantum}%
  \BibitemOpen
  \bibfield  {author} {\bibinfo {author} {\bibfnamefont {M.}~\bibnamefont
  {Tsang}}, \bibinfo {author} {\bibfnamefont {R.}~\bibnamefont {Nair}},\ and\
  \bibinfo {author} {\bibfnamefont {X.-M.}\ \bibnamefont {Lu}},\ }\bibfield
  {title} {\bibinfo {title} {Quantum theory of superresolution for two
  incoherent optical point sources},\ }\href@noop {} {\bibfield  {journal}
  {\bibinfo  {journal} {Phys. Rev. X}\ }\textbf {\bibinfo {volume} {6}},\
  \bibinfo {pages} {031033} (\bibinfo {year} {2016})}\BibitemShut {NoStop}%
\bibitem [{\citenamefont {Abobeih}\ \emph {et~al.}(2019)\citenamefont
  {Abobeih}, \citenamefont {Randall}, \citenamefont {Bradley}, \citenamefont
  {Bartling}, \citenamefont {Bakker}, \citenamefont {Degen}, \citenamefont
  {Markham}, \citenamefont {Twitchen},\ and\ \citenamefont
  {Taminiau}}]{abobeih2019atomic}%
  \BibitemOpen
  \bibfield  {author} {\bibinfo {author} {\bibfnamefont {M.}~\bibnamefont
  {Abobeih}}, \bibinfo {author} {\bibfnamefont {J.}~\bibnamefont {Randall}},
  \bibinfo {author} {\bibfnamefont {C.}~\bibnamefont {Bradley}}, \bibinfo
  {author} {\bibfnamefont {H.}~\bibnamefont {Bartling}}, \bibinfo {author}
  {\bibfnamefont {M.}~\bibnamefont {Bakker}}, \bibinfo {author} {\bibfnamefont
  {M.}~\bibnamefont {Degen}}, \bibinfo {author} {\bibfnamefont
  {M.}~\bibnamefont {Markham}}, \bibinfo {author} {\bibfnamefont
  {D.}~\bibnamefont {Twitchen}},\ and\ \bibinfo {author} {\bibfnamefont
  {T.}~\bibnamefont {Taminiau}},\ }\bibfield  {title} {\bibinfo {title}
  {Atomic-scale imaging of a 27-nuclear-spin cluster using a quantum sensor},\
  }\href@noop {} {\bibfield  {journal} {\bibinfo  {journal} {Nature}\ }\textbf
  {\bibinfo {volume} {576}},\ \bibinfo {pages} {411} (\bibinfo {year}
  {2019})}\BibitemShut {NoStop}%
\bibitem [{\citenamefont {Wineland}\ \emph {et~al.}(1992)\citenamefont
  {Wineland}, \citenamefont {Bollinger}, \citenamefont {Itano}, \citenamefont
  {Moore},\ and\ \citenamefont {Heinzen}}]{wineland1992spin}%
  \BibitemOpen
  \bibfield  {author} {\bibinfo {author} {\bibfnamefont {D.~J.}\ \bibnamefont
  {Wineland}}, \bibinfo {author} {\bibfnamefont {J.~J.}\ \bibnamefont
  {Bollinger}}, \bibinfo {author} {\bibfnamefont {W.~M.}\ \bibnamefont
  {Itano}}, \bibinfo {author} {\bibfnamefont {F.~L.}\ \bibnamefont {Moore}},\
  and\ \bibinfo {author} {\bibfnamefont {D.~J.}\ \bibnamefont {Heinzen}},\
  }\bibfield  {title} {\bibinfo {title} {Spin squeezing and reduced quantum
  noise in spectroscopy},\ }\href@noop {} {\bibfield  {journal} {\bibinfo
  {journal} {Phys. Rev. A}\ }\textbf {\bibinfo {volume} {46}},\ \bibinfo
  {pages} {R6797} (\bibinfo {year} {1992})}\BibitemShut {NoStop}%
\bibitem [{\citenamefont {Bollinger}\ \emph {et~al.}(1996)\citenamefont
  {Bollinger}, \citenamefont {Itano}, \citenamefont {Wineland},\ and\
  \citenamefont {Heinzen}}]{bollinger1996optimal}%
  \BibitemOpen
  \bibfield  {author} {\bibinfo {author} {\bibfnamefont {J.~J.}\ \bibnamefont
  {Bollinger}}, \bibinfo {author} {\bibfnamefont {W.~M.}\ \bibnamefont
  {Itano}}, \bibinfo {author} {\bibfnamefont {D.~J.}\ \bibnamefont
  {Wineland}},\ and\ \bibinfo {author} {\bibfnamefont {D.~J.}\ \bibnamefont
  {Heinzen}},\ }\bibfield  {title} {\bibinfo {title} {Optimal frequency
  measurements with maximally correlated states},\ }\href@noop {} {\bibfield
  {journal} {\bibinfo  {journal} {Phys. Rev. A}\ }\textbf {\bibinfo {volume}
  {54}},\ \bibinfo {pages} {R4649} (\bibinfo {year} {1996})}\BibitemShut
  {NoStop}%
\bibitem [{\citenamefont {Leibfried}\ \emph {et~al.}(2004)\citenamefont
  {Leibfried}, \citenamefont {Barrett}, \citenamefont {Schaetz}, \citenamefont
  {Britton}, \citenamefont {Chiaverini}, \citenamefont {Itano}, \citenamefont
  {Jost}, \citenamefont {Langer},\ and\ \citenamefont
  {Wineland}}]{leibfried2004toward}%
  \BibitemOpen
  \bibfield  {author} {\bibinfo {author} {\bibfnamefont {D.}~\bibnamefont
  {Leibfried}}, \bibinfo {author} {\bibfnamefont {M.~D.}\ \bibnamefont
  {Barrett}}, \bibinfo {author} {\bibfnamefont {T.}~\bibnamefont {Schaetz}},
  \bibinfo {author} {\bibfnamefont {J.}~\bibnamefont {Britton}}, \bibinfo
  {author} {\bibfnamefont {J.}~\bibnamefont {Chiaverini}}, \bibinfo {author}
  {\bibfnamefont {W.~M.}\ \bibnamefont {Itano}}, \bibinfo {author}
  {\bibfnamefont {J.~D.}\ \bibnamefont {Jost}}, \bibinfo {author}
  {\bibfnamefont {C.}~\bibnamefont {Langer}},\ and\ \bibinfo {author}
  {\bibfnamefont {D.~J.}\ \bibnamefont {Wineland}},\ }\bibfield  {title}
  {\bibinfo {title} {Toward {Heisenberg-limited} spectroscopy with
  multiparticle entangled states},\ }\href@noop {} {\bibfield  {journal}
  {\bibinfo  {journal} {Science}\ }\textbf {\bibinfo {volume} {304}},\ \bibinfo
  {pages} {1476} (\bibinfo {year} {2004})}\BibitemShut {NoStop}%
\bibitem [{\citenamefont {Taylor}\ \emph {et~al.}(2008)\citenamefont {Taylor},
  \citenamefont {Cappellaro}, \citenamefont {Childress}, \citenamefont {Jiang},
  \citenamefont {Budker}, \citenamefont {Hemmer}, \citenamefont {Yacoby},
  \citenamefont {Walsworth},\ and\ \citenamefont {Lukin}}]{taylor2008high}%
  \BibitemOpen
  \bibfield  {author} {\bibinfo {author} {\bibfnamefont {J.}~\bibnamefont
  {Taylor}}, \bibinfo {author} {\bibfnamefont {P.}~\bibnamefont {Cappellaro}},
  \bibinfo {author} {\bibfnamefont {L.}~\bibnamefont {Childress}}, \bibinfo
  {author} {\bibfnamefont {L.}~\bibnamefont {Jiang}}, \bibinfo {author}
  {\bibfnamefont {D.}~\bibnamefont {Budker}}, \bibinfo {author} {\bibfnamefont
  {P.}~\bibnamefont {Hemmer}}, \bibinfo {author} {\bibfnamefont
  {A.}~\bibnamefont {Yacoby}}, \bibinfo {author} {\bibfnamefont
  {R.}~\bibnamefont {Walsworth}},\ and\ \bibinfo {author} {\bibfnamefont
  {M.}~\bibnamefont {Lukin}},\ }\bibfield  {title} {\bibinfo {title}
  {High-sensitivity diamond magnetometer with nanoscale resolution},\
  }\href@noop {} {\bibfield  {journal} {\bibinfo  {journal} {Nat. Phys.}\
  }\textbf {\bibinfo {volume} {4}},\ \bibinfo {pages} {810} (\bibinfo {year}
  {2008})}\BibitemShut {NoStop}%
\bibitem [{\citenamefont {Zhou}\ \emph {et~al.}(2020)\citenamefont {Zhou},
  \citenamefont {Choi}, \citenamefont {Choi}, \citenamefont {Landig},
  \citenamefont {Douglas}, \citenamefont {Isoya}, \citenamefont {Jelezko},
  \citenamefont {Onoda}, \citenamefont {Sumiya}, \citenamefont {Cappellaro}
  \emph {et~al.}}]{zhou2020quantum}%
  \BibitemOpen
  \bibfield  {author} {\bibinfo {author} {\bibfnamefont {H.}~\bibnamefont
  {Zhou}}, \bibinfo {author} {\bibfnamefont {J.}~\bibnamefont {Choi}}, \bibinfo
  {author} {\bibfnamefont {S.}~\bibnamefont {Choi}}, \bibinfo {author}
  {\bibfnamefont {R.}~\bibnamefont {Landig}}, \bibinfo {author} {\bibfnamefont
  {A.~M.}\ \bibnamefont {Douglas}}, \bibinfo {author} {\bibfnamefont
  {J.}~\bibnamefont {Isoya}}, \bibinfo {author} {\bibfnamefont
  {F.}~\bibnamefont {Jelezko}}, \bibinfo {author} {\bibfnamefont
  {S.}~\bibnamefont {Onoda}}, \bibinfo {author} {\bibfnamefont
  {H.}~\bibnamefont {Sumiya}}, \bibinfo {author} {\bibfnamefont
  {P.}~\bibnamefont {Cappellaro}}, \emph {et~al.},\ }\bibfield  {title}
  {\bibinfo {title} {Quantum metrology with strongly interacting spin
  systems},\ }\href@noop {} {\bibfield  {journal} {\bibinfo  {journal} {Phys.
  Rev. X}\ }\textbf {\bibinfo {volume} {10}},\ \bibinfo {pages} {031003}
  (\bibinfo {year} {2020})}\BibitemShut {NoStop}%
\bibitem [{\citenamefont {Rosenband}\ \emph {et~al.}(2008)\citenamefont
  {Rosenband}, \citenamefont {Hume}, \citenamefont {Schmidt}, \citenamefont
  {Chou}, \citenamefont {Brusch}, \citenamefont {Lorini}, \citenamefont
  {Oskay}, \citenamefont {Drullinger}, \citenamefont {Fortier}, \citenamefont
  {Stalnaker} \emph {et~al.}}]{rosenband2008frequency}%
  \BibitemOpen
  \bibfield  {author} {\bibinfo {author} {\bibfnamefont {T.}~\bibnamefont
  {Rosenband}}, \bibinfo {author} {\bibfnamefont {D.}~\bibnamefont {Hume}},
  \bibinfo {author} {\bibfnamefont {P.}~\bibnamefont {Schmidt}}, \bibinfo
  {author} {\bibfnamefont {C.-W.}\ \bibnamefont {Chou}}, \bibinfo {author}
  {\bibfnamefont {A.}~\bibnamefont {Brusch}}, \bibinfo {author} {\bibfnamefont
  {L.}~\bibnamefont {Lorini}}, \bibinfo {author} {\bibfnamefont
  {W.}~\bibnamefont {Oskay}}, \bibinfo {author} {\bibfnamefont {R.~E.}\
  \bibnamefont {Drullinger}}, \bibinfo {author} {\bibfnamefont {T.~M.}\
  \bibnamefont {Fortier}}, \bibinfo {author} {\bibfnamefont {J.~E.}\
  \bibnamefont {Stalnaker}}, \emph {et~al.},\ }\bibfield  {title} {\bibinfo
  {title} {Frequency ratio of {Al}+ and {Hg}+ single-ion optical clocks;
  metrology at the 17th decimal place},\ }\href@noop {} {\bibfield  {journal}
  {\bibinfo  {journal} {Science}\ }\textbf {\bibinfo {volume} {319}},\ \bibinfo
  {pages} {1808} (\bibinfo {year} {2008})}\BibitemShut {NoStop}%
\bibitem [{\citenamefont {Appel}\ \emph {et~al.}(2009)\citenamefont {Appel},
  \citenamefont {Windpassinger}, \citenamefont {Oblak}, \citenamefont {Hoff},
  \citenamefont {Kj{\ae}rgaard},\ and\ \citenamefont
  {Polzik}}]{appel2009mesoscopic}%
  \BibitemOpen
  \bibfield  {author} {\bibinfo {author} {\bibfnamefont {J.}~\bibnamefont
  {Appel}}, \bibinfo {author} {\bibfnamefont {P.~J.}\ \bibnamefont
  {Windpassinger}}, \bibinfo {author} {\bibfnamefont {D.}~\bibnamefont
  {Oblak}}, \bibinfo {author} {\bibfnamefont {U.~B.}\ \bibnamefont {Hoff}},
  \bibinfo {author} {\bibfnamefont {N.}~\bibnamefont {Kj{\ae}rgaard}},\ and\
  \bibinfo {author} {\bibfnamefont {E.~S.}\ \bibnamefont {Polzik}},\ }\bibfield
   {title} {\bibinfo {title} {Mesoscopic atomic entanglement for precision
  measurements beyond the standard quantum limit},\ }\href@noop {} {\bibfield
  {journal} {\bibinfo  {journal} {Proc. Natl. Acad. Sci.}\ }\textbf {\bibinfo
  {volume} {106}},\ \bibinfo {pages} {10960} (\bibinfo {year}
  {2009})}\BibitemShut {NoStop}%
\bibitem [{\citenamefont {Ludlow}\ \emph {et~al.}(2015)\citenamefont {Ludlow},
  \citenamefont {Boyd}, \citenamefont {Ye}, \citenamefont {Peik},\ and\
  \citenamefont {Schmidt}}]{ludlow2015optical}%
  \BibitemOpen
  \bibfield  {author} {\bibinfo {author} {\bibfnamefont {A.~D.}\ \bibnamefont
  {Ludlow}}, \bibinfo {author} {\bibfnamefont {M.~M.}\ \bibnamefont {Boyd}},
  \bibinfo {author} {\bibfnamefont {J.}~\bibnamefont {Ye}}, \bibinfo {author}
  {\bibfnamefont {E.}~\bibnamefont {Peik}},\ and\ \bibinfo {author}
  {\bibfnamefont {P.~O.}\ \bibnamefont {Schmidt}},\ }\bibfield  {title}
  {\bibinfo {title} {Optical atomic clocks},\ }\href@noop {} {\bibfield
  {journal} {\bibinfo  {journal} {Rev. Mod. Phys.}\ }\textbf {\bibinfo {volume}
  {87}},\ \bibinfo {pages} {637} (\bibinfo {year} {2015})}\BibitemShut
  {NoStop}%
\bibitem [{\citenamefont {Kaubruegger}\ \emph {et~al.}(2021)\citenamefont
  {Kaubruegger}, \citenamefont {Vasilyev}, \citenamefont {Schulte},
  \citenamefont {Hammerer},\ and\ \citenamefont
  {Zoller}}]{kaubruegger2021quantum}%
  \BibitemOpen
  \bibfield  {author} {\bibinfo {author} {\bibfnamefont {R.}~\bibnamefont
  {Kaubruegger}}, \bibinfo {author} {\bibfnamefont {D.~V.}\ \bibnamefont
  {Vasilyev}}, \bibinfo {author} {\bibfnamefont {M.}~\bibnamefont {Schulte}},
  \bibinfo {author} {\bibfnamefont {K.}~\bibnamefont {Hammerer}},\ and\
  \bibinfo {author} {\bibfnamefont {P.}~\bibnamefont {Zoller}},\ }\bibfield
  {title} {\bibinfo {title} {Quantum variational optimization of {Ramsey}
  interferometry and atomic clocks},\ }\href@noop {} {\bibfield  {journal}
  {\bibinfo  {journal} {Phys. Rev. X}\ }\textbf {\bibinfo {volume} {11}},\
  \bibinfo {pages} {041045} (\bibinfo {year} {2021})}\BibitemShut {NoStop}%
\bibitem [{\citenamefont {Marciniak}\ \emph {et~al.}(2022)\citenamefont
  {Marciniak}, \citenamefont {Feldker}, \citenamefont {Pogorelov},
  \citenamefont {Kaubruegger}, \citenamefont {Vasilyev}, \citenamefont {van
  Bijnen}, \citenamefont {Schindler}, \citenamefont {Zoller}, \citenamefont
  {Blatt},\ and\ \citenamefont {Monz}}]{marciniak2022optimal}%
  \BibitemOpen
  \bibfield  {author} {\bibinfo {author} {\bibfnamefont {C.~D.}\ \bibnamefont
  {Marciniak}}, \bibinfo {author} {\bibfnamefont {T.}~\bibnamefont {Feldker}},
  \bibinfo {author} {\bibfnamefont {I.}~\bibnamefont {Pogorelov}}, \bibinfo
  {author} {\bibfnamefont {R.}~\bibnamefont {Kaubruegger}}, \bibinfo {author}
  {\bibfnamefont {D.~V.}\ \bibnamefont {Vasilyev}}, \bibinfo {author}
  {\bibfnamefont {R.}~\bibnamefont {van Bijnen}}, \bibinfo {author}
  {\bibfnamefont {P.}~\bibnamefont {Schindler}}, \bibinfo {author}
  {\bibfnamefont {P.}~\bibnamefont {Zoller}}, \bibinfo {author} {\bibfnamefont
  {R.}~\bibnamefont {Blatt}},\ and\ \bibinfo {author} {\bibfnamefont
  {T.}~\bibnamefont {Monz}},\ }\bibfield  {title} {\bibinfo {title} {Optimal
  metrology with programmable quantum sensors},\ }\href@noop {} {\bibfield
  {journal} {\bibinfo  {journal} {Nature}\ }\textbf {\bibinfo {volume} {603}},\
  \bibinfo {pages} {604} (\bibinfo {year} {2022})}\BibitemShut {NoStop}%
\bibitem [{\citenamefont {Giovannetti}\ \emph {et~al.}(2006)\citenamefont
  {Giovannetti}, \citenamefont {Lloyd},\ and\ \citenamefont
  {Maccone}}]{giovannetti2006quantum}%
  \BibitemOpen
  \bibfield  {author} {\bibinfo {author} {\bibfnamefont {V.}~\bibnamefont
  {Giovannetti}}, \bibinfo {author} {\bibfnamefont {S.}~\bibnamefont {Lloyd}},\
  and\ \bibinfo {author} {\bibfnamefont {L.}~\bibnamefont {Maccone}},\
  }\bibfield  {title} {\bibinfo {title} {Quantum metrology},\ }\href@noop {}
  {\bibfield  {journal} {\bibinfo  {journal} {Phys. Rev. Lett.}\ }\textbf
  {\bibinfo {volume} {96}},\ \bibinfo {pages} {010401} (\bibinfo {year}
  {2006})}\BibitemShut {NoStop}%
\bibitem [{\citenamefont {Giovannetti}\ \emph {et~al.}(2004)\citenamefont
  {Giovannetti}, \citenamefont {Lloyd},\ and\ \citenamefont
  {Maccone}}]{giovannetti2004quantum}%
  \BibitemOpen
  \bibfield  {author} {\bibinfo {author} {\bibfnamefont {V.}~\bibnamefont
  {Giovannetti}}, \bibinfo {author} {\bibfnamefont {S.}~\bibnamefont {Lloyd}},\
  and\ \bibinfo {author} {\bibfnamefont {L.}~\bibnamefont {Maccone}},\
  }\bibfield  {title} {\bibinfo {title} {Quantum-enhanced measurements: Beating
  the standard quantum limit},\ }\href@noop {} {\bibfield  {journal} {\bibinfo
  {journal} {Science}\ }\textbf {\bibinfo {volume} {306}},\ \bibinfo {pages}
  {1330} (\bibinfo {year} {2004})}\BibitemShut {NoStop}%
\bibitem [{\citenamefont {Kitagawa}\ and\ \citenamefont
  {Ueda}(1993)}]{kitagawa1993squeezed}%
  \BibitemOpen
  \bibfield  {author} {\bibinfo {author} {\bibfnamefont {M.}~\bibnamefont
  {Kitagawa}}\ and\ \bibinfo {author} {\bibfnamefont {M.}~\bibnamefont
  {Ueda}},\ }\bibfield  {title} {\bibinfo {title} {Squeezed spin states},\
  }\href@noop {} {\bibfield  {journal} {\bibinfo  {journal} {Phys. Rev. A}\
  }\textbf {\bibinfo {volume} {47}},\ \bibinfo {pages} {5138} (\bibinfo {year}
  {1993})}\BibitemShut {NoStop}%
\bibitem [{\citenamefont {Ramsey}(1950)}]{ramsey1950molecular}%
  \BibitemOpen
  \bibfield  {author} {\bibinfo {author} {\bibfnamefont {N.~F.}\ \bibnamefont
  {Ramsey}},\ }\bibfield  {title} {\bibinfo {title} {A molecular beam resonance
  method with separated oscillating fields},\ }\href@noop {} {\bibfield
  {journal} {\bibinfo  {journal} {Physical Review}\ }\textbf {\bibinfo {volume}
  {78}},\ \bibinfo {pages} {695} (\bibinfo {year} {1950})}\BibitemShut
  {NoStop}%
\bibitem [{\citenamefont {Huelga}\ \emph {et~al.}(1997)\citenamefont {Huelga},
  \citenamefont {Macchiavello}, \citenamefont {Pellizzari}, \citenamefont
  {Ekert}, \citenamefont {Plenio},\ and\ \citenamefont
  {Cirac}}]{huelga1997improvement}%
  \BibitemOpen
  \bibfield  {author} {\bibinfo {author} {\bibfnamefont {S.~F.}\ \bibnamefont
  {Huelga}}, \bibinfo {author} {\bibfnamefont {C.}~\bibnamefont
  {Macchiavello}}, \bibinfo {author} {\bibfnamefont {T.}~\bibnamefont
  {Pellizzari}}, \bibinfo {author} {\bibfnamefont {A.~K.}\ \bibnamefont
  {Ekert}}, \bibinfo {author} {\bibfnamefont {M.~B.}\ \bibnamefont {Plenio}},\
  and\ \bibinfo {author} {\bibfnamefont {J.~I.}\ \bibnamefont {Cirac}},\
  }\bibfield  {title} {\bibinfo {title} {Improvement of frequency standards
  with quantum entanglement},\ }\href@noop {} {\bibfield  {journal} {\bibinfo
  {journal} {Phys. Rev. Lett.}\ }\textbf {\bibinfo {volume} {79}},\ \bibinfo
  {pages} {3865} (\bibinfo {year} {1997})}\BibitemShut {NoStop}%
\bibitem [{\citenamefont {Ulam-Orgikh}\ and\ \citenamefont
  {Kitagawa}(2001)}]{ulam2001spin}%
  \BibitemOpen
  \bibfield  {author} {\bibinfo {author} {\bibfnamefont {D.}~\bibnamefont
  {Ulam-Orgikh}}\ and\ \bibinfo {author} {\bibfnamefont {M.}~\bibnamefont
  {Kitagawa}},\ }\bibfield  {title} {\bibinfo {title} {Spin squeezing and
  decoherence limit in {Ramsey} spectroscopy},\ }\href@noop {} {\bibfield
  {journal} {\bibinfo  {journal} {Phys. Rev. A}\ }\textbf {\bibinfo {volume}
  {64}},\ \bibinfo {pages} {052106} (\bibinfo {year} {2001})}\BibitemShut
  {NoStop}%
\bibitem [{\citenamefont {Escher}\ \emph {et~al.}(2011)\citenamefont {Escher},
  \citenamefont {de~Matos~Filho},\ and\ \citenamefont
  {Davidovich}}]{escher2011general}%
  \BibitemOpen
  \bibfield  {author} {\bibinfo {author} {\bibfnamefont {B.}~\bibnamefont
  {Escher}}, \bibinfo {author} {\bibfnamefont {R.}~\bibnamefont
  {de~Matos~Filho}},\ and\ \bibinfo {author} {\bibfnamefont {L.}~\bibnamefont
  {Davidovich}},\ }\bibfield  {title} {\bibinfo {title} {General framework for
  estimating the ultimate precision limit in noisy quantum-enhanced
  metrology},\ }\href@noop {} {\bibfield  {journal} {\bibinfo  {journal} {Nat.
  Phys.}\ }\textbf {\bibinfo {volume} {7}},\ \bibinfo {pages} {406} (\bibinfo
  {year} {2011})}\BibitemShut {NoStop}%
\bibitem [{\citenamefont {Demkowicz-Dobrza{\'n}ski}\ \emph
  {et~al.}(2012)\citenamefont {Demkowicz-Dobrza{\'n}ski}, \citenamefont
  {Ko{\l}ody{\'n}ski},\ and\ \citenamefont
  {Gu{\c{t}}{\u{a}}}}]{demkowicz2012elusive}%
  \BibitemOpen
  \bibfield  {author} {\bibinfo {author} {\bibfnamefont {R.}~\bibnamefont
  {Demkowicz-Dobrza{\'n}ski}}, \bibinfo {author} {\bibfnamefont
  {J.}~\bibnamefont {Ko{\l}ody{\'n}ski}},\ and\ \bibinfo {author}
  {\bibfnamefont {M.}~\bibnamefont {Gu{\c{t}}{\u{a}}}},\ }\bibfield  {title}
  {\bibinfo {title} {The elusive {Heisenberg} limit in quantum-enhanced
  metrology},\ }\href@noop {} {\bibfield  {journal} {\bibinfo  {journal} {Nat.
  Commun.}\ }\textbf {\bibinfo {volume} {3}},\ \bibinfo {pages} {1063}
  (\bibinfo {year} {2012})}\BibitemShut {NoStop}%
\bibitem [{\citenamefont {Demkowicz-Dobrza\ifmmode~\acute{n}\else
  \'{n}\fi{}ski}\ and\ \citenamefont {Maccone}(2014)}]{demkowicz2014using}%
  \BibitemOpen
  \bibfield  {author} {\bibinfo {author} {\bibfnamefont {R.}~\bibnamefont
  {Demkowicz-Dobrza\ifmmode~\acute{n}\else \'{n}\fi{}ski}}\ and\ \bibinfo
  {author} {\bibfnamefont {L.}~\bibnamefont {Maccone}},\ }\bibfield  {title}
  {\bibinfo {title} {Using entanglement against noise in quantum metrology},\
  }\href@noop {} {\bibfield  {journal} {\bibinfo  {journal} {Phys. Rev. Lett.}\
  }\textbf {\bibinfo {volume} {113}},\ \bibinfo {pages} {250801} (\bibinfo
  {year} {2014})}\BibitemShut {NoStop}%
\bibitem [{\citenamefont {Sekatski}\ \emph {et~al.}(2017)\citenamefont
  {Sekatski}, \citenamefont {Skotiniotis}, \citenamefont
  {Ko{\l{}}ody{\'{n}}ski},\ and\ \citenamefont
  {D{\"{u}}r}}]{sekatski2017quantum}%
  \BibitemOpen
  \bibfield  {author} {\bibinfo {author} {\bibfnamefont {P.}~\bibnamefont
  {Sekatski}}, \bibinfo {author} {\bibfnamefont {M.}~\bibnamefont
  {Skotiniotis}}, \bibinfo {author} {\bibfnamefont {J.}~\bibnamefont
  {Ko{\l{}}ody{\'{n}}ski}},\ and\ \bibinfo {author} {\bibfnamefont
  {W.}~\bibnamefont {D{\"{u}}r}},\ }\bibfield  {title} {\bibinfo {title}
  {Quantum metrology with full and fast quantum control},\ }\href@noop {}
  {\bibfield  {journal} {\bibinfo  {journal} {{Quantum}}\ }\textbf {\bibinfo
  {volume} {1}},\ \bibinfo {pages} {27} (\bibinfo {year} {2017})}\BibitemShut
  {NoStop}%
\bibitem [{\citenamefont {Demkowicz-Dobrza\ifmmode~\acute{n}\else
  \'{n}\fi{}ski}\ \emph {et~al.}(2017)\citenamefont
  {Demkowicz-Dobrza\ifmmode~\acute{n}\else \'{n}\fi{}ski}, \citenamefont
  {Czajkowski},\ and\ \citenamefont {Sekatski}}]{demkowicz2017adaptive}%
  \BibitemOpen
  \bibfield  {author} {\bibinfo {author} {\bibfnamefont {R.}~\bibnamefont
  {Demkowicz-Dobrza\ifmmode~\acute{n}\else \'{n}\fi{}ski}}, \bibinfo {author}
  {\bibfnamefont {J.}~\bibnamefont {Czajkowski}},\ and\ \bibinfo {author}
  {\bibfnamefont {P.}~\bibnamefont {Sekatski}},\ }\bibfield  {title} {\bibinfo
  {title} {Adaptive quantum metrology under general {Markovian} noise},\
  }\href@noop {} {\bibfield  {journal} {\bibinfo  {journal} {Phys. Rev. X}\
  }\textbf {\bibinfo {volume} {7}},\ \bibinfo {pages} {041009} (\bibinfo {year}
  {2017})}\BibitemShut {NoStop}%
\bibitem [{\citenamefont {Zhou}\ \emph {et~al.}(2018)\citenamefont {Zhou},
  \citenamefont {Zhang}, \citenamefont {Preskill},\ and\ \citenamefont
  {Jiang}}]{zhou2018achieving}%
  \BibitemOpen
  \bibfield  {author} {\bibinfo {author} {\bibfnamefont {S.}~\bibnamefont
  {Zhou}}, \bibinfo {author} {\bibfnamefont {M.}~\bibnamefont {Zhang}},
  \bibinfo {author} {\bibfnamefont {J.}~\bibnamefont {Preskill}},\ and\
  \bibinfo {author} {\bibfnamefont {L.}~\bibnamefont {Jiang}},\ }\bibfield
  {title} {\bibinfo {title} {Achieving the {Heisenberg} limit in quantum
  metrology using quantum error correction},\ }\href@noop {} {\bibfield
  {journal} {\bibinfo  {journal} {Nat. Commun.}\ }\textbf {\bibinfo {volume}
  {9}},\ \bibinfo {pages} {78} (\bibinfo {year} {2018})}\BibitemShut {NoStop}%
\bibitem [{\citenamefont {Zhou}\ and\ \citenamefont
  {Jiang}(2021)}]{zhou2021asymptotic}%
  \BibitemOpen
  \bibfield  {author} {\bibinfo {author} {\bibfnamefont {S.}~\bibnamefont
  {Zhou}}\ and\ \bibinfo {author} {\bibfnamefont {L.}~\bibnamefont {Jiang}},\
  }\bibfield  {title} {\bibinfo {title} {Asymptotic theory of quantum channel
  estimation},\ }\href@noop {} {\bibfield  {journal} {\bibinfo  {journal} {PRX
  Quantum}\ }\textbf {\bibinfo {volume} {2}},\ \bibinfo {pages} {010343}
  (\bibinfo {year} {2021})}\BibitemShut {NoStop}%
\bibitem [{\citenamefont {Kessler}\ \emph {et~al.}(2014)\citenamefont
  {Kessler}, \citenamefont {Lovchinsky}, \citenamefont {Sushkov},\ and\
  \citenamefont {Lukin}}]{kessler2014quantum}%
  \BibitemOpen
  \bibfield  {author} {\bibinfo {author} {\bibfnamefont {E.~M.}\ \bibnamefont
  {Kessler}}, \bibinfo {author} {\bibfnamefont {I.}~\bibnamefont {Lovchinsky}},
  \bibinfo {author} {\bibfnamefont {A.~O.}\ \bibnamefont {Sushkov}},\ and\
  \bibinfo {author} {\bibfnamefont {M.~D.}\ \bibnamefont {Lukin}},\ }\bibfield
  {title} {\bibinfo {title} {Quantum error correction for metrology},\
  }\href@noop {} {\bibfield  {journal} {\bibinfo  {journal} {Phys. Rev. Lett.}\
  }\textbf {\bibinfo {volume} {112}},\ \bibinfo {pages} {150802} (\bibinfo
  {year} {2014})}\BibitemShut {NoStop}%
\bibitem [{\citenamefont {Arrad}\ \emph {et~al.}(2014)\citenamefont {Arrad},
  \citenamefont {Vinkler}, \citenamefont {Aharonov},\ and\ \citenamefont
  {Retzker}}]{arrad2014increasing}%
  \BibitemOpen
  \bibfield  {author} {\bibinfo {author} {\bibfnamefont {G.}~\bibnamefont
  {Arrad}}, \bibinfo {author} {\bibfnamefont {Y.}~\bibnamefont {Vinkler}},
  \bibinfo {author} {\bibfnamefont {D.}~\bibnamefont {Aharonov}},\ and\
  \bibinfo {author} {\bibfnamefont {A.}~\bibnamefont {Retzker}},\ }\bibfield
  {title} {\bibinfo {title} {Increasing sensing resolution with error
  correction},\ }\href@noop {} {\bibfield  {journal} {\bibinfo  {journal}
  {Phys. Rev. Lett.}\ }\textbf {\bibinfo {volume} {112}},\ \bibinfo {pages}
  {150801} (\bibinfo {year} {2014})}\BibitemShut {NoStop}%
\bibitem [{\citenamefont {Unden}\ \emph {et~al.}(2016)\citenamefont {Unden},
  \citenamefont {Balasubramanian}, \citenamefont {Louzon}, \citenamefont
  {Vinkler}, \citenamefont {Plenio}, \citenamefont {Markham}, \citenamefont
  {Twitchen}, \citenamefont {Stacey}, \citenamefont {Lovchinsky}, \citenamefont
  {Sushkov}, \citenamefont {Lukin}, \citenamefont {Retzker}, \citenamefont
  {Naydenov}, \citenamefont {McGuinness},\ and\ \citenamefont
  {Jelezko}}]{unden2016quantum}%
  \BibitemOpen
  \bibfield  {author} {\bibinfo {author} {\bibfnamefont {T.}~\bibnamefont
  {Unden}}, \bibinfo {author} {\bibfnamefont {P.}~\bibnamefont
  {Balasubramanian}}, \bibinfo {author} {\bibfnamefont {D.}~\bibnamefont
  {Louzon}}, \bibinfo {author} {\bibfnamefont {Y.}~\bibnamefont {Vinkler}},
  \bibinfo {author} {\bibfnamefont {M.~B.}\ \bibnamefont {Plenio}}, \bibinfo
  {author} {\bibfnamefont {M.}~\bibnamefont {Markham}}, \bibinfo {author}
  {\bibfnamefont {D.}~\bibnamefont {Twitchen}}, \bibinfo {author}
  {\bibfnamefont {A.}~\bibnamefont {Stacey}}, \bibinfo {author} {\bibfnamefont
  {I.}~\bibnamefont {Lovchinsky}}, \bibinfo {author} {\bibfnamefont {A.~O.}\
  \bibnamefont {Sushkov}}, \bibinfo {author} {\bibfnamefont {M.~D.}\
  \bibnamefont {Lukin}}, \bibinfo {author} {\bibfnamefont {A.}~\bibnamefont
  {Retzker}}, \bibinfo {author} {\bibfnamefont {B.}~\bibnamefont {Naydenov}},
  \bibinfo {author} {\bibfnamefont {L.~P.}\ \bibnamefont {McGuinness}},\ and\
  \bibinfo {author} {\bibfnamefont {F.}~\bibnamefont {Jelezko}},\ }\bibfield
  {title} {\bibinfo {title} {Quantum metrology enhanced by repetitive quantum
  error correction},\ }\href@noop {} {\bibfield  {journal} {\bibinfo  {journal}
  {Phys. Rev. Lett.}\ }\textbf {\bibinfo {volume} {116}},\ \bibinfo {pages}
  {230502} (\bibinfo {year} {2016})}\BibitemShut {NoStop}%
\bibitem [{\citenamefont {D\"ur}\ \emph {et~al.}(2014)\citenamefont {D\"ur},
  \citenamefont {Skotiniotis}, \citenamefont {Fr\"owis},\ and\ \citenamefont
  {Kraus}}]{dur2014improved}%
  \BibitemOpen
  \bibfield  {author} {\bibinfo {author} {\bibfnamefont {W.}~\bibnamefont
  {D\"ur}}, \bibinfo {author} {\bibfnamefont {M.}~\bibnamefont {Skotiniotis}},
  \bibinfo {author} {\bibfnamefont {F.}~\bibnamefont {Fr\"owis}},\ and\
  \bibinfo {author} {\bibfnamefont {B.}~\bibnamefont {Kraus}},\ }\bibfield
  {title} {\bibinfo {title} {Improved quantum metrology using quantum error
  correction},\ }\href@noop {} {\bibfield  {journal} {\bibinfo  {journal}
  {Phys. Rev. Lett.}\ }\textbf {\bibinfo {volume} {112}},\ \bibinfo {pages}
  {080801} (\bibinfo {year} {2014})}\BibitemShut {NoStop}%
\bibitem [{Note1()}]{Note1}%
  \BibitemOpen
  \bibinfo {note} {We will use $\protect \dot {f_\theta }$ to denote the
  derivative of any function $f_\theta $ with respect to $\theta $ at $\theta =
  0$. We will also implicitly assume functions of $\theta $ are taken at its
  true value $\theta = 0$ and sometimes drop the subscript $_\theta $ when
  there is no ambiguity.}\BibitemShut {Stop}%
\bibitem [{\citenamefont {Holevo}(2011)}]{holevo2011probabilistic}%
  \BibitemOpen
  \bibfield  {author} {\bibinfo {author} {\bibfnamefont {A.~S.}\ \bibnamefont
  {Holevo}},\ }\href@noop {} {\emph {\bibinfo {title} {Probabilistic and
  statistical aspects of quantum theory}}},\ Vol.~\bibinfo {volume} {1}\
  (\bibinfo  {publisher} {Springer Science \& Business Media},\ \bibinfo {year}
  {2011})\BibitemShut {NoStop}%
\bibitem [{\citenamefont {Helstrom}(1976)}]{helstrom1976quantum}%
  \BibitemOpen
  \bibfield  {author} {\bibinfo {author} {\bibfnamefont {C.~W.}\ \bibnamefont
  {Helstrom}},\ }\href@noop {} {\emph {\bibinfo {title} {Quantum detection and
  estimation theory}}}\ (\bibinfo  {publisher} {Academic press},\ \bibinfo
  {year} {1976})\BibitemShut {NoStop}%
\bibitem [{\citenamefont {Braunstein}\ and\ \citenamefont
  {Caves}(1994)}]{braunstein1994statistical}%
  \BibitemOpen
  \bibfield  {author} {\bibinfo {author} {\bibfnamefont {S.~L.}\ \bibnamefont
  {Braunstein}}\ and\ \bibinfo {author} {\bibfnamefont {C.~M.}\ \bibnamefont
  {Caves}},\ }\bibfield  {title} {\bibinfo {title} {Statistical distance and
  the geometry of quantum states},\ }\href@noop {} {\bibfield  {journal}
  {\bibinfo  {journal} {Phys. Rev. Lett.}\ }\textbf {\bibinfo {volume} {72}},\
  \bibinfo {pages} {3439} (\bibinfo {year} {1994})}\BibitemShut {NoStop}%
\bibitem [{\citenamefont {Paris}(2009)}]{paris2009quantum}%
  \BibitemOpen
  \bibfield  {author} {\bibinfo {author} {\bibfnamefont {M.~G.}\ \bibnamefont
  {Paris}},\ }\bibfield  {title} {\bibinfo {title} {Quantum estimation for
  quantum technology},\ }\href@noop {} {\bibfield  {journal} {\bibinfo
  {journal} {Int. J. Quantum Inf.}\ }\textbf {\bibinfo {volume} {7}},\ \bibinfo
  {pages} {125} (\bibinfo {year} {2009})}\BibitemShut {NoStop}%
\bibitem [{\citenamefont {Petz}\ and\ \citenamefont
  {Sud{\'a}r}(1996)}]{petz1996geometries}%
  \BibitemOpen
  \bibfield  {author} {\bibinfo {author} {\bibfnamefont {D.}~\bibnamefont
  {Petz}}\ and\ \bibinfo {author} {\bibfnamefont {C.}~\bibnamefont
  {Sud{\'a}r}},\ }\bibfield  {title} {\bibinfo {title} {Geometries of quantum
  states},\ }\href@noop {} {\bibfield  {journal} {\bibinfo  {journal} {J. Math.
  Phys.}\ }\textbf {\bibinfo {volume} {37}},\ \bibinfo {pages} {2662} (\bibinfo
  {year} {1996})}\BibitemShut {NoStop}%
\bibitem [{\citenamefont {Barndorff-Nielsen}\ and\ \citenamefont
  {Gill}(2000)}]{barndorff2000fisher}%
  \BibitemOpen
  \bibfield  {author} {\bibinfo {author} {\bibfnamefont {O.~E.}\ \bibnamefont
  {Barndorff-Nielsen}}\ and\ \bibinfo {author} {\bibfnamefont {R.~D.}\
  \bibnamefont {Gill}},\ }\bibfield  {title} {\bibinfo {title} {Fisher
  information in quantum statistics},\ }\href@noop {} {\bibfield  {journal}
  {\bibinfo  {journal} {J. Phys. A: Math. Gen.}\ }\textbf {\bibinfo {volume}
  {33}},\ \bibinfo {pages} {4481} (\bibinfo {year} {2000})}\BibitemShut
  {NoStop}%
\bibitem [{\citenamefont {Kay}(1993)}]{kay1993fundamentals}%
  \BibitemOpen
  \bibfield  {author} {\bibinfo {author} {\bibfnamefont {S.~M.}\ \bibnamefont
  {Kay}},\ }\href@noop {} {\emph {\bibinfo {title} {Fundamentals of statistical
  signal processing: Volume I Estimation theory}}}\ (\bibinfo  {publisher}
  {Prentice-Hall, Inc.},\ \bibinfo {year} {1993})\BibitemShut {NoStop}%
\bibitem [{\citenamefont {Casella}\ and\ \citenamefont
  {Berger}(2002)}]{casella2002statistical}%
  \BibitemOpen
  \bibfield  {author} {\bibinfo {author} {\bibfnamefont {G.}~\bibnamefont
  {Casella}}\ and\ \bibinfo {author} {\bibfnamefont {R.~L.}\ \bibnamefont
  {Berger}},\ }\href@noop {} {\emph {\bibinfo {title} {Statistical
  inference}}},\ Vol.~\bibinfo {volume} {2}\ (\bibinfo  {publisher} {Duxbury
  Pacific Grove, CA},\ \bibinfo {year} {2002})\BibitemShut {NoStop}%
\bibitem [{\citenamefont {Kobayashi}\ \emph {et~al.}(2011)\citenamefont
  {Kobayashi}, \citenamefont {Mark},\ and\ \citenamefont
  {Turin}}]{kobayashi2011probability}%
  \BibitemOpen
  \bibfield  {author} {\bibinfo {author} {\bibfnamefont {H.}~\bibnamefont
  {Kobayashi}}, \bibinfo {author} {\bibfnamefont {B.~L.}\ \bibnamefont
  {Mark}},\ and\ \bibinfo {author} {\bibfnamefont {W.}~\bibnamefont {Turin}},\
  }\href@noop {} {\emph {\bibinfo {title} {Probability, random processes, and
  statistical analysis: applications to communications, signal processing,
  queueing theory and mathematical finance}}}\ (\bibinfo  {publisher}
  {Cambridge University Press},\ \bibinfo {year} {2011})\BibitemShut {NoStop}%
\bibitem [{SM()}]{SM}%
  \BibitemOpen
  \href@noop {} {}\bibinfo {note} {See Supplemental Material for preliminaries,
  proofs and extensions of results, which includes
  Refs.~\cite{kolodynski2013efficient,kolodynski2014precision,katariya2020geometric,helstrom1967minimum,yuan2017fidelity,perez2006contractivity}.}\BibitemShut
  {Stop}%
\bibitem [{\citenamefont {Ko{\l}ody{\'{n}}ski}\ and\ \citenamefont
  {Demkowicz-Dobrza{\'{n}}ski}(2013)}]{kolodynski2013efficient}%
  \BibitemOpen
  \bibfield  {author} {\bibinfo {author} {\bibfnamefont {J.}~\bibnamefont
  {Ko{\l}ody{\'{n}}ski}}\ and\ \bibinfo {author} {\bibfnamefont
  {R.}~\bibnamefont {Demkowicz-Dobrza{\'{n}}ski}},\ }\bibfield  {title}
  {\bibinfo {title} {Efficient tools for quantum metrology with uncorrelated
  noise},\ }\href@noop {} {\bibfield  {journal} {\bibinfo  {journal} {New J.
  Phys.}\ }\textbf {\bibinfo {volume} {15}},\ \bibinfo {pages} {073043}
  (\bibinfo {year} {2013})}\BibitemShut {NoStop}%
\bibitem [{\citenamefont {Kolodynski}(2014)}]{kolodynski2014precision}%
  \BibitemOpen
  \bibfield  {author} {\bibinfo {author} {\bibfnamefont {J.}~\bibnamefont
  {Kolodynski}},\ }\bibfield  {title} {\bibinfo {title} {Precision bounds in
  noisy quantum metrology},\ }\href@noop {} {\bibfield  {journal} {\bibinfo
  {journal} {(PhD Thesis) arXiv:1409.0535}\ } (\bibinfo {year}
  {2014})}\BibitemShut {NoStop}%
\bibitem [{\citenamefont {Katariya}\ and\ \citenamefont
  {Wilde}(2021)}]{katariya2020geometric}%
  \BibitemOpen
  \bibfield  {author} {\bibinfo {author} {\bibfnamefont {V.}~\bibnamefont
  {Katariya}}\ and\ \bibinfo {author} {\bibfnamefont {M.~M.}\ \bibnamefont
  {Wilde}},\ }\bibfield  {title} {\bibinfo {title} {Geometric
  distinguishability measures limit quantum channel estimation and
  discrimination},\ }\href@noop {} {\bibfield  {journal} {\bibinfo  {journal}
  {Quantum Inf. Process.}\ }\textbf {\bibinfo {volume} {20}},\ \bibinfo {pages}
  {78} (\bibinfo {year} {2021})}\BibitemShut {NoStop}%
\bibitem [{\citenamefont {Helstrom}(1967)}]{helstrom1967minimum}%
  \BibitemOpen
  \bibfield  {author} {\bibinfo {author} {\bibfnamefont {C.~W.}\ \bibnamefont
  {Helstrom}},\ }\bibfield  {title} {\bibinfo {title} {Minimum mean-squared
  error of estimates in quantum statistics},\ }\href@noop {} {\bibfield
  {journal} {\bibinfo  {journal} {Phys. Lett. A}\ }\textbf {\bibinfo {volume}
  {25}},\ \bibinfo {pages} {101} (\bibinfo {year} {1967})}\BibitemShut
  {NoStop}%
\bibitem [{\citenamefont {Yuan}\ and\ \citenamefont
  {Fung}(2017)}]{yuan2017fidelity}%
  \BibitemOpen
  \bibfield  {author} {\bibinfo {author} {\bibfnamefont {H.}~\bibnamefont
  {Yuan}}\ and\ \bibinfo {author} {\bibfnamefont {C.-H.~F.}\ \bibnamefont
  {Fung}},\ }\bibfield  {title} {\bibinfo {title} {Fidelity and fisher
  information on quantum channels},\ }\href@noop {} {\bibfield  {journal}
  {\bibinfo  {journal} {New J. Phys.}\ }\textbf {\bibinfo {volume} {19}},\
  \bibinfo {pages} {113039} (\bibinfo {year} {2017})}\BibitemShut {NoStop}%
\bibitem [{\citenamefont {Perez-Garcia}\ \emph {et~al.}(2006)\citenamefont
  {Perez-Garcia}, \citenamefont {Wolf}, \citenamefont {Petz},\ and\
  \citenamefont {Ruskai}}]{perez2006contractivity}%
  \BibitemOpen
  \bibfield  {author} {\bibinfo {author} {\bibfnamefont {D.}~\bibnamefont
  {Perez-Garcia}}, \bibinfo {author} {\bibfnamefont {M.~M.}\ \bibnamefont
  {Wolf}}, \bibinfo {author} {\bibfnamefont {D.}~\bibnamefont {Petz}},\ and\
  \bibinfo {author} {\bibfnamefont {M.~B.}\ \bibnamefont {Ruskai}},\ }\bibfield
   {title} {\bibinfo {title} {Contractivity of positive and trace-preserving
  maps under {Lp} norms},\ }\href@noop {} {\bibfield  {journal} {\bibinfo
  {journal} {J. Math. Phys.}\ }\textbf {\bibinfo {volume} {47}} (\bibinfo
  {year} {2006})}\BibitemShut {NoStop}%
\bibitem [{\citenamefont {Nielsen}\ and\ \citenamefont
  {Chuang}(2010)}]{nielsen2002quantum}%
  \BibitemOpen
  \bibfield  {author} {\bibinfo {author} {\bibfnamefont {M.~A.}\ \bibnamefont
  {Nielsen}}\ and\ \bibinfo {author} {\bibfnamefont {I.~L.}\ \bibnamefont
  {Chuang}},\ }\href@noop {} {\emph {\bibinfo {title} {Quantum Computation and
  Quantum Information}}}\ (\bibinfo  {publisher} {Cambridge University Press},\
  \bibinfo {year} {2010})\BibitemShut {NoStop}%
\bibitem [{\citenamefont {Watrous}(2018)}]{watrous2018theory}%
  \BibitemOpen
  \bibfield  {author} {\bibinfo {author} {\bibfnamefont {J.}~\bibnamefont
  {Watrous}},\ }\href@noop {} {\emph {\bibinfo {title} {The theory of quantum
  information}}}\ (\bibinfo  {publisher} {Cambridge University Press},\
  \bibinfo {year} {2018})\BibitemShut {NoStop}%
\bibitem [{Note2()}]{Note2}%
  \BibitemOpen
  \bibinfo {note} {Later on, when we say \protect \emph {ancilla} we usually
  mean the noiseless ancillary system that extends the probe system to a larger
  Hilbert space. We do not explicitly discuss the ancillary qubits used in
  practice to implement CPTP controls, e.g., through Stinespring
  dilation.}\BibitemShut {Stop}%
\bibitem [{Note3()}]{Note3}%
  \BibitemOpen
  \bibinfo {note} {We use the big-O notation for non-negative functions where
  $f(n) = O(g(n))$ means $f(n) \leq c g(n)$ for some constant $c$ and large
  enough $n$. Similarly, $\Omega $, $\Theta $ and $o$ are also used in a
  conventional manner.}\BibitemShut {Stop}%
\bibitem [{Note4()}]{Note4}%
  \BibitemOpen
  \bibinfo {note} {Note that although $H$ is not uniquely defined given a
  quantum channel, $H\protect \notin {\protect \mathcal {S}}$ is a well-defined
  condition with no ambiguity.}\BibitemShut {Stop}%
\bibitem [{\citenamefont {Wan}\ and\ \citenamefont
  {Lasenby}(2022)}]{wan2022bounds}%
  \BibitemOpen
  \bibfield  {author} {\bibinfo {author} {\bibfnamefont {K.}~\bibnamefont
  {Wan}}\ and\ \bibinfo {author} {\bibfnamefont {R.}~\bibnamefont {Lasenby}},\
  }\bibfield  {title} {\bibinfo {title} {Bounds on adaptive quantum metrology
  under {Markovian} noise},\ }\href@noop {} {\bibfield  {journal} {\bibinfo
  {journal} {Phys. Rev. Res.}\ }\textbf {\bibinfo {volume} {4}},\ \bibinfo
  {pages} {033092} (\bibinfo {year} {2022})}\BibitemShut {NoStop}%
\bibitem [{\citenamefont {Kurdzialek}\ \emph {et~al.}(2023)\citenamefont
  {Kurdzialek}, \citenamefont {Gorecki}, \citenamefont {Albarelli},\ and\
  \citenamefont {Demkowicz-Dobrzanski}}]{kurdzialek2022using}%
  \BibitemOpen
  \bibfield  {author} {\bibinfo {author} {\bibfnamefont {S.}~\bibnamefont
  {Kurdzialek}}, \bibinfo {author} {\bibfnamefont {W.}~\bibnamefont {Gorecki}},
  \bibinfo {author} {\bibfnamefont {F.}~\bibnamefont {Albarelli}},\ and\
  \bibinfo {author} {\bibfnamefont {R.}~\bibnamefont {Demkowicz-Dobrzanski}},\
  }\bibfield  {title} {\bibinfo {title} {Using adaptiveness and causal
  superpositions against noise in quantum metrology},\ }\href@noop {}
  {\bibfield  {journal} {\bibinfo  {journal} {Phys. Rev. Lett.}\ }\textbf
  {\bibinfo {volume} {131}},\ \bibinfo {pages} {090801} (\bibinfo {year}
  {2023})}\BibitemShut {NoStop}%
\bibitem [{\citenamefont {Altherr}\ and\ \citenamefont
  {Yang}(2021)}]{altherr2021quantum}%
  \BibitemOpen
  \bibfield  {author} {\bibinfo {author} {\bibfnamefont {A.}~\bibnamefont
  {Altherr}}\ and\ \bibinfo {author} {\bibfnamefont {Y.}~\bibnamefont {Yang}},\
  }\bibfield  {title} {\bibinfo {title} {Quantum metrology for non-{Markovian}
  processes},\ }\href@noop {} {\bibfield  {journal} {\bibinfo  {journal} {Phys.
  Rev. Lett.}\ }\textbf {\bibinfo {volume} {127}},\ \bibinfo {pages} {060501}
  (\bibinfo {year} {2021})}\BibitemShut {NoStop}%
\bibitem [{\citenamefont {Liu}\ \emph {et~al.}(2023)\citenamefont {Liu},
  \citenamefont {Hu}, \citenamefont {Yuan},\ and\ \citenamefont
  {Yang}}]{liu2023optimal}%
  \BibitemOpen
  \bibfield  {author} {\bibinfo {author} {\bibfnamefont {Q.}~\bibnamefont
  {Liu}}, \bibinfo {author} {\bibfnamefont {Z.}~\bibnamefont {Hu}}, \bibinfo
  {author} {\bibfnamefont {H.}~\bibnamefont {Yuan}},\ and\ \bibinfo {author}
  {\bibfnamefont {Y.}~\bibnamefont {Yang}},\ }\bibfield  {title} {\bibinfo
  {title} {Optimal strategies of quantum metrology with a strict hierarchy},\
  }\href@noop {} {\bibfield  {journal} {\bibinfo  {journal} {Phys. Rev. Lett.}\
  }\textbf {\bibinfo {volume} {130}},\ \bibinfo {pages} {070803} (\bibinfo
  {year} {2023})}\BibitemShut {NoStop}%
\bibitem [{\citenamefont {Fujiwara}\ and\ \citenamefont
  {Imai}(2008)}]{fujiwara2008fibre}%
  \BibitemOpen
  \bibfield  {author} {\bibinfo {author} {\bibfnamefont {A.}~\bibnamefont
  {Fujiwara}}\ and\ \bibinfo {author} {\bibfnamefont {H.}~\bibnamefont
  {Imai}},\ }\bibfield  {title} {\bibinfo {title} {A fibre bundle over
  manifolds of quantum channels and its application to quantum statistics},\
  }\href@noop {} {\bibfield  {journal} {\bibinfo  {journal} {J. Phys. A: Math.
  Theor.}\ }\textbf {\bibinfo {volume} {41}},\ \bibinfo {pages} {255304}
  (\bibinfo {year} {2008})}\BibitemShut {NoStop}%
\bibitem [{\citenamefont {King}\ and\ \citenamefont
  {Ruskai}(2001)}]{king2001minimal}%
  \BibitemOpen
  \bibfield  {author} {\bibinfo {author} {\bibfnamefont {C.}~\bibnamefont
  {King}}\ and\ \bibinfo {author} {\bibfnamefont {M.~B.}\ \bibnamefont
  {Ruskai}},\ }\bibfield  {title} {\bibinfo {title} {Minimal entropy of states
  emerging from noisy quantum channels},\ }\href@noop {} {\bibfield  {journal}
  {\bibinfo  {journal} {IEEE Trans. Inf. Theory}\ }\textbf {\bibinfo {volume}
  {47}},\ \bibinfo {pages} {192} (\bibinfo {year} {2001})}\BibitemShut
  {NoStop}%
\bibitem [{\citenamefont {Kurdzialek}\ \emph {et~al.}(2024)\citenamefont
  {Kurdzialek}, \citenamefont {Dulian}, \citenamefont {Majsak}, \citenamefont
  {Chakraborty},\ and\ \citenamefont
  {Demkowicz-Dobrzanski}}]{kurdzialek2024quantum}%
  \BibitemOpen
  \bibfield  {author} {\bibinfo {author} {\bibfnamefont {S.}~\bibnamefont
  {Kurdzialek}}, \bibinfo {author} {\bibfnamefont {P.}~\bibnamefont {Dulian}},
  \bibinfo {author} {\bibfnamefont {J.}~\bibnamefont {Majsak}}, \bibinfo
  {author} {\bibfnamefont {S.}~\bibnamefont {Chakraborty}},\ and\ \bibinfo
  {author} {\bibfnamefont {R.}~\bibnamefont {Demkowicz-Dobrzanski}},\
  }\bibfield  {title} {\bibinfo {title} {Quantum metrology using quantum combs
  and tensor network formalism},\ }\href@noop {} {\bibfield  {journal}
  {\bibinfo  {journal} {arXiv:2403.04854}\ } (\bibinfo {year}
  {2024})}\BibitemShut {NoStop}%
\bibitem [{\citenamefont {Zhou}\ and\ \citenamefont
  {Jiang}(2019)}]{zhou2019exact}%
  \BibitemOpen
  \bibfield  {author} {\bibinfo {author} {\bibfnamefont {S.}~\bibnamefont
  {Zhou}}\ and\ \bibinfo {author} {\bibfnamefont {L.}~\bibnamefont {Jiang}},\
  }\bibfield  {title} {\bibinfo {title} {An exact correspondence between the
  quantum {Fisher} information and the {Bures} metric},\ }\href@noop {}
  {\bibfield  {journal} {\bibinfo  {journal} {arXiv:1910.08473}\ } (\bibinfo
  {year} {2019})}\BibitemShut {NoStop}%
\bibitem [{\citenamefont {H{\"u}bner}(1992)}]{hubner1992explicit}%
  \BibitemOpen
  \bibfield  {author} {\bibinfo {author} {\bibfnamefont {M.}~\bibnamefont
  {H{\"u}bner}},\ }\bibfield  {title} {\bibinfo {title} {Explicit computation
  of the bures distance for density matrices},\ }\href@noop {} {\bibfield
  {journal} {\bibinfo  {journal} {Phys. Lett. A}\ }\textbf {\bibinfo {volume}
  {163}},\ \bibinfo {pages} {239} (\bibinfo {year} {1992})}\BibitemShut
  {NoStop}%
\bibitem [{\citenamefont {Hiai}\ and\ \citenamefont
  {Ruskai}(2016)}]{hiai2016contraction}%
  \BibitemOpen
  \bibfield  {author} {\bibinfo {author} {\bibfnamefont {F.}~\bibnamefont
  {Hiai}}\ and\ \bibinfo {author} {\bibfnamefont {M.~B.}\ \bibnamefont
  {Ruskai}},\ }\bibfield  {title} {\bibinfo {title} {Contraction coefficients
  for noisy quantum channels},\ }\href@noop {} {\bibfield  {journal} {\bibinfo
  {journal} {J. Math. Phys.}\ }\textbf {\bibinfo {volume} {57}} (\bibinfo
  {year} {2016})}\BibitemShut {NoStop}%
\bibitem [{\citenamefont {Len}\ \emph {et~al.}(2022)\citenamefont {Len},
  \citenamefont {Gefen}, \citenamefont {Retzker},\ and\ \citenamefont
  {Ko{\l}ody{\'n}ski}}]{len2021quantum}%
  \BibitemOpen
  \bibfield  {author} {\bibinfo {author} {\bibfnamefont {Y.~L.}\ \bibnamefont
  {Len}}, \bibinfo {author} {\bibfnamefont {T.}~\bibnamefont {Gefen}}, \bibinfo
  {author} {\bibfnamefont {A.}~\bibnamefont {Retzker}},\ and\ \bibinfo {author}
  {\bibfnamefont {J.}~\bibnamefont {Ko{\l}ody{\'n}ski}},\ }\bibfield  {title}
  {\bibinfo {title} {Quantum metrology with imperfect measurements},\
  }\href@noop {} {\bibfield  {journal} {\bibinfo  {journal} {Nat. Commun.}\
  }\textbf {\bibinfo {volume} {13}},\ \bibinfo {pages} {6971} (\bibinfo {year}
  {2022})}\BibitemShut {NoStop}%
\bibitem [{\citenamefont {Zhou}\ \emph {et~al.}(2023)\citenamefont {Zhou},
  \citenamefont {Michalakis},\ and\ \citenamefont {Gefen}}]{zhou2023optimal}%
  \BibitemOpen
  \bibfield  {author} {\bibinfo {author} {\bibfnamefont {S.}~\bibnamefont
  {Zhou}}, \bibinfo {author} {\bibfnamefont {S.}~\bibnamefont {Michalakis}},\
  and\ \bibinfo {author} {\bibfnamefont {T.}~\bibnamefont {Gefen}},\ }\bibfield
   {title} {\bibinfo {title} {Optimal protocols for quantum metrology with
  noisy measurements},\ }\href@noop {} {\bibfield  {journal} {\bibinfo
  {journal} {PRX Quantum}\ }\textbf {\bibinfo {volume} {4}},\ \bibinfo {pages}
  {040305} (\bibinfo {year} {2023})}\BibitemShut {NoStop}%
\bibitem [{\citenamefont {Krantz}\ \emph {et~al.}(2019)\citenamefont {Krantz},
  \citenamefont {Kjaergaard}, \citenamefont {Yan}, \citenamefont {Orlando},
  \citenamefont {Gustavsson},\ and\ \citenamefont
  {Oliver}}]{krantz2019quantum}%
  \BibitemOpen
  \bibfield  {author} {\bibinfo {author} {\bibfnamefont {P.}~\bibnamefont
  {Krantz}}, \bibinfo {author} {\bibfnamefont {M.}~\bibnamefont {Kjaergaard}},
  \bibinfo {author} {\bibfnamefont {F.}~\bibnamefont {Yan}}, \bibinfo {author}
  {\bibfnamefont {T.~P.}\ \bibnamefont {Orlando}}, \bibinfo {author}
  {\bibfnamefont {S.}~\bibnamefont {Gustavsson}},\ and\ \bibinfo {author}
  {\bibfnamefont {W.~D.}\ \bibnamefont {Oliver}},\ }\bibfield  {title}
  {\bibinfo {title} {A quantum engineer's guide to superconducting qubits},\
  }\href@noop {} {\bibfield  {journal} {\bibinfo  {journal} {Appl. Phys. Rev.}\
  }\textbf {\bibinfo {volume} {6}},\ \bibinfo {pages} {021318} (\bibinfo {year}
  {2019})}\BibitemShut {NoStop}%
\bibitem [{\citenamefont {Bruzewicz}\ \emph {et~al.}(2019)\citenamefont
  {Bruzewicz}, \citenamefont {Chiaverini}, \citenamefont {McConnell},\ and\
  \citenamefont {Sage}}]{bruzewicz2019trapped}%
  \BibitemOpen
  \bibfield  {author} {\bibinfo {author} {\bibfnamefont {C.~D.}\ \bibnamefont
  {Bruzewicz}}, \bibinfo {author} {\bibfnamefont {J.}~\bibnamefont
  {Chiaverini}}, \bibinfo {author} {\bibfnamefont {R.}~\bibnamefont
  {McConnell}},\ and\ \bibinfo {author} {\bibfnamefont {J.~M.}\ \bibnamefont
  {Sage}},\ }\bibfield  {title} {\bibinfo {title} {Trapped-ion quantum
  computing: Progress and challenges},\ }\href@noop {} {\bibfield  {journal}
  {\bibinfo  {journal} {Appl. Phys. Rev.}\ }\textbf {\bibinfo {volume} {6}}
  (\bibinfo {year} {2019})}\BibitemShut {NoStop}%
\bibitem [{\citenamefont {Dutt}\ \emph {et~al.}(2007)\citenamefont {Dutt},
  \citenamefont {Childress}, \citenamefont {Jiang}, \citenamefont {Togan},
  \citenamefont {Maze}, \citenamefont {Jelezko}, \citenamefont {Zibrov},
  \citenamefont {Hemmer},\ and\ \citenamefont {Lukin}}]{dutt2007quantum}%
  \BibitemOpen
  \bibfield  {author} {\bibinfo {author} {\bibfnamefont {M.~G.}\ \bibnamefont
  {Dutt}}, \bibinfo {author} {\bibfnamefont {L.}~\bibnamefont {Childress}},
  \bibinfo {author} {\bibfnamefont {L.}~\bibnamefont {Jiang}}, \bibinfo
  {author} {\bibfnamefont {E.}~\bibnamefont {Togan}}, \bibinfo {author}
  {\bibfnamefont {J.}~\bibnamefont {Maze}}, \bibinfo {author} {\bibfnamefont
  {F.}~\bibnamefont {Jelezko}}, \bibinfo {author} {\bibfnamefont
  {A.}~\bibnamefont {Zibrov}}, \bibinfo {author} {\bibfnamefont
  {P.}~\bibnamefont {Hemmer}},\ and\ \bibinfo {author} {\bibfnamefont
  {M.}~\bibnamefont {Lukin}},\ }\bibfield  {title} {\bibinfo {title} {Quantum
  register based on individual electronic and nuclear spin qubits in diamond},\
  }\href@noop {} {\bibfield  {journal} {\bibinfo  {journal} {Science}\ }\textbf
  {\bibinfo {volume} {316}},\ \bibinfo {pages} {1312} (\bibinfo {year}
  {2007})}\BibitemShut {NoStop}%
\bibitem [{\citenamefont {Jiang}\ \emph {et~al.}(2009)\citenamefont {Jiang},
  \citenamefont {Hodges}, \citenamefont {Maze}, \citenamefont {Maurer},
  \citenamefont {Taylor}, \citenamefont {Cory}, \citenamefont {Hemmer},
  \citenamefont {Walsworth}, \citenamefont {Yacoby}, \citenamefont {Zibrov}
  \emph {et~al.}}]{jiang2009repetitive}%
  \BibitemOpen
  \bibfield  {author} {\bibinfo {author} {\bibfnamefont {L.}~\bibnamefont
  {Jiang}}, \bibinfo {author} {\bibfnamefont {J.}~\bibnamefont {Hodges}},
  \bibinfo {author} {\bibfnamefont {J.}~\bibnamefont {Maze}}, \bibinfo {author}
  {\bibfnamefont {P.}~\bibnamefont {Maurer}}, \bibinfo {author} {\bibfnamefont
  {J.}~\bibnamefont {Taylor}}, \bibinfo {author} {\bibfnamefont
  {D.}~\bibnamefont {Cory}}, \bibinfo {author} {\bibfnamefont {P.}~\bibnamefont
  {Hemmer}}, \bibinfo {author} {\bibfnamefont {R.~L.}\ \bibnamefont
  {Walsworth}}, \bibinfo {author} {\bibfnamefont {A.}~\bibnamefont {Yacoby}},
  \bibinfo {author} {\bibfnamefont {A.~S.}\ \bibnamefont {Zibrov}}, \emph
  {et~al.},\ }\bibfield  {title} {\bibinfo {title} {Repetitive readout of a
  single electronic spin via quantum logic with nuclear spin ancillae},\
  }\href@noop {} {\bibfield  {journal} {\bibinfo  {journal} {Science}\ }\textbf
  {\bibinfo {volume} {326}},\ \bibinfo {pages} {267} (\bibinfo {year}
  {2009})}\BibitemShut {NoStop}%
\bibitem [{ibm()}]{ibm2024}%
  \BibitemOpen
  \href@noop {} {\bibinfo {title} {{IBM Quantum}}},\ \bibinfo {note}
  {\url{https://quantum-computing.ibm.com/services?
  services=systems}~(Date:~2024-07-08)}\BibitemShut {NoStop}%
\bibitem [{\citenamefont {Liu}\ and\ \citenamefont
  {Yang}(2024{\natexlab{a}})}]{liu2024heisenberg}%
  \BibitemOpen
  \bibfield  {author} {\bibinfo {author} {\bibfnamefont {Q.}~\bibnamefont
  {Liu}}\ and\ \bibinfo {author} {\bibfnamefont {Y.}~\bibnamefont {Yang}},\
  }\bibfield  {title} {\bibinfo {title} {Heisenberg-limited quantum metrology
  without ancilla},\ }\href@noop {} {\bibfield  {journal} {\bibinfo  {journal}
  {arXiv:2403.04585}\ } (\bibinfo {year} {2024}{\natexlab{a}})}\BibitemShut
  {NoStop}%
\bibitem [{\citenamefont {Liu}\ and\ \citenamefont
  {Yang}(2024{\natexlab{b}})}]{liu2024efficient}%
  \BibitemOpen
  \bibfield  {author} {\bibinfo {author} {\bibfnamefont {Q.}~\bibnamefont
  {Liu}}\ and\ \bibinfo {author} {\bibfnamefont {Y.}~\bibnamefont {Yang}},\
  }\bibfield  {title} {\bibinfo {title} {Efficient tensor networks for
  control-enhanced quantum metrology},\ }\href@noop {} {\bibfield  {journal}
  {\bibinfo  {journal} {arXiv:2403.09519}\ } (\bibinfo {year}
  {2024}{\natexlab{b}})}\BibitemShut {NoStop}%
\bibitem [{\citenamefont {Chen}\ \emph
  {et~al.}(2022{\natexlab{a}})\citenamefont {Chen}, \citenamefont {Zhou},
  \citenamefont {Seif},\ and\ \citenamefont {Jiang}}]{chen2022quantum}%
  \BibitemOpen
  \bibfield  {author} {\bibinfo {author} {\bibfnamefont {S.}~\bibnamefont
  {Chen}}, \bibinfo {author} {\bibfnamefont {S.}~\bibnamefont {Zhou}}, \bibinfo
  {author} {\bibfnamefont {A.}~\bibnamefont {Seif}},\ and\ \bibinfo {author}
  {\bibfnamefont {L.}~\bibnamefont {Jiang}},\ }\bibfield  {title} {\bibinfo
  {title} {Quantum advantages for pauli channel estimation},\ }\href@noop {}
  {\bibfield  {journal} {\bibinfo  {journal} {Phys. Rev. A}\ }\textbf {\bibinfo
  {volume} {105}},\ \bibinfo {pages} {032435} (\bibinfo {year}
  {2022}{\natexlab{a}})}\BibitemShut {NoStop}%
\bibitem [{\citenamefont {Chen}\ \emph {et~al.}(2024)\citenamefont {Chen},
  \citenamefont {Oh}, \citenamefont {Zhou}, \citenamefont {Huang},\ and\
  \citenamefont {Jiang}}]{chen2023tight}%
  \BibitemOpen
  \bibfield  {author} {\bibinfo {author} {\bibfnamefont {S.}~\bibnamefont
  {Chen}}, \bibinfo {author} {\bibfnamefont {C.}~\bibnamefont {Oh}}, \bibinfo
  {author} {\bibfnamefont {S.}~\bibnamefont {Zhou}}, \bibinfo {author}
  {\bibfnamefont {H.-Y.}\ \bibnamefont {Huang}},\ and\ \bibinfo {author}
  {\bibfnamefont {L.}~\bibnamefont {Jiang}},\ }\bibfield  {title} {\bibinfo
  {title} {Tight bounds on pauli channel learning without entanglement},\
  }\href@noop {} {\bibfield  {journal} {\bibinfo  {journal} {Phys. Rev. Lett.}\
  }\textbf {\bibinfo {volume} {132}},\ \bibinfo {pages} {180805} (\bibinfo
  {year} {2024})}\BibitemShut {NoStop}%
\bibitem [{\citenamefont {Chen}\ and\ \citenamefont
  {Gong}(2023)}]{chen2023efficient}%
  \BibitemOpen
  \bibfield  {author} {\bibinfo {author} {\bibfnamefont {S.}~\bibnamefont
  {Chen}}\ and\ \bibinfo {author} {\bibfnamefont {W.}~\bibnamefont {Gong}},\
  }\bibfield  {title} {\bibinfo {title} {Efficient pauli channel estimation
  with logarithmic quantum memory},\ }\href@noop {} {\bibfield  {journal}
  {\bibinfo  {journal} {arXiv:2309.14326}\ } (\bibinfo {year}
  {2023})}\BibitemShut {NoStop}%
\bibitem [{\citenamefont {Aharonov}\ \emph {et~al.}(2022)\citenamefont
  {Aharonov}, \citenamefont {Cotler},\ and\ \citenamefont
  {Qi}}]{aharonov2022quantum}%
  \BibitemOpen
  \bibfield  {author} {\bibinfo {author} {\bibfnamefont {D.}~\bibnamefont
  {Aharonov}}, \bibinfo {author} {\bibfnamefont {J.}~\bibnamefont {Cotler}},\
  and\ \bibinfo {author} {\bibfnamefont {X.-L.}\ \bibnamefont {Qi}},\
  }\bibfield  {title} {\bibinfo {title} {Quantum algorithmic measurement},\
  }\href@noop {} {\bibfield  {journal} {\bibinfo  {journal} {Nat. Commun.}\
  }\textbf {\bibinfo {volume} {13}} (\bibinfo {year} {2022})}\BibitemShut
  {NoStop}%
\bibitem [{\citenamefont {Huang}\ \emph {et~al.}(2022)\citenamefont {Huang},
  \citenamefont {Broughton}, \citenamefont {Cotler}, \citenamefont {Chen},
  \citenamefont {Li}, \citenamefont {Mohseni}, \citenamefont {Neven},
  \citenamefont {Babbush}, \citenamefont {Kueng}, \citenamefont {Preskill}
  \emph {et~al.}}]{huang2022quantum}%
  \BibitemOpen
  \bibfield  {author} {\bibinfo {author} {\bibfnamefont {H.-Y.}\ \bibnamefont
  {Huang}}, \bibinfo {author} {\bibfnamefont {M.}~\bibnamefont {Broughton}},
  \bibinfo {author} {\bibfnamefont {J.}~\bibnamefont {Cotler}}, \bibinfo
  {author} {\bibfnamefont {S.}~\bibnamefont {Chen}}, \bibinfo {author}
  {\bibfnamefont {J.}~\bibnamefont {Li}}, \bibinfo {author} {\bibfnamefont
  {M.}~\bibnamefont {Mohseni}}, \bibinfo {author} {\bibfnamefont
  {H.}~\bibnamefont {Neven}}, \bibinfo {author} {\bibfnamefont
  {R.}~\bibnamefont {Babbush}}, \bibinfo {author} {\bibfnamefont
  {R.}~\bibnamefont {Kueng}}, \bibinfo {author} {\bibfnamefont
  {J.}~\bibnamefont {Preskill}}, \emph {et~al.},\ }\bibfield  {title} {\bibinfo
  {title} {Quantum advantage in learning from experiments},\ }\href@noop {}
  {\bibfield  {journal} {\bibinfo  {journal} {Science}\ }\textbf {\bibinfo
  {volume} {376}},\ \bibinfo {pages} {1182} (\bibinfo {year}
  {2022})}\BibitemShut {NoStop}%
\bibitem [{\citenamefont {Chen}\ \emph
  {et~al.}(2022{\natexlab{b}})\citenamefont {Chen}, \citenamefont {Cotler},
  \citenamefont {Huang},\ and\ \citenamefont {Li}}]{chen2022exponential}%
  \BibitemOpen
  \bibfield  {author} {\bibinfo {author} {\bibfnamefont {S.}~\bibnamefont
  {Chen}}, \bibinfo {author} {\bibfnamefont {J.}~\bibnamefont {Cotler}},
  \bibinfo {author} {\bibfnamefont {H.-Y.}\ \bibnamefont {Huang}},\ and\
  \bibinfo {author} {\bibfnamefont {J.}~\bibnamefont {Li}},\ }\bibfield
  {title} {\bibinfo {title} {Exponential separations between learning with and
  without quantum memory},\ }in\ \href@noop {} {\emph {\bibinfo {booktitle}
  {2021 IEEE 62nd Annual Symposium on Foundations of Computer Science
  (FOCS)}}}\ (\bibinfo {organization} {IEEE},\ \bibinfo {year} {2022})\ pp.\
  \bibinfo {pages} {574--585}\BibitemShut {NoStop}%
\bibitem [{\citenamefont {Chen}\ \emph {et~al.}(2023)\citenamefont {Chen},
  \citenamefont {Cotler}, \citenamefont {Huang},\ and\ \citenamefont
  {Li}}]{chen2023complexity}%
  \BibitemOpen
  \bibfield  {author} {\bibinfo {author} {\bibfnamefont {S.}~\bibnamefont
  {Chen}}, \bibinfo {author} {\bibfnamefont {J.}~\bibnamefont {Cotler}},
  \bibinfo {author} {\bibfnamefont {H.-Y.}\ \bibnamefont {Huang}},\ and\
  \bibinfo {author} {\bibfnamefont {J.}~\bibnamefont {Li}},\ }\bibfield
  {title} {\bibinfo {title} {The complexity of nisq},\ }\href@noop {}
  {\bibfield  {journal} {\bibinfo  {journal} {Nat. Commun.}\ }\textbf {\bibinfo
  {volume} {14}},\ \bibinfo {pages} {6001} (\bibinfo {year}
  {2023})}\BibitemShut {NoStop}%
\bibitem [{\citenamefont {Tsubouchi}\ \emph {et~al.}(2023)\citenamefont
  {Tsubouchi}, \citenamefont {Sagawa},\ and\ \citenamefont
  {Yoshioka}}]{tsubouchi2023universal}%
  \BibitemOpen
  \bibfield  {author} {\bibinfo {author} {\bibfnamefont {K.}~\bibnamefont
  {Tsubouchi}}, \bibinfo {author} {\bibfnamefont {T.}~\bibnamefont {Sagawa}},\
  and\ \bibinfo {author} {\bibfnamefont {N.}~\bibnamefont {Yoshioka}},\
  }\bibfield  {title} {\bibinfo {title} {Universal cost bound of quantum error
  mitigation based on quantum estimation theory},\ }\href@noop {} {\bibfield
  {journal} {\bibinfo  {journal} {Phys. Rev. Lett.}\ }\textbf {\bibinfo
  {volume} {131}},\ \bibinfo {pages} {210601} (\bibinfo {year}
  {2023})}\BibitemShut {NoStop}%
\bibitem [{\citenamefont {Takagi}\ \emph {et~al.}(2023)\citenamefont {Takagi},
  \citenamefont {Tajima},\ and\ \citenamefont {Gu}}]{takagi2023universal}%
  \BibitemOpen
  \bibfield  {author} {\bibinfo {author} {\bibfnamefont {R.}~\bibnamefont
  {Takagi}}, \bibinfo {author} {\bibfnamefont {H.}~\bibnamefont {Tajima}},\
  and\ \bibinfo {author} {\bibfnamefont {M.}~\bibnamefont {Gu}},\ }\bibfield
  {title} {\bibinfo {title} {Universal sampling lower bounds for quantum error
  mitigation},\ }\href@noop {} {\bibfield  {journal} {\bibinfo  {journal}
  {Phys. Rev. Lett.}\ }\textbf {\bibinfo {volume} {131}},\ \bibinfo {pages}
  {210602} (\bibinfo {year} {2023})}\BibitemShut {NoStop}%
\bibitem [{\citenamefont {Quek}\ \emph {et~al.}(2022)\citenamefont {Quek},
  \citenamefont {Fran{\c{c}}a}, \citenamefont {Khatri}, \citenamefont {Meyer},\
  and\ \citenamefont {Eisert}}]{quek2022exponentially}%
  \BibitemOpen
  \bibfield  {author} {\bibinfo {author} {\bibfnamefont {Y.}~\bibnamefont
  {Quek}}, \bibinfo {author} {\bibfnamefont {D.~S.}\ \bibnamefont
  {Fran{\c{c}}a}}, \bibinfo {author} {\bibfnamefont {S.}~\bibnamefont
  {Khatri}}, \bibinfo {author} {\bibfnamefont {J.~J.}\ \bibnamefont {Meyer}},\
  and\ \bibinfo {author} {\bibfnamefont {J.}~\bibnamefont {Eisert}},\
  }\bibfield  {title} {\bibinfo {title} {Exponentially tighter bounds on
  limitations of quantum error mitigation},\ }\href@noop {} {\bibfield
  {journal} {\bibinfo  {journal} {arXiv:2210.11505}\ } (\bibinfo {year}
  {2022})}\BibitemShut {NoStop}%
\end{thebibliography}%

\onecolumngrid
\newpage
\appendix


\setcounter{theorem}{0}
\setcounter{proposition}{0}
\setcounter{lemma}{0}
\setcounter{figure}{0}
\renewcommand{\thefigure}{S\arabic{figure}}
\renewcommand{\thelemma}{S\arabic{lemma}}
\renewcommand{\thetheorem}{S\arabic{theorem}}
\renewcommand{\thecorollary}{S\arabic{corollary}}
\renewcommand{\theproposition}{S\arabic{proposition}}
\renewcommand{\theHfigure}{Supplement.\arabic{figure}}
\renewcommand{\theHlemma}{Supplement.\arabic{lemma}}
\renewcommand{\theHtheorem}{Supplement.\arabic{theorem}}
\renewcommand{\theHcorollary}{Supplement.\arabic{corollary}}

\setcounter{page}{1}

\begin{center}
\large\textbf{Supplemental Material: Limits of noisy quantum metrology with restricted quantum controls}

\vspace{0.3in}

\normalsize Sisi Zhou
\end{center}

\section{Summary of the Supplemental Material}
\label{app:table}

\begin{itemize}[wide, labelwidth=!, labelindent=0pt]
\item \appref{app:preliminaries} and \appref{app:qubit} include preliminaries on QFI, quantum channel estimation, and classification of qubit channels. 
\item In \appref{app:HNKS}, we prove the HNKS condition is violated for all strictly contractive channels, which implies the optimality of the SQL in strategies (a) and (b) (see the first column in \tabref{table}). We also identify the forms of dephasing-class channels that satisfy the HNKS condition. 
\item In \appref{app:dephasing}, we consider estimating dephasing-class channels. We first prove an upper bound of $O(n)$ on the QFI in strategies (c), i.e., ancilla-free sequential strategies with unital controls. Then we generalize that to strategies (e). Finally, we consider CPTP controls and prove an upper bound of $O(n^{3/2})$ on the QFI in strategies (d). (Note that the results are slightly more general than described here, where we in fact allow the assistance of an arbitrary unbounded ancilla in the initial state and the final measurement, but not the controls inbetween.) 
\item In \appref{app:sql}, we show a metrological protocol that uses sequential strategies with only unitary controls and no ancilla that achieves the SQL for dephasing-class channels when RGNKS holds. It implies a lower bound $\Omega(n)$ of the QFI in strategies (c), (d) and (e). 
\item In \appref{app:dephasing-2}, we prove, to estimate dephasing-class channels in strategies (c), the QFI is at most a constant when RGNKS fails. 
\item In \appref{app:contractive}, we prove for strictly constractive channels, the QFI is at most a constant in strategies (d). 
\end{itemize}

\section{Preliminaries: Quantum Fisher information and quantum channel estimation}
\label{app:preliminaries}

Consider a quantum state $\rho_\theta$ which is a differentiable function of an unknown parameter $\theta$. The function $\rho_\theta$ is assumed to be differentiable around its true value. The value of $\theta$ can be inferred from the measurement statistics of $\rho_\theta$ on a positive operator-valued measure (POVM) $\{M_i\}$ ($M_i$ are positive semi-definite operators that sum up to $\id$). An estimator $\hat\theta$ maps the measurement outcomes to the inferred value of $\theta$. $\hat\theta$ is called a locally unbiased estimator at $\theta = \theta_0$, if $\bE[\hat\theta|\theta]|_{\theta = \theta_0} = \theta_0$ and $\partial_\theta\bE[\hat\theta|\theta]|_{\theta = \theta_0} = 1$, where $\bE[\hat\theta|\theta]$ means the expectation value of the estimator given the quantum state $\rho_\theta$. The Cram\'er--Rao bound~\cite{kay1993fundamentals,casella2002statistical,kobayashi2011probability} provides a lower bound on the estimation precision, i.e., the standard deviation of any locally unbiased estimator at $\theta$ that is 
\begin{equation}
    \Delta \hat\theta := \bE[(\hat\theta-\theta)^2|\theta]^{1/2} \geq \frac{1}{\sqrt{N_{\rm expr} F(\rho_\theta,\{M_i\})}},
\end{equation}
where $N_{\rm expr}$ is the number of repeated experiments and $F(\rho_\theta,\{M_i\}) $ is the (classical) Fisher information (FI) defined by 
\begin{equation}
    F(\rho_\theta,\{M_i\}) := \sum_{i:p_{i,\theta} \neq 0} \frac{(\partial_\theta p_{i,\theta})^2}{p_{i,\theta}} =  \sum_{i:\trace(\rho_\theta M_i) \neq 0} \frac{(\partial_\theta \trace(\rho_\theta M_i))^2}{\trace(\rho_\theta M_i)},
\end{equation}
as a function of the measurement probability distribution $\{p_{i,\theta} =  \trace(\rho_\theta M_i)\}$. The larger the FI, the preciser the parameter estimation can be. The quantum Fisher information (QFI)~\cite{braunstein1994statistical,holevo2011probabilistic,helstrom1976quantum,barndorff2000fisher,petz1996geometries,paris2009quantum} is the maximum FI among all possible POVM. 
It can be calculated as 
\begin{equation}
    F(\rho_\theta) :=\trace(\rho_\theta L^2),
\end{equation}
where $L$ is any Hermitian operator (called the symmetric logarithmic derivative) satisfying $\frac{1}{2}(L\rho_\theta+\rho_\theta L) = \partial_\theta\rho_\theta$. 
Then we have the quantum Cram\'er--Rao bound~\cite{holevo2011probabilistic,helstrom1976quantum,barndorff2000fisher},
\begin{equation}
    \Delta \hat\theta \geq \frac{1}{\sqrt{N_{\rm expr} F(\rho_\theta)}},
\end{equation}
which represents the ultimate estimation precision allowed by quantum mechanics.

The QFI is a nice information-theoretic quantity and below we list a few of its mathematical  properties~\cite{kolodynski2014precision,katariya2020geometric} that are useful. (Below we use $\rho_\theta$, $\sigma_\theta$ to represent arbitrary one-parameter quantum states.) 
\begin{itemize}[wide, labelwidth=!, labelindent=0pt]
    \item Faithfulness. $F(\rho_\theta) \geq 0$ and the equality holds if and only if $\partial_\theta \rho_\theta = 0$. 
    \item Data-processing inequality. $F(\mN(\rho_\theta)) \leq F(\rho_\theta)$ for any parameter-independent CPTP map $\mN(\cdot)$. 
    \item Convexity. Let $\rho_\theta = \sum_i p_i \rho_{\theta,i}$ where $\{p_i\}$ is a parameter-independent probability distribution and $\rho_{\theta,i}$ are density matrices, 
    \begin{equation}
        F(\rho_\theta) \leq \sum_i p_i F(\rho_{\theta,i}). 
    \end{equation}
    \item Additivity. $F(\rho_\theta \otimes \sigma_\theta) = F(\rho_\theta) + F(\sigma_\theta)$. 
    \item Expression of QFI from diagonalization. Let $\rho_\theta = \sum_{i} \lambda_{\theta,i}\ket{\psi_{\theta,i}}\bra{\psi_{\theta,i}}$ where $\lambda_{\theta,i}$ and $\ket{\psi_{\theta,i}}$ are its eigenvalues and eigenstates. 
    \begin{equation}
        F(\rho_\theta) = 2 \sum_{i,j:\lambda_{\theta,i}+\lambda_{\theta,j} \neq 0} \frac{\abs{\bra{\psi_{\theta,i}}\partial_\theta\rho_\theta\ket{\psi_{\theta,j}}}^2}{\lambda_{\theta,i}+\lambda_{\theta,j}}.
    \end{equation}
    In particular, when $\rho_\theta = \ket{\psi_\theta}\bra{\psi_\theta}$ is pure, $
        F(\rho_\theta) = 4 \left(\braket{\partial_\theta\psi_\theta|\partial_\theta\psi_\theta} - \abs{\braket{\psi_\theta|\partial_\theta\psi_\theta}}^2\right)$. 
    \item Relation to the Bures distance~\cite{zhou2019exact,hubner1992explicit}.
    \begin{equation}
    F(\rho_\theta) = \lim_{d\theta \rightarrow 0}\frac{d^2_{\rm B}(\rho_{\theta+d\theta},\rho_{\theta-d\theta})}{(d\theta)^2},    
    \end{equation}
    where $d_{\rm B}$ is the Bures distance defined by $d_{\rm B}(\sigma_1,\sigma_2) = \sqrt{2(1 - \trace\left(\sqrt{\sqrt{\sigma_1}\sigma_2\sqrt{\sigma_1}}\right))}$~\cite{helstrom1967minimum}.   
    \item Purification-based definition~\cite{fujiwara2008fibre,kolodynski2014precision}.     \begin{equation}
    F(\rho_\theta) = 4 \min_{\ket{\Psi_\theta}:\trace_E(\ket{\Psi_\theta}\bra{\Psi_\theta})=\rho_\theta} \braket{\partial_\theta\Psi_\theta|\partial_\theta\Psi_\theta},
    \end{equation}
    where $\ket{\Psi_\theta}$ is any (differentiable) purification state of $\rho_\theta$ when extended to an unbounded environment $E$. 
\end{itemize}

The definition of the QFI of quantum states can be extended to the QFI of quantum channels when the input state of a quantum channel are allowed to be chosen arbitrarily. Given a one-parameter quantum channel $\mE_\theta$, the channel QFI is~\cite{kolodynski2013efficient}
\begin{equation}
    F(\mE_\theta) := \sup_{\rho}F(\mE_\theta(\rho)) = \sup_{\psi = \ket{\psi}\bra{\psi}}F(\mE_\theta(\psi)), 
\end{equation}
where we use the convexity of the QFI in the second equality. We will use $F(\mE_\theta \otimes \id)$ to denote the ancilla-assisted channel QFI~\cite{fujiwara2008fibre} (Here $\id$ represents an identity quantum channel on an unbounded ancillary system), which means the input state $\rho$ can be chosen as an entangled state between the probe system and an ancillary system. Clearly, $F(\mE_\theta \otimes \id) \geq F(\mE_\theta)$. 

The purification-based definition of QFI can be applied to efficiently compute the ancilla-assisted channel QFI $F(\mE_\theta \otimes \id)$. We explain how it works, which includes some techniques that we will adopt later. (In this Letter, we consider only one-parameter quantum channels that can be represented using differentiable Kraus operators.) Let $\mE_\theta(\cdot) = \sum_{i=1}^r K_{\theta,i}(\cdot)K_{\theta,i}^\dagger$. Any purification $\ket{\Psi_\theta}$ of $\mE_\theta(\psi)$ has the form~\cite{nielsen2002quantum} 
    \begin{equation}
    \label{eq:purification-channel}
    \ket{\Psi_\theta} 
    = \sum_{i=1}^r K_{\theta,i}\ket{\psi} \otimes U_\theta \ket{i}_E 
    = \sum_{i=1}^r\sum_{j=1}^{r'} u_{\theta,ji} K_{\theta,i}\ket{\psi} \otimes \ket{j}_E
    = \sum_{j=1}^{r'} \tK_{\theta,j}\ket{\psi} \otimes \ket{j}_E,
    \end{equation}
    where $U_\theta$ is an isometric operator satisfying $U_\theta^\dagger U_\theta = \id$, $u_{\theta,ji} = \bra{j}U_\theta\ket{i}$ satisfies $u_\theta^\dagger u_\theta = \id$ and $\tK_{\theta,i} = \sum_{i} u_{\theta,ji} K_{\theta,i}$ is another Kraus representation of $\mE_\theta$.  
    Using the purification-based definition of QFI, we have 
    \begin{equation}
    \label{eq:purification-channel-2}
    \begin{split}
    F(\mE_\theta(\psi)) 
    &= 4 \min_{\ket{\Psi_\theta}\text{ in \eqref{eq:purification-channel}}} \braket{\dot\Psi_\theta|\dot\Psi_\theta} = 4 \min_{U_\theta:U_\theta^\dagger U_\theta = \id} \trace\bigg(\psi \sum_{j} \dot\tK^\dagger_{j} \dot\tK_{j} \bigg)\\
    &= 4 \min_{\text{Hermitian }h} \trace\bigg(\psi \sum_{j} \Big(\dot K_{j} - i \sum_{i} h_{ji}  K_{i}\Big)^\dagger \Big(\dot K_{j} - i \sum_{i} h_{ji}  K_{i}\Big) \bigg),
    \end{split}
    \end{equation}
    where $h = i  u^\dagger_\theta \dot u_\theta$ can be an arbitrary Hermitian operator. In general, we can without loss of generality always assume $u = e^{-ih\theta}$, or equivalently, 
    \begin{equation}
    \tK_{i} = K_{i},\quad 
    \dot\tK_{i} = \dot K_{i} - i\sum_j h_{ij}K_{j}, 
    \end{equation}
    where $h$ can be any Hermitian operators representing the degrees of freedom in choosing different Kraus representations. 
    Then we have    
    \begin{equation}
    \label{eq:channel-QFI}
    {F(\mE_\theta)} = 4\sup_\psi \min_{\text{Hermitian }h} \trace\bigg(\psi \sum_{j} \Big(\dot K_{\theta,j} - i \sum_{i} h_{ji}  K_{\theta,i}\Big)^\dagger \Big(\dot K_{\theta,j} - i \sum_{i} h_{ji}  K_{\theta,i}\Big) \bigg). 
    \end{equation}
    No efficient algorithms to compute \eqref{eq:channel-QFI} were known. However, when an unbounded ancilla is allowed, we have~\cite{fujiwara2008fibre,demkowicz2012elusive,kolodynski2013efficient} 
    \begin{equation}
    \begin{split}
    {F(\mE_\theta\otimes\id)} &= 4\sup_\rho \min_{\text{Hermitian }h} \trace\bigg(\rho \sum_{j} \Big(\dot K_{\theta,j} - i \sum_{i} h_{ji}  K_{\theta,i}\Big)^\dagger \Big(\dot K_{\theta,j} - i \sum_{i} h_{ji}  K_{\theta,i}\Big) \bigg)\\
    &= 4 \min_{\text{Hermitian }h} \bigg\|\sum_{j} \Big(\dot K_{\theta,j} - i \sum_{i} h_{ji}  K_{\theta,i}\Big)^\dagger \Big(\dot K_{\theta,j} - i \sum_{i} h_{ji}  K_{\theta,i}\Big) \bigg\|,  \label{eq:sdp}
    \end{split}
    \end{equation}
    and Sion's minimax theorem is used in the second step to exchange the orders of minimization and maximization. The ancilla-assisted channel QFI in \eqref{eq:sdp} is thus efficiently computable using a semi-definite program. 

Another useful property of the ancilla-assisted channel QFI that we will use later is the chain rule of root QFI~\cite{yuan2017fidelity,katariya2020geometric}, i.e., 
\begin{equation}
    F( (\mM_\theta \circ \mN_\theta) \otimes \id)^{1/2} \leq F(  \mN_\theta \otimes \id)^{1/2} + F( \mM_\theta \otimes \id)^{1/2},
\end{equation}
where $\circ$ represents composition of quantum channels and $\mM_\theta$ and $\mN_\theta$ are arbitrary one-parameter quantum channels.

\section{Classification of qubit channels}
\label{app:qubit}

Any qubit channel $\mE(\cdot)$ can be represented through the Bloch sphere language as~\cite{king2001minimal} 
\begin{equation}
\mE\left(\frac{1}{2}(\id + \vw\cdot\vsig)\right) = \frac{1}{2}(\id + (\vt + T\vw)\cdot\vsig),
\end{equation}
where $\vw \in \bR^3$ is a vector inside the Bloch sphere (i.e., $\norm{\vw}^2 \leq 1$), $\vsig = \begin{pmatrix}X \\ Y \\ Z \end{pmatrix}$ is the vector of Pauli matrices and we use $\vw\cdot\vsig$ to represent the inner product between $\vw$ and $\vsig$, i.e., $\sum_{i=1}^3 v_i \sigma_i = v_1 X + v_2 Y + v_3 Z$. $\vt \in \bR^3$ and $T \in \bR^{3\times 3}$ uniquely determine the qubit channel $\mE$ and vise versa. When $\mE$ is a unital channel, i.e., $\mE(\id) = \id$, $\vt = 0$. 
In order for the channel to be positive, the output vector $\vt + T \vw$ must have norm at most 1, implying $\norm{T} \leq 1$. (However, not every $\vt$ and $T$ satisfying $\norm{\vt + T \vw} \leq 1$ for $\norm{\vw} \leq 1$ is a quantum channel.) 

There always exist unitaries $U$ and $V$ such that~\cite{king2001minimal} 
\begin{equation}
\label{eq:class-def}
V \mE\left(U\frac{1}{2}(\id + \vw\cdot\vsig)U^\dagger\right)V^\dagger = \frac{1}{2}(\id + (\vs + \Lambda\vw)\cdot\vsig),
\end{equation}
where $\Lambda$ is a diagonal matrix. When diagonal elements of $\Lambda$ are $\pm 1$, $\mE$ is a \textbf{unitary} channel. When $\norm{\Lambda} < 1$, we call $\mE$ a  \textbf{strictly contractive} channel, which is a class of channels that includes both depolarizing channels and amplitude damping channels. Otherwise, we call $\mE$ a \textbf{dephasing-class} channel. In this case, $\vs = 0$ and $\Lambda$ has one diagonal entry equal to $1$ and the other two smaller than $1$~\cite{king2001minimal}. We will specifically choose $U$ and $V$ such that the third diagonal element of $\Lambda$ is $1$, and then  
\begin{equation}
\label{eq:dephasing-class}
V \mE(U \rho U^\dagger) V^\dagger = (1-p) \rho + p Z \rho Z,
\end{equation}
which means $\mE$ is equivalent to a qubit dephasing channel up to unitary rotations. 

In channel estimation, we consider channels $\mE_\theta$ that are functions of an unknown parameter $\theta$ around a local point $\theta = 0$. We will assume (1) our one-parameter channels under consideration do not switch to other families in a neighborhood of the local point $\theta = 0$; (2) there exists a parametrization of  $U_\theta,V_\theta,\vs_\theta$ and $\Lambda_\theta$ in \eqref{eq:class-def} such that all of them are differentiable. Finally, without loss of generality, we assume $\trace(U_\theta^\dagger \partial_\theta U_\theta) = \trace(V_\theta^\dagger \partial_\theta  V_\theta) = 0$ because otherwise we can always multiply $U_\theta$ and $V_\theta$ by two $U(1)$ operators respectively which results in a different parametrization that both satisfies the requirement and does not affect the parametrization of the original quantum channel. 

We also state the formal definition of the ``rotation-generators-not-in-Kraus-span'' (RGNKS) condition here. Note that the RGNKS condition is a necessary but not sufficient condition of the HNKS condition. 

\begin{definition}[RGNKS]
    A one-parameter qubit channel satisfies the RGNKS condition if there exists a suitable parametrization (\eqref{eq:class-def}) for it such that either $H_0 \notin \mS$ or $H_1 \notin \mS$, where $H_0 := -i U_\theta^\dagger \partial_\theta U_\theta$ and $H_1 := -i V_\theta^\dagger \partial_\theta V_\theta$ are called the rotation generators of the qubit channel.
\end{definition}

Note that for unitary channels that are not parameter independent, the RGNKS condition holds true and for strictly contractive channels, the RGNKS condition is violated. For dephasing-class channels (\eqref{eq:dephasing-class}), both situations can arise and the RGNKS condition means at least of $H_0$ and $H_1$ is not in $\mS = {\rm span}\{\id,U Z U^\dagger\}$. 

Here we also explain why RGNKS is a necessary condition for HNKS but not vise versa.
For unitary channels, HNKS and RGNKS always hold true when $\dot\mE \neq 0$ and fails when $\dot\mE = 0$. For strictly contractive channels, the Kraus span $\mS$ is the entire matrix space and HNKS and RGNKS always fail. Therefore, we only need to show for dephasing-class channels HNKS $\subsetneq$ RGNKS. To see this, we first assume $\mE_\theta$ has the form in \eqref{eq:one-parameter-dephasing-qubit-channel} which is the most general form of dephasing channels up to constant unitary rotations. Then we can compute the Hamiltonian $H$ using Kraus operators $K_{\theta,0} = \sqrt{1-p_\theta} e^{-iG_0\theta}$ and $K_{\theta,1} = \sqrt{p_\theta} Z e^{-iG_1\theta}$ which gives $H = (1-p)G_0 + p G_1$ and $\mS = \{\id,Z\}$. On the other hand, RGNKS requires that at least one of $G_0,G_1$ is not in $\mS$, which is a necessary condition for HNKS. However, there are easily counter examples where RGNKS holds but HNKS fails. For example, $G_0 = p X$ and $G_1 = -(1-p)X$ leads to $H = 0$, but RGNKS still holds.

\section{HNKS condition: Qubit channels}
\label{app:HNKS}

Here we prove for qubit channels, the HNKS condition can be satisfied only when the channel is unitary or dephasing-class. We also prove a lemma for later use in \appref{app:dephasing}.

\subsection{Kraus representation in Pauli basis}

Consider a qubit channel $\mN(\rho)  = \sum_{i=1}^r K_i \rho K_i^\dagger$ and let $K_i = \sum_{j=0}^3 \sM_{ij} \sigma_j$ where $\sigma_0 = \id$. Then the CPTP condition translates to 
\begin{equation}
\label{eq:CPTP}
    \sum_{i} K_i^\dagger K_i 
    = \sum_i \left(\sum_j \sM_{ij}^* \sigma_j\right) \left(\sum_k \sM_{ik} \sigma_k\right) = \sum_{jk}\sigma_j \sigma_k \sum_i \sM_{ij}^* \sM_{ik} = \vvsig^\dagger (\sM^\dagger \sM) \vvsig = \id,
\end{equation}
where $\vvsig = \begin{pmatrix}
\id \\ X \\ Y \\ Z
\end{pmatrix}$ is the vector of Pauli matrices including $\id$ as the first element. Multiple $\sM$'s can represent the same quantum channel due to the freedom in choosing Kraus operators. In particular, when $\sM$ represents $\mN$, $\sU \sM$ represents $\mN$ as well, when $\sU$ is an isometry satisfying $\sU^\dagger \sU = 1$. Thus, without loss of generality, we can assume $\sM \in \bC^{4 \times 4}$ and 
    \begin{equation}
    \label{eq:canonical}
    \sM := \begin{pmatrix}
    m_{00} & \vm^\dagger \\
    0 & \frakm
    \end{pmatrix}, 
    \end{equation}
    where $m_{00} \geq 0$ and $\frakm \in \bC^{3 \times 3}$ is positive semi-definite. Note that $\sM$ is uniquely defined, i.e., $\sM$ is a function of $\mN$ (from the uniqueness of polar decomposition). We will call, in this appendix, the Kraus representation given by \eqref{eq:canonical} the canonical Kraus representation of $\mN(\cdot)$.  Plugging the following into \eqref{eq:CPTP}  
    \begin{equation}
    \sM^\dagger \sM = 
    \begin{pmatrix}
    m_{00} & 0 \\
    \vm & \frakm^\dagger
    \end{pmatrix}
    \begin{pmatrix}
    m_{00} & \vm^\dagger \\
    0 & \frakm
    \end{pmatrix}
    = 
    \begin{pmatrix}
    m_{00}^2 & m_{00}\vm^\dagger \\
    m_{00}\vm & \vm\vm^\dagger + \frakm^\dagger\frakm
    \end{pmatrix}, 
    \end{equation}
    we have 
    \begin{equation}
    \vvsig^\dagger (\sM^\dagger \sM) \vvsig = (m_{00}^2 + \vm^\dagger \vm)\id + 2m_{00}\Re[\vm] \cdot \vsig + 2 (\Re[\vm]\times\Im[\vm]) \cdot \vsig + 
    \vsig^\dagger (\frakm^\dagger \frakm) \vsig = \id, 
    \end{equation}
    where we use $\Re[\cdot]$ and $\Im[\cdot]$ to represent the real and imaginary parts of vectors. 
    If the channel is also unital, i.e., $\mN(\id) = \id$, we have 
    \begin{equation}
    \label{eq:unital}
    \sum_{i} K_i K_i^\dagger 
    = \sum_i \left(\sum_j \sM_{ij} \sigma_j\right) \left(\sum_k \sM_{ik}^* \sigma_k\right)  = \sum_{jk}\sigma_j \sigma_k \sum_i \sM_{ij} \sM_{ik}^* = \vvsig^\dagger (\sM^T \sM^*) \vvsig = \id. 
    \end{equation}    
Furthermore, let the diagonalization of $\frakm$ be 
    \begin{equation}
    \frakm = \sqrt{\gamma_1} \vv_1 \vv_1^\dagger + \sqrt{\gamma_2} \vv_2 \vv_2^\dagger + \sqrt{\gamma_3} \vv_3 \vv_3^\dagger,
    \end{equation}
    where $\vv_1$, $\vv_2$, $\vv_3$ are mutually orthogornal unit vectors and $\gamma_1 \geq \gamma_2 \geq \gamma_3 \geq 0$.  
The CPTP condition (\eqref{eq:CPTP}) translates to 
    \begin{gather}
    m_{00}^2 + \vm^\dagger \vm + \gamma_1 + \gamma_2 + \gamma_3 = 1,
\\
     m_{00} \Re[\vm] +  \Re[\vm]\times \Im[\vm] +  \gamma_1 \Re[\vv_1]\times \Im[\vv_1] +  \gamma_2 \Re[\vv_2]\times \Im[\vv_2] +  \gamma_3 \Re[\vv_3]\times \Im[\vv_3] = 0.\label{eq:CPTP-2}
    \end{gather}
    In addition, if the channel is unital, \eqref{eq:unital} becomes
    \begin{equation}
     m_{00} \Re[\vm] -  \Re[\vm]\times \Im[\vm] -  \gamma_1 \Re[\vv_1]\times \Im[\vv_1] -  \gamma_2 \Re[\vv_2]\times \Im[\vv_2] -  \gamma_3 \Re[\vv_3]\times \Im[\vv_3] = 0.
    \end{equation}
Then we have the following lemma:  
\begin{lemma}
\label{lemma:non-unital-criterion}
A qubit channel $\mN$ is unital if and only if $m_{00}\Re[\vm] = 0$. 
\end{lemma}

\subsection{Non-unital qubit channels: HNKS is violated}

Using the canonical Kraus representation of qubit channels, we now prove the following lemma for one-parameter qubit channels. 
\begin{lemma}
\label{lemma:non-unital-HKS}
Any non-unital qubit channel $\mN_\theta$ violates the HNKS condition.
\end{lemma} 
    
    \begin{proof} 
    Assume $\mN_\theta(\cdot) = \sum_{i=1}^4 K_{\theta,i}(\cdot) K_{\theta,i}^\dagger$ is a non-unital one-parameter qubit channel with canonical Kraus representation. $\sM,m_{00},\vm,\frakm,\gamma_i,\vv_i$ are defined for $\mN(\cdot)$ as in the previous section.
We will prove that for any Hermitian $H$, $\exists$ Hermitian $ h$, s.t. 
    \begin{equation}
    H + \sum_{ij} h_{ij} K_{i}^\dagger K_{j} = 0,
    \end{equation}
where
    \begin{equation}
    \sum_{ij} h_{ij} K_{i}^\dagger K_{j} = \sum_{ij} h_{ij} \left(\sum_l \sM_{il}^* \sigma_l\right) \left(\sum_k \sM_{jk} \sigma_k\right) = \sum_{lk} \sigma_l \sigma_k  \sum_{ij} \sM_{il}^* h_{ij} \sM_{jk} = \vvsig^\dagger (\sM^\dagger h \sM) \vvsig . 
    \end{equation}
Then the HNKS condition $i \sum_{i} K_{i}^\dagger\dot K_{i} + \sum_{ij} h_{ij} K_{i}^\dagger K_{j} = 0$ must be violated when we choose $H = i \sum_{i} K_{i}^\dagger\dot K_{i}$, the Hamiltonian of $\mN_\theta$. 

In the following, we will prove $\vvsig^\dagger (\sM^\dagger h \sM) \vvsig$ can represent any Hermitian operator in $\bC^{2\times 2}$. Note that $\vvsig^\dagger (\sM^\dagger \id \sM) \vvsig = \id$ from \eqref{eq:CPTP}. It will be sufficient to prove that the traceless part of $\vvsig^\dagger (\sM^\dagger h \sM) \vvsig$ can represent any traceless Hermitian operator in $\bC^{2\times 2}$. 
    Let 
    \begin{equation}
    h = \begin{pmatrix}
    0 & \vh^\dagger \\
    \vh & \frakh
    \end{pmatrix}.
    \end{equation}
        Then
    \begin{equation}
    \begin{split}
    \sM^\dagger h \sM 
    &= \begin{pmatrix}
    m_{00} & 0 \\
    \vm & \frakm^\dagger
    \end{pmatrix}
    \begin{pmatrix}
    0 & \vh^\dagger \\
    \vh & \frakh
    \end{pmatrix}
    \begin{pmatrix}
    m_{00} & \vm^\dagger \\
    0 & \frakm
    \end{pmatrix} 
    = 
    \begin{pmatrix}
    0 & m_{00}\vh^\dagger \frakm \\
    m_{00}\frakm^\dagger \vh  & \vm \vh^\dagger \frakm + \frakm^\dagger \vh \vm^\dagger + \frakm^\dagger \frakh\frakm
    \end{pmatrix}.
    \end{split}
    \end{equation}
Using the following property of Pauli decomposition  for any $\va,\vb \in \bC^{3}$, 
    \begin{equation}
    \begin{split}
    \vsig^\dagger (\va \vb^\dagger) \vsig 
    &= (\va\cdot\vsig)(\vb^*\cdot\vsig) = (\va\cdot\vb^*) \id + i (\va \times \vb^*) \cdot \vsig
    \\
    &= (\va\cdot\vb^*) \id + (\Re[\va] \times \Im[\vb] - \Im[\va] \times \Re[\vb]) \cdot \vsig + i(\Re[\va] \times \Re[\vb] - \Im[\va] \times \Im[\vb]) \cdot \vsig, 
    \end{split}
    \end{equation} 
    we have 
    \begin{equation}
    \begin{split}
    &\quad\; \vvsig^\dagger (\sM^\dagger h \sM ) \vvsig \\
    &= 2 ( m_{00}\Re[\frakm \vh]\cdot \vsig) + \left( 2 \Re[\vh^\dagger \frakm \vm ]\right)\id + i (\vm \times (\frakm \vh)^*) \cdot \vsig - i (\vm^*\times \frakm \vh ) \cdot \vsig + \vsig^\dagger (\frakm^\dagger \frakh\frakm) \vsig \\
    &= \left( 2 \Re[\vh^\dagger \frakm \vm ]\right)\id + \vsig^\dagger (\frakm^\dagger \frakh\frakm) \vsig 
    + 2 ( m_{00}\Re[\frakm \vh]\cdot \vsig) 
    + 2 ( (\Re[\vm]\times\Im[\frakm \vh])\cdot \vsig) 
    - 2 ( (\Im[\vm]\times\Re[\frakm \vh])\cdot \vsig) . 
    \end{split}
    \end{equation}
    Let $\vh = h_1\vv_1 + h_2 \vv_2 + h_3 \vv_3$ and $\frakh = g_1 \vv_1\vv_1^\dagger + g_2 \vv_2\vv_2^\dagger + g_3 \vv_3\vv_3^\dagger$. 
    A few additional calculations show that the traceless part of $\vvsig^\dagger (\sM^\dagger h \sM ) \vvsig$ is equal to 
    \begin{gather}
    \label{eq:traceless}
    \left( \sum_i \widetilde{\frakm}_i \begin{pmatrix}
    \Re[h_i]\sqrt{\gamma_i}\\
    \Im[h_i] \sqrt{\gamma_i}\\
    g_i  \gamma_i
    \end{pmatrix}\right) \cdot (2\vsig),  
    \end{gather}
    where $\widetilde{\frakm}_i \in \bC^{3\times 3}$ and 
    \begin{multline*}
    \widetilde{\frakm}_i := 
            \left( m_{00}\Re[\vv_i] + \Re[\vm] \times \Im[\vv_i] - \Im[\vm] \times \Re[\vv_i] , \right. \\ \left.  
            -m_{00}\Im[\vv_i] + \Re[\vm] \times \Re[\vv_i] + \Im[\vm] \times \Im[\vv_i] , \Re[\vv_i]\times \Im[\vv_i] \right). 
    \end{multline*}
    \eqref{eq:traceless} spans the entire traceless Hermitian matrix subspace if for some $i$,  
    \begin{equation}
    \gamma_i \det\left(\widetilde{\frakm}_i \right) \neq 0, 
    \end{equation}
    Let $\Re[\vv_i]\times \Im[\vv_i] = \norm{\Re[\vv_i]\times \Im[\vv_i]} \ve_i$ where $\ve_i$ is a unit vector, 
    \begin{equation}
    \det\left( \widetilde{\frakm}_i \right) = \norm{\Re[\vv_i]\times \Im[\vv_i]} \big(m_{00} (\Re[\vm]\cdot\ve_i) - \left(m_{00}^2 + (\Re[\vm]\cdot\ve_i)^2 + (\Im[\vm]\cdot\ve_i)^2 \right) \norm{\Re[\vv_i]\times \Im[\vv_i]}\big). 
    \end{equation}
    Note that the CPTP condition implies that (from $m_{00}\Re[\vm] \cdot$ \eqref{eq:CPTP-2})
    \begin{equation}
    \norm{m_{00} \Re[\vm]}^2 +  \sum_i \gamma_i (m_{00}\Re[\vm]\cdot\ve_i) \norm{\Re[\vv_i]\times \Im[\vv_i]} = 0. 
    \end{equation}
    According to \lemmaref{lemma:non-unital-criterion}, $\norm{m_{00} \Re[\vm]}^2 > 0$ for non-unital channels, there must be an $i$ such that $\gamma_i (m_{00}\Re[\vm]\cdot\ve_i) \norm{\Re[\vv_i]\times \Im[\vv_i]} < 0$ and then  
    \begin{equation}
    \gamma_i \det\left( \widetilde{\frakm}_i \right) \neq 0. 
    \end{equation}
    The claim is then proven.
\end{proof}

    We've shown above that any non-unital channel violate the HNKS condition. Given a non-unital qubit channel $\mN(\cdot) = \sum_i K_i (\cdot)K_i^\dagger$ in canonical form and $H$, there might be different choices of $h$ such that 
    \begin{equation}
        H + \sum_{ij} h_{ij} K_{i}^\dagger K_{j} = 0. 
    \end{equation}
    Here we pick a unique gauge for the proof of a lemma below. (Readers can skip to the next appendix for now and come back here later as the meaning of the lemma below will not be perfectly clear until it is applied in \appref{app:dephasing}.)  
    Using the procedure in the proof above, we can pick an $h$ of the form  
    \begin{equation}
    h = \trace\left(-H-\vvsig^\dagger \left(\sM^\dagger \begin{pmatrix}
        0 & \vh^\dagger \\ 
        \vh  & \frakh \\
    \end{pmatrix} \sM\right) \vvsig\right) \cdot \frac{\id}{2}  + \begin{pmatrix}
        0 & \vh^\dagger \\ 
        \vh  & \frakh \\
    \end{pmatrix}, 
    \end{equation}
    where $\vh = h_i \vv_i$ and $\frakh = g_i \vv_i\vv_i^\dagger$ for some $i$ such that $\gamma_i \det\left( \widetilde{\frakm}_i \right) \neq 0$. Specifically, $h_i$ and $g_i$ are the solution of 
    \begin{equation}
    \label{eq:h-solution}
    H - \frac{\trace(H)}{2} \id = \left(  \widetilde{\frakm}_i \begin{pmatrix}
    \Re[h_i]\sqrt{\gamma_i}\\
    \Im[h_i] \sqrt{\gamma_i}\\
    g_i  \gamma_i
    \end{pmatrix}\right) \cdot (2\vsig), 
    \end{equation}
    which is guaranteed to exist. There can be different $i$ that satisfies $\gamma_i \det\left( \widetilde{\frakm}_i \right) \neq 0$, then we pick one such that $\norm{h}$ is minimized. (If $\norm{h}$ are the same for more than one choices of $i$ that satisfies $\gamma_i \det\left( \widetilde{\frakm}_i \right) \neq 0$, we can uniquely choose $i$ using any pre-defined total order on $\bC^3$ that we won't specify here as because it will not affect our discussion below.) 
    Then $\norm{h^{\mathrm{cn}}(H,\mN)}$ becomes a well-defined function of $H$ and $\mN$ where $h^{\mathrm{cn}}$ is a canonical choice of $h$ picked from the above procedure. Note that $\norm{h^{\mathrm{cn}}(a H,\mN)} = a \norm{h^{\mathrm{cn}}(H,\mN)}$ for any $a > 0$. 
    
    Then we have the following lemma. (Readers can skip this lemma and come back to it later as its meaning will not be perfectly clear until it is applied in \appref{app:dephasing}.)

\begin{lemma}[Bounding {$\norm{\alpha}$} when {$\beta = 0$}]
Let $\mN_\theta$ be a non-unital qubit channel and $\{K_i\}_{i=1}^r$ one of its Kraus representations, such that $\norm{i \sum_{i} K_{i}^\dagger\dot K_{i}} \leq a$ for some Kraus representation $K_{i}$ and $\norm{\mN_\theta(\id) - \id} \geq b > 0$ for some constants $a,b$, then there exists a universal non-negative and finite function $\xi(a,b)$ of $(a,b)$, such that 
\begin{equation}
\label{eq:alpha-bound-non-unital}
\min_{h: i \sum_{i} K_{i}^\dagger\dot K_{i} + \sum_{ij} h_{ij} K_{i}^\dagger K_{j} = 0} \bigg\|{\sum_{i} \Big(\dot K_{i} - i\sum_j h_{ij}K_{j}\Big)^\dagger \Big(\dot K_{i} - i\sum_j h_{ij}K_{j}\Big)}\bigg\| \leq 2\bigg\|\sum_{i} \dot K_{i}^\dagger \dot K_{i}\bigg\|+ \xi(a,b),
\end{equation}
and $\xi(a,b) = a^2  \xi(1,b)$. 
\label{lemma:alpha-bound-non-unital}
\end{lemma}

\begin{proof}
    Since $\norm{\mN_\theta(\id) - \id} \geq b > 0$, $\mN_\theta$ must be a non-unital channel.  The HNKS condition must be violated.  Let $H = i \sum_{i} K_{i}^\dagger\dot K_{i} = i \sum_{i} K_{i}^{\mathrm{cn}\dagger}\dot K_{i}^{\mathrm{cn}}$. 
    We choose $h^{\mathrm{cn}} = h^{\mathrm{cn}}(H,\mN)$ as defined above, such that $H + \sum_{ij} h^{\mathrm{cn}}_{ij} K_{i}^{\mathrm{cn}\dagger} K_{j}^{\mathrm{cn}} = 0$, where $\{K_i^{\mathrm{cn}}\}_{i=1}^4$ is the canonical Kraus representation of $\mN$. (Here we implicitly assume the Kraus rank is $4$. In some cases, some canonical Kraus operators can be zero, and we should remove it and replace $4$ with the actual Kraus rank. It will not affect the discussion below.) Since both $\{K_i\}_{i=1}^r$ and $\{K_i^{\mathrm{cn}}\}_{i=1}^4$ represent the same channel, there is an isometry $\sU$ such that $K_i = \sum_i \sU_{ij} K_j^{\mathrm{cn}}$.  Let $h = \sU h^{\mathrm{cn}} \sU^\dagger$, we have $H + \sum_{ij} h^{\mathrm{cn}}_{ij} K_{i}^{\mathrm{cn}\dagger} K_{j}^{\mathrm{cn}} = H + \sum_{ij} h_{ij} K_{i}^{\dagger} K_{j} = 0$. 
    Let 
    \begin{equation}
    \vK = \begin{pmatrix}
        K_1 \\
        K_2 \\ 
        \vdots \\
        K_r
    \end{pmatrix}, \quad     
    \dot\vK = \begin{pmatrix}
        \dot K_1 \\
        \dot K_2 \\ 
        \vdots \\
        \dot K_r
    \end{pmatrix}. 
    \end{equation}
    Then 
    \begin{equation}
    \begin{split}
    &\quad \bigg\|{\sum_{i} \Big(\dot K_{i} - i\sum_j h_{ij}K_{j}\Big)^\dagger \Big(\dot K_{i} - i\sum_j h_{ij}K_{j}\Big)}\bigg\| \\
    & \leq \bigg\|{\sum_{i} \dot K_{i}^\dagger \dot K_{i}}\bigg\| + \bigg\|{ i \sum_{ij} h_{ij}^* K_j^\dagger \dot K_i - i \sum_{ij} h_{ij} \dot K_i^\dagger K_j + \sum_{ijj'} h_{ij}^* h_{ij'} K_j^\dagger K_{j'} }\bigg\|\\
    & \leq  \norm{\dot\vK^\dagger \dot\vK} + 2 \norm{\dot\vK}\norm{\vK}\norm{h} + \norm{h}^2 \norm{\vK}^2 = \norm{\dot\vK^\dagger \dot\vK} + 2 \sqrt{ \norm{\dot\vK^\dagger \dot\vK}}\norm{h^{\mathrm{cn}}} + \norm{h^{\mathrm{cn}}}^2 \\
    & \leq 2 \norm{\dot\vK^\dagger \dot\vK} + 2\norm{h^{\mathrm{cn}}}^2 
    \leq 2 \norm{\dot\vK^\dagger \dot\vK} + 2f(a,b)^2, 
    \end{split}
    \end{equation}
    where we use the submultiplicity of operator norm, $\norm{\vK} = \norm{\vK^\dagger} = \sqrt{\norm{\vK^\dagger\vK}}$ and $\norm{h^{\mathrm{cn}}} = \big\|{\sU h^{\mathrm{cn}} \sU^\dagger}\big\| = \norm{h}$. 
Here we use the following functions: 
    \begin{equation}
    f(a,b) := \sup_{\substack{H:\norm{H} \leq a,\\ \mN:\norm{\mN(\id)-\id} \geq b}} \norm{h^{\mathrm{cn}}(H,\mN)}. 
    \end{equation}
    \sloppy In order for our upper bound to be non-trivial, we must prove $f(a,b)$ is well-defined, i.e., $f(a,b) < \infty$ for any $a,b > 0$, and we prove this below. Define compact subsets $S_1:=\{H:\norm{H} \leq a\}$ and $S_2:=\{\mN:\norm{\mN(\id)-\id} \geq b\}$. Then for any $\mN_0 \in S_2$, there exists $\epsilon_{\mN_0} > 0$ such that $f(H,\mN)$ bounded by a constant $C(\mN_0)$ in an open subset $S_1^{\mN_0}\times B(\mN_0)$ where $S_1^{\mN_0} = \{H:\norm{H} < a + \epsilon_{\mN_0}\}$, $B(\mN_0) = \{\mN:d(\mN_0,\mN) < \epsilon_{\mN_0}\}$ and $d(\mN_0,\mN)$ can be e.g., the diamond distance between quantum channels that we use as the metric in the space of quantum channels~\cite{watrous2018theory}. Here $\epsilon_{\mN_0}$ needs only to be small enough such that for some $i$, $\gamma_i \det\left( \widetilde{\frakm}_i \right)\neq 0$ is satisfied for all $\mN \in B(\mN_0)$. It guarantees the continuity and then the boundness of $\norm{h(H,\mN)}$ in $S_1^{\mN_0}\times B(\mN_0)$ when $h$ is the solution of \eqref{eq:h-solution} for the specific $i$ chosen such that $\gamma_i \det\left( \widetilde{\frakm}_i \right)\neq 0$, and then $\norm{h^{\mathrm{cn}}(H,\mN)} \leq \norm{h(H,\mN)}$ must also be bounded. Note that $S_1^{\mN}\times B(\mN)$ with all $\mN \in S_2$ forms an open cover of $S_1 \times S_2$. The compactness of $S_1 \times S_2$ guarantees there exists a finite subcover of $S_1\times S_2$, i.e, there exists a finite number $\ell$ and a finite set $\{\mN_k\}_{k=1}^\ell$ such that $\bigcup_{k=1}^\ell S_1^{\mN_k} \times B(\mN_k) \supseteq S_1\times S_2$. It implies 
    \begin{equation}
        f(a,b) = \sup_{(H,\mN) \in S_1\times S_2} f(H,\mN) \leq \max_{k\in[\ell]} \sup_{(H,\mN) \in S_1^{\mN_k}\times B(\mN_k)} f(H,\mN) \leq \max_{k\in[\ell]} C(\mN_k) <\infty. 
    \end{equation}

    Clearly, $f(a,b) = a f(1,b)$ because $\norm{h^{\mathrm{cn}}(a H,\mN)} = a \norm{h^{\mathrm{cn}}(H,\mN)}$ for any $a > 0$. Then let 
    \begin{equation}
    \xi(a,b) := 2f(a,b)^2. 
    \end{equation}
    \eqref{eq:alpha-bound-non-unital} is proven. $\xi(a,b) = a^2 \xi(1,b)$ is also satisfied. 
\end{proof}

\subsection{Unital qubit channels: when HNKS is satisfied}

    For unital qubit channels, we prove the following necessary and sufficient condition of the HNKS condition. 
    \begin{theorem}[Single-qubit HNKS]
    \label{thm:single-qubit-HNKS}
    A one-parameter qubit channel $\mN_\theta$ satisfies the HNKS condition if and only if 
    \begin{enumerate}[(1)]
        \item There exist unitaries $U,V$ and $p \in [0,1/2]$, such that 
    \begin{equation}
    V \mN(U \rho U^\dagger) V^\dagger = (1-p) \rho + p Z \rho Z,
    \end{equation}
    That is, $\mN$ is either a unitary channel or a dephasing-class channel.
        \item The Hamiltonian satisfies $H = i \sum_{i} K_{i}^\dagger\dot K_{i} \notin \mS = {\rm span}\{\id,U Z U^\dagger\}$ . 
    \end{enumerate}
    \end{theorem}
    \begin{proof}
    The sufficiency of the above conditions is obvious by definition. Thus we only prove their necessity. Assume the HNKS condition is satisfied. Using \lemmaref{lemma:non-unital-HKS}, $\mN$ must be unital. For unital channels, $m_{00} \Re[\vm] = 0$. It implies, from \eqref{eq:CPTP-2}, that 
    \begin{equation}
    \Re[\vm]\times \Im[\vm] +  \gamma_1 \Re[\vv_1]\times \Im[\vv_1] +  \gamma_2 \Re[\vv_2]\times \Im[\vv_2] +  \gamma_3 \Re[\vv_3]\times \Im[\vv_3] = 0.
    \end{equation}
    Assume $\Re[\vm] = 0$ and $m_{00}\neq 0$. (Note that if $m_{00} = 0$, it is always possible to transform it into another quantum channel through unitary rotation before or after the channel so that $m_{00} \neq 0$.) 
    \begin{equation}
    \det\left( \widetilde{\frakm}_i \right) = -\norm{\Re[\vv_i]\times \Im[\vv_i]}^2 \big(m_{00}^2 + (\Im[\vm]\cdot\ve_i)^2\big). 
    \end{equation}
    In the case where the HNKS condition is satisfied, all determinants have to be zero and we have 
    \begin{equation}
    \Re[\vv_i]\times \Im[\vv_i] = 0,\quad \text{i.e.},\, \Re[\vv_i]\propto \Im[\vv_i] ,\quad \forall i~ \text{s.t.} ~\gamma_i > 0. 
    \end{equation}
    Without loss of generality, assume $\Im[\vv_i] = 0$ for all $i$ s.t. $\gamma_i > 0$. 
    Let $h = \begin{pmatrix}
    0 & \vh^\dagger \\
    \vh & \frakh
    \end{pmatrix}$ and assume $\vh = h_1 \vv_1 + h_2 \vv_2 + h_3 \vv_3$. 
    There are four possible situations: 
    \begin{enumerate}[(1)]
        \item Rank of $\frakm$ is 0. The Kraus span contains only the identity operator. $\mN$ is a unitary channel in this case. 
        \item Rank of $\frakm$ is 1. The Kraus span is spanned by the identity operator and 
        $\left(m_{00} \vv_i - \Im[\vm] \times \vv_i \right) \cdot \vsig$. $\mN$ is a dephasing-class channel in this case. 
        \item Rank of $\frakm$ is 2. Then $\gamma_{1,2} > 0$. If instead of choosing $\frakh$ to be diagonal in basis $\{\vv_1,\vv_2,\vv_3\}$ like before, we choose $\frakh = g_{12} \vv_1 \vv_2^\dagger + g_{12}^* \vv_2 \vv_1^\dagger$. Then the traceless part of $\vvsig^\dagger (\sM^\dagger h \sM) \vvsig$ is 
        \begin{equation} 
        \begin{pmatrix}
        m_{00} \vv_1 - \Im[\vm] \times \vv_1 & 
        m_{00} \vv_2 - \Im[\vm] \times \vv_2 & 
        \vv_1 \times \vv_2
        \end{pmatrix}
        \begin{pmatrix}
        \sqrt{\gamma_1}\Re[h_1]\\
        \sqrt{\gamma_2}\Re[h_2]\\
        - \sqrt{\gamma_1\gamma_2} \Im[g_{12}]
        \end{pmatrix} \cdot (2\vsig).
        \end{equation}
        The determinant of the first matrix term above is always positive (which is obvious by drawing a 3d coordinate system). Then $\vvsig^\dagger (\sM^\dagger h \sM) \vvsig$ spans the entire Hermitian matrix space. The HNKS condition must be violated and therefore this case is impossible. 
        \item Rank of $\frakm$ is 3. The HNKS condition must be violated and therefore this case is impossible, using the same argument in the rank-2 case. 
    \end{enumerate}
    \end{proof}
    
\section{Dephasing-class channels: Unachievability of the HL}
\label{app:dephasing}
   
    Here we consider the metrological limit in estimating dephasing-class channels. We will first consider the ancilla-free sequential strategies with unital controls (strategies (c) in \figref{fig:strategies}). Then we will generalize to strategies (d) (with CPTP controls) and (e) (with bounded ancilla). In these strategies, $n$ copies of $\mE_\theta$ are sequentially applied on the quantum device and quantum channels $\{\mC_k\}_{k=1}^n$ are applied after each $\mE_\theta$ as quantum controls. Arbitrarily initial states and measurements are allowed. 
    
    The entire quantum channel under consideration (without ancilla) is  
    \begin{equation}
    \label{eq:entire-channel}
    \mE^{(n)}_\theta = \mC_n \circ \mE_{\theta} \circ \mC_{n-1} \circ \mE_{\theta} \circ \cdots \circ \mC_{1} \circ \mE_{\theta}. 
    \end{equation}
    We will use $\mE^{(k)}_\theta$ to represent the first $k$ step evolution for all $0\leq k \leq n$ and $\mE^{(0)}_\theta$ is the identity channel. Our goal is to obtain the optimal scaling of the channel QFI $F(\mE^{(n)}_\theta)$ when arbitrary controls (or unital controls) $\{\mC_k\}_{k=1}^n$ are allowed.

\subsection{General upper bound on channel QFI (Channel extension method)}
\label{app:sequential-upper}

    According to~\eqref{eq:channel-QFI} and \eqref{eq:sdp}, we have 
    \begin{gather}
    F(\mE^{(n)}_\theta) = \sup_{\psi=\ket{\psi}\bra{\psi}} \min_{\mE^{(n)}_\theta(\cdot) = \sum_i \tK^{(n)}_{\theta,i} (\cdot)\tK^{(n)\dagger}_{\theta,i}} 4 \trace\Big(\psi \sum_i \dot \tK^{(n)\dagger}_{i} \dot \tK^{(n)}_{i}\Big),\\
    F(\mE^{(n)}_\theta \otimes \id ) = \sup_{\rho} \min_{\mE^{(n)}_\theta(\cdot) = \sum_i \tK^{(n)}_{\theta,i} (\cdot)\tK^{(n)\dagger}_{\theta,i}} 4 \trace\Big(\rho \sum_i \dot \tK^{(n)\dagger}_{i} \dot \tK^{(n)}_{i}\Big),
    \end{gather}
    where $\tK^{(n)}_{i}$ is any Kraus representation of $\mE^{(n)}$. Since $\ket{\psi}\bra{\psi} \leq \id$ (we use $A \leq B$ to indicate $B - A$ is positive semi-definite), we have the following upper bound: 
    \begin{equation}
    \label{eq:upper-channel}
    F(\mE^{(n)}_\theta) \leq F(\mE^{(n)}_\theta \otimes \id) \leq 4 \trace\Big(\sum_i \dot \tK^{(n)\dagger}_{i} \dot \tK^{(n)}_{i}\Big),
    \end{equation}
    where $\tK^{(n)}_{\theta,i}$ can be any Kraus representation of $\mE^{(n)}_\theta$. By choosing a proper Kraus representation, the following Lemma that involves an expression of an upper bound on $F(\mE^{(n)}_\theta)$ can be derived, which later allows us to further obtain non-trivial upper bounds on $F(\mE^{(n)}_\theta)$ for dephasing-class channels. 
    \begin{lemma}[Refined channel extension method]
    \label{lemma:channel-extension}
    For any quantum channel to be estimated and any corresponding quantum controls, the channel QFI of \eqref{eq:entire-channel} satisfies 
    \begin{equation}
        F(\mE^{(n)}_\theta) \leq F(\mE^{(n)}_\theta \otimes \id) \leq \sum_{k=1}^{n} 4 \trace(\iota_{k-1} \alpha_k) + \sum_{k=1}^{n-1} 8\, \trace(\ugamma_k \beta_{k+1}),
    \end{equation}
    where the relevant Hermitian operators used here are defined by 
    \begin{gather}
    \alpha_k = \sum_{i_k} \dot \tK_{i_k}^\dagger \dot \tK_{i_k},\quad \beta_k = i \sum_{i_k}  \tK_{i_k}^\dagger \dot \tK_{i_k},\quad 
    \ugamma_{k} = \mC_{k}\circ\mE_\theta(\ugamma_{k-1}) + \ubeta_{k},\quad \ugamma_0 = 0,\label{eq:alpha-beta}\\
    \ubeta_k = \frac{1}{2} \bigg( i \sum_{i_k} \dot\tK_{i_k}  \iota_{k-1} \tK_{i_k}^\dagger - i \sum_{i_k} \tK_{i_k}  \iota_{k-1}  \dot\tK_{i_k}^\dagger\bigg),\label{eq:ubeta}\\
    \iota_{k-1} = \mE^{(k-1)}_\theta(\id) = \mC_{k-1}\circ\mE_\theta \circ \cdots \mC_{1}\circ \mE_\theta(\id),
    \end{gather}
    for all $ 1 \leq k \leq n$. 
    Here $\{\tilde K_{\theta,i_k}\}_{i_k=1}^{r_k}$ can be any Kraus representation of quantum channel $\mC_k\circ\mE_\theta$. (Note that although the Kraus operators for different $\mC_k$ are not necessarily the same, we slightly abuse the notation and only distinguish them using the subscript in $_k$ in the indices $i_k$ of Kraus operators for simplicity of notation.)
    \end{lemma}

    \begin{proof}
    An obvious choice of the Kraus representation of the entire quantum channel $\mE^{(n)}_\theta$ is given by 
    \begin{equation}
    \tK^{(n)}_{\theta,\vi} = \tK_{\theta,i_n} \cdots \tK_{\theta,i_1},
    \end{equation}
    where $\{\tK_{\theta,i_k}\}_{i_k=1}^{r_k}$ is a Kraus representation of $\mC_k \circ \mE_\theta$ and $\vi = (i_1,i_2,\ldots,i_n) \in [r_1] \times [r_2] \times \cdots \times [r_n]$. Its derivative is 
    \begin{equation}
    \begin{split}
    \dot\tK^{(n)}_{\vi} 
    = \dot\tK_{i_n} \cdots \tK_{i_1} 
    + \tK_{i_n} \dot\tK_{i_{n-1}}  \cdots \tK_{i_1} 
    + \cdots 
    + \tK_{i_n} \cdots \dot\tK_{i_1}
    = \dot\tK_{i_n} \tK^{(n-1)}_{\vi} + \tK_{i_n} \dot \tK^{(n-1)}_{\vi}, 
    \end{split}
    \end{equation}
    where $\tK^{(k)}_{\vi} := K_{\theta,i_k} \cdots K_{\theta,i_1}$ is the Kraus operator of $\mE^{(k)}$. Let $\alpha^{(n)} = \sum_\vi \dot \tK^{(n)\dagger}_{\vi} \dot \tK^{(n)}_{\vi}$, we then have 
    \begin{equation}
    \alpha^{(n)} = \alpha^{(n-1)} + \sum_{i_1\cdots i_{n-1}}  \tK^{(n-1)\dagger}_{\vi} \alpha_n \tK^{(n-1)}_{\vi} +  \tK^{(n-1)\dagger}_{\vi} i\beta_n \dot\tK^{(n-1)}_{\vi} + \dot\tK^{(n-1)\dagger}_{\vi} (-i\beta_n) \tK^{(n-1)}_{\vi},
    \end{equation}
    where $\alpha_n = \sum_{i_n} \dot \tK_{i_n}^\dagger \dot \tK_{i_n}$, $\beta_n = i \sum_{i_n}  \tK_{i_n}^\dagger \dot \tK_{i_n} = - i\sum_{i_n}  \dot\tK_{i_n}^\dagger \tK_{i_n}$. Continuing the iteration, we have 
    \begin{equation}
    \alpha^{(n)} = \sum_{k=1}^n \bigg(\sum_{\vi}  \tK^{(k-1)\dagger}_{\vi} \alpha_{k} \tK^{(k-1)}_{\vi}\bigg) + \sum_{k=1}^n \bigg( \sum_{\vi} \tK^{(k-1)\dagger}_{\vi} i\beta_k \dot\tK^{(k-1)}_{\vi} +  \dot\tK^{(k-1)\dagger}_{\vi} (-i\beta_k) \tK^{(k-1)}_{\vi} \bigg). 
    \end{equation}
    \eqref{eq:upper-channel} implies that 
    \begin{equation}
    \begin{split}
    \frac{F(\mE_\theta^{(n)} \otimes \id)}{4}\leq \trace(\alpha^{(n)}) 
    &= \sum_{k=1}^n \trace\bigg(\sum_{\vi}  \tK^{(k-1)}_{\vi} \tK^{(k-1)\dagger}_{\vi} \alpha_{k} \bigg)  + \sum_{k=1}^n \trace\bigg( \sum_{\vi} \left( i\dot\tK^{(k-1)}_{\vi}\tK^{(k-1)\dagger}_{\vi} - i\tK^{(k-1)}_{\vi} \dot\tK^{(k-1)\dagger}_{\vi} \right) \beta_k \bigg). 
    \end{split}
    \end{equation} 
    To further simplify, we use 
    \begin{equation}
    \begin{split}
    & \quad\; \frac{1}{2}\bigg( i \sum_\vi \dot\tK^{(k-1)}_{\vi}\tK^{(k-1)\dagger}_{\vi} - i \sum_\vi \tK^{(k-1)}_{\vi}\dot\tK^{(k-1)\dagger}_{\vi} \bigg)\\
    &    =  \sum_\vi   (\tK_{i_{k-1}} \cdots \tK_{i_2})   (\ubeta_1 ) ( \tK_{i_{k-1}} \cdots  \tK_{i_2}) ^\dagger  +  (\tK_{i_{k-1}} \cdots  \tK_{i_3})   (\ubeta_2 ) (\tK_{i_{k-1}} \cdots \tK_{i_3}) ^\dagger + \cdots +  \ubeta_{k-1} \\
    & =    \mC_{k-1}\circ\mE_\theta \circ\cdots \circ \mC_{2}\circ\mE_\theta(\ubeta_1)  + \mC_{k-1}\circ\mE_\theta \circ\cdots \mC_{3}\circ\mE_\theta (\ubeta_2) + \cdots +  \ubeta_{k-1}. 
    \end{split}
    \end{equation}
    Then we have 
    \begin{equation}
    \trace(\alpha^{(n)}) =  \sum_{k=1}^{n} \trace(\iota_{k-1} \alpha_k) + \sum_{k=1}^{n-1} 2\, \trace(\ugamma_k \beta_{k+1}),  
    \end{equation}
    using the definitions of $\ubeta_k$ and $\ugamma_k$. 
    \end{proof}
\subsection{The most general form of one-parameter dephasing-class channels}

\lemmaref{lemma:channel-extension} applies to any general quantum channel $\mE_\theta$. Here we focus on dephasing-class qubit channels. 

According to \thmref{thm:single-qubit-HNKS}, the HNKS condition is satisfied only when $\mE_\theta$ is a dephasing-class channel (except for unitary channels).  That means, there exists $U$ and $V$ such that 
\begin{equation}
V \mE(U\rho U^\dagger) V^\dagger = (1-p)  \rho + p  Z  \rho  Z, 
\end{equation}
for some $p \in (1,1/2]$. 
In the metrological strategies we consider in this Letter, arbitrary single-qubit unitary controls are always allowed in the system. Therefore, we can always apply unitary rotations such that $\mE_\theta$ is a dephasing channel at $\theta = 0$. Without loss of generality, we can assume $\mE_\theta$ depends on $\theta$ in the following way: 
\begin{equation}
\label{eq:one-parameter-dephasing-qubit-channel}
\mE_\theta(\rho) = (1-p_\theta) e^{-i G_0 \theta} \rho e^{i G_0 \theta} + p_\theta  Z e^{-i G_1 \theta} \rho e^{-i G_1 \theta} Z, 
\end{equation}
where $p_\theta \in (0,1/2]$ is differentiable and $G_0$, $G_1$ are traceless Hermitian operators independent of $\theta$.

Below we explain why \eqref{eq:one-parameter-dephasing-qubit-channel} is sufficient to represent any one-parameter qubit dephasing-class channel that we consider up to unitary rotations. We note that any QFI $F(\rho_\theta)$ is completely determined by $\rho_\theta$ and its derivative $\dot\rho_\theta$ at its true value ($\theta = 0$). Therefore, it is sufficient to show the derivative of \eqref{eq:one-parameter-dephasing-qubit-channel} is has the most general form. To see this, remember that in \appref{app:qubit}, we assumed any dephasing-class channel $\mE_\theta$, \emph{up to unitary rotations}, can be written as 
\begin{equation}
\label{eq:dephasing-1}
    V_\theta\mE_\theta(U_\theta \rho U_\theta^\dagger)V_\theta^\dagger = (1 - \tp_\theta) \rho + \tp_\theta Z \rho Z,
\end{equation}
where we can assume $V_\theta,U_\theta$ and $\tp_\theta$ are differentiable and $V_\theta=U_\theta=\id$ at $\theta = 0$. Moreover, $H_0 = -i U_\theta^\dagger \partial_\theta U_\theta$ and $H_1 = -i V_\theta^\dagger \partial_\theta V_\theta$ are two traceless Hermitian operators. We assert that the derivative of \eqref{eq:one-parameter-dephasing-qubit-channel} is the same as the derivative of $\mE_\theta$ in \eqref{eq:dephasing-1} when 
\begin{equation}
\label{eq:dephasing-2}
    p_\theta = \tp_\theta, \quad G_0 = H_1 + H_0,\quad G_1 = Z H_1 Z + H_0. 
\end{equation}
This can be verified straightforwardly, noting that $\mE_\theta$ in \eqref{eq:dephasing-1} can be represented using Kraus operators $\tK_{\theta,0} = \sqrt{1-p_\theta} V_\theta^\dagger U_\theta^\dagger$ and $\tK_{\theta,1} = \sqrt{p_\theta}  V_\theta^\dagger Z U_\theta^\dagger$. Clearly, at $\theta = 0$, they satisfy 
\begin{gather}
\label{eq:dephasing-3}
\tK_0 = \sqrt{1-\tp}\id,\quad \tK_1 = \sqrt{\tp}  Z,\\ 
\label{eq:dephasing-4}
    \dot \tK_0 = -i\sqrt{1-\tp}(H_0 + H_1) - \frac{\dot \tp}{2\sqrt{1-\tp}} \id,\quad \dot \tK_1 = -i\sqrt{\tp}(H_1 Z + Z H_0) + \frac{\dot \tp}{2\sqrt{\tp}} Z.   
\end{gather}
On the other hand, from \eqref{eq:one-parameter-dephasing-qubit-channel}, $\mE_\theta$ can be represented using Kraus operators $K_{\theta,0} = \sqrt{1-p_\theta} e^{-iG_0\theta}$ and $K_{\theta,1} = \sqrt{p_\theta} Z e^{-iG_1\theta}$. They satisfies 
\begin{gather}
K_0 = \sqrt{1-p} \id,\quad K_1 = \sqrt{p} Z,\\ 
    \dot K_0 = -i\sqrt{1-p}G_0 - \frac{\dot p}{2\sqrt{1-p}} \id,\quad \dot K_1 = -i\sqrt{p}ZG_1 + \frac{\dot p}{2\sqrt{p}} Z,  
\end{gather}
which are clearly identical to \eqref{eq:dephasing-3} and \eqref{eq:dephasing-4} if \eqref{eq:dephasing-2} holds, proving the generality of \eqref{eq:one-parameter-dephasing-qubit-channel}. 
Finally, we remark that there are can multiple sets of $(G_0,G_1)$ representing the same one-parameter quantum channel around a local point $\theta = 0$. But we will not specify a choice. Additionally, we note that the condition that $\exists G_0,G_1$ representing $\mE_\theta$ such that at least one of $G_0,G_1 \not\propto Z$ is equivalent to the RGNKS condition, i.e., $\exists H_0,H_1$ representing $\mE_\theta$ such that at least one of $H_0,H_1 \not\propto Z$.

Below we explain the degree of freedom in choosing the Kraus representation of \eqref{eq:one-parameter-dephasing-qubit-channel}. As mentioned before, one natural choice of the Kraus operators for \eqref{eq:one-parameter-dephasing-qubit-channel} are 
\begin{equation}
\label{eq:natural-kraus}
K_{\theta,0} = \sqrt{1-p_\theta}e^{-iG_0 \theta},\quad K_{\theta,1} = \sqrt{p_\theta}Ze^{-iG_1 \theta}. 
\end{equation}
In general (see e.g., discussions below \eqref{eq:purification-channel} and \eqref{eq:purification-channel-2}), we shall choose 
\begin{equation}
\begin{pmatrix}
  \tK_{\theta,0}\\
  \tK_{\theta,1}
\end{pmatrix}
= (e^{-i h \theta} \otimes \id)
\begin{pmatrix}
  K_{\theta,0}\\
  K_{\theta,1}
\end{pmatrix},
\end{equation}
where $h$ is an arbitrary Hermitian matrix in $\bC^{2\times 2}$ that represents the degree of freedom in choosing Kraus operators. Then the Kraus operators and their derivatives are 
\begin{gather}
\label{eq:tk01} \tK_{0} = \sqrt{1-p}\id,\quad \tK_{1} = \sqrt{p}Z,\\
\label{eq:tk0d} i \dot \tK_{0} = i \dot K_{0} + h_{00} K_0 + h_{01} K_1 = \frac{-i\dot p}{2\sqrt{1-p}}\id + \sqrt{1-p} G_0 + h_{00} \sqrt{1-p}\id + h_{01}\sqrt{p}Z,\\
\label{eq:tk1d} i \dot \tK_{1} = i \dot K_{1} + h_{10} K_0 + h_{11} K_1  = \frac{i\dot p}{2\sqrt{p}}Z + \sqrt{p} Z G_1 + h_{10} \sqrt{1-p}\id + h_{11}\sqrt{p}Z. 
\end{gather}
For some fixed $k$ and taking $\mC_k = \id$, the operators $\alpha_k, \beta_k$ and $\ubeta_k$ that appear in \eqref{eq:alpha-beta} and \eqref{eq:ubeta} in \lemmaref{lemma:channel-extension} can be calculated as follows, where $h \in \bC^{2\times 2}$ can be chosen as any Hermitian matrix for our convenience.
{\small
\begin{align}
    \alpha_k|_{\mC_k = \id}
    &= \sum_{i} \dot \tK_{i_k}^\dagger \dot \tK_{i_k} = 2 (1-p)h_{00}G_0 + 2 p h_{11}G_1 + 2\sqrt{p(1-p)}(h_{00}+h_{11})\Re[h_{01}]  Z \nonumber\\
     + & \left(\frac{1-p}{2}\trace(G_0^2) + \frac{p}{2}\trace(G_1^2) + (1-p)h_{00}^2 + p h_{11}^2 + \abs{h_{01}}^2 + \Re[h_{01}] \sqrt{p(1-p)} \trace((G_0+G_1)Z) + \frac{\dot p^2}{4p(1-p)}\right) \id \nonumber\\
     + & \sqrt{p(1-p)} \Im[h_{01}]\left( - \trace((G_0 -G_1)Y) X + \trace((G_0 -G_1)X) Y + \frac{-\dot p}{p(1-p)}Z\right),\label{eq:alpha-general}
\\
    \beta_k|_{\mC_k = \id} &= i \sum_{i} \tK_{i_k}^\dagger\dot \tK_{i_k}
    = (1-p)G_0 + p G_1 + ((1-p)h_{00} + p h_{11}) \id +  2\sqrt{p(1-p)}\Re[h_{10}] Z
    \label{eq:beta-general}
\\
    \ubeta_k |_{\mC_k = \id}
    &= \frac{1}{2} \sum_{i_k} i\dot \tK_{i_k} \iota_{k-1}  \tK_{i_k}^\dagger - i \tK_{i_k} \iota_{k-1} \dot \tK_{i_k} ^\dagger, \nonumber\\
    &= \frac{1-p}{2}\{\iota_{k-1},G_0\} + \frac{p}{2}Z\{\iota_{k-1},G_1\}Z + (1-p) h_{00} \iota_{k-1} + p h_{11}  Z \iota_{k-1} Z + \sqrt{p(1-p)}\Re[h_{10}] \{\iota_{k-1},Z\}.
    \label{eq:ubeta-general}
\end{align}}In general, for any unital channel $\mC_k$, we will take the Kraus operators for $\mC_k \circ \mE_k$ to be $\tK_{(a,b)_k} = C_{a_k}  \tK_{b_k}$ where $\{C_{a_k}\}_a$ are Kraus operators of $\mC_k$ and $\{\tK_{b_k}\}_b$ are Kraus operators of $\mE_k$ for some specific choices of $h$. As a result, we have 
\begin{equation}
    \alpha_k = \alpha_k|_{\mC_k = \id}, \quad \beta_k = \beta_k|_{\mC_k = \id}, \quad \ubeta_k = \mC_k(\ubeta_k|_{\mC_k = \id}). 
\end{equation}
For simplicity, we also define the following notations for later use: 
\begin{equation}
G_{+} := (1-p) G_0 + p G_1,
\quad 
G_{-} := (1-p) G_0 - p G_1,
\end{equation}
Also, note that $h$ can in principle be different for each $k$ but we omit the subscript $k$ for simplicity.

\subsection{Unital controls}
\label{app:dephasing-unital}

We first consider the situation where $\mC_k$ are unital controls (Note that single-qubit unital channels are equivalent to single-qubit mixed unitary channels~\cite{watrous2018theory}) and we will extend the result to non-unital controls later. In this section, $\mC_k$ are all assumed to be unital and denoted by $\mU_k$. For unital controls $\{\mU_k\}_{k=1}^n$, we have for $1 \leq k \leq n$, 
\begin{equation}
    \iota_{k-1} = \mE^{(k-1)}_\theta(\id) = \mU_{k-1}\circ\mE_\theta \circ \cdots \mU_{1}\circ \mE_\theta(\id) = \id,
\end{equation}
because dephasing-class channels $\mE_\theta$ are also unital. 

We first calculate $\alpha_k, \beta_k$ and $\ubeta_k$. Let 
\begin{equation}
\label{eq:h-choice}
h_{00} = h_{11} = 0 \text{~~and~~} 
    h_{01} = h_{10} = -\frac{(1-p)\trace(G_0 Z) + p \trace(G_1 Z)}{4\sqrt{p(1-p)}} 
\end{equation}
in \eqsref{eq:alpha-general}{eq:ubeta-general}. Using $\iota_{k-1} = \id$, we have 
\small 
\begin{align}
    \alpha_k 
    &= \bigg(\!\frac{1-p}{2}\trace(G_0^2) + \frac{p}{2}\trace(G_1^2) + \frac{(1-p)(1-4p)}{16p}\trace(G_0Z)^2 - \frac{1}{8} \trace(G_0Z)\trace(G_1Z) - \frac{(3-4p)p}{16(1-p)}\trace(G_1Z)^2 + \frac{\dot p^2}{4p(1-p)}\!\bigg) \id,
    \label{eq:alpha-dephasing-unital}
    \\
    \beta_k &= \frac{\trace(G_+ X)}{2} X  + \frac{\trace(G_+ Y)}{2} Y,
    \label{eq:beta-dephasing-unital}
\\
    \ubeta_k 
    &=  \frac{\trace(G_- X)}{2} \mU_k(X)  + \frac{\trace(G_- Y)}{2} \mU_k(Y).
    \label{eq:ubeta-dephasing-unital}
\end{align}
\normalsize
Note that the Hamiltonian of our dephasing channel (\eqref{eq:one-parameter-dephasing-qubit-channel}) is $H = G_+$ (using Kraus representation in \eqref{eq:natural-kraus}), the theorem below is trivially true when $G_+ \in {\rm span}\{\id,Z\}$ from previous results~\cite{zhou2021asymptotic}. Therefore, we will focus only on the case where $G_+ \notin {\rm span}\{\id,Z\}$ below. This implies $\trace(\beta_k^2) > 0$ for all $k$.  
Now we are ready to prove the following theorem. 
    \begin{theorem}[QFI upper bound for estimating dephasing-class channels using ancilla-free sequential strategies with unital controls]  
\label{thm:dephasing-unital}
    Consider estimating an unknown parameter $\theta$ in a dephasing-class channel $\mE_\theta$ of the most general form in \eqref{eq:one-parameter-dephasing-qubit-channel} using the ancilla-free sequential strategy and assume all controls $\mU_k$ are unital, the channel QFI of \eqref{eq:entire-channel} satisfies 
    \begin{equation}
    F(\mE_\theta^{(n)}) \leq F(\mE^{(n)}_\theta \otimes \id) \leq \mu_1 n,  
    \end{equation}
    where $\mu_1$ is a constant independent of $n$ and $\mu_1 = O(1/p_\theta)$ for small $p_\theta$.
    \end{theorem}

    \begin{proof}
    According to \lemmaref{lemma:channel-extension} and the discussion above, we only need to show that 
    \begin{equation}
        \sum_{k=1}^{n} 4\trace(\iota_{k-1} \alpha_k) + \sum_{k=1}^{n-1} 8\, \trace(\ugamma_k \beta_{k+1}) \leq \mu_1 n,
    \end{equation} 
    where $\iota_{k-1} = \id$, $\alpha_k$ and $\beta_k$ are chosen as in \eqref{eq:alpha-dephasing-unital} and \eqref{eq:beta-dephasing-unital}, and 
    \begin{equation}
    \ugamma_0 = 0,\quad \ugamma_{k} = \mU_{k}\circ\mE_\theta(\ugamma_{k-1}) + \ubeta_{k},
    \end{equation}
    where $\ubeta_{k}$ is chosen as in \eqref{eq:ubeta-dephasing-unital}. 
    Note that for all $1 \leq k \leq n$, 
    \begin{multline}
    4\trace(\iota_{k-1} \alpha_k) = 4\trace(\alpha_k) = 4(1-p)\trace(G_0^2) + 4p\trace(G_1^2) + 
    \\ \frac{(1-p)(1-4p)}{2p}\trace(G_0Z)^2 - \trace(G_0Z)\trace(G_1Z) - \frac{(3-4p)p}{2(1-p)}\trace(G_1Z)^2 + \frac{2\dot p^2}{p(1-p)}. 
    \end{multline}
    We only need to prove that 
    \begin{equation}
    \sum_{k=1}^{n-1} \, \trace(\ugamma_k \beta_{k+1}) 
    \leq \mu n .
    \end{equation}
    for some constant $\mu = O(1/p)$. 
    In fact, we will prove that 
    \begin{equation}
    \label{eq:upper}
    \sum_{k=1}^{n-1} \, \abs{\trace(\ugamma_k \beta_{k+1})} 
    \leq \mu n .
    \end{equation}

We assume $n \geq 2$. Then the average of $|\trace(\ugamma_k\beta_{k+1})|$ satisfies 
\begin{equation}
\label{eq:cond-2}
    f(n) := \frac{\sum_{k=1}^{n-1} \abs{\trace(\ugamma_k \beta_{k+1})}}{n-1} >
    \frac{(1-2p)b\abs{\ub_1}}{2p(1-p)} (1+\delta), 
\end{equation}
where $b>0$ and $\ub_1$ are real $O(1)$ constants depending on $\mE_\theta$ that will be specified later, and $\delta = 1$. (Note that choosing $\delta$ to be a different number will improve our upper bound by a constant factor but we will not perform the calculation here as it will not be informative).  If \eqref{eq:cond-2} is not correct, we have proven \eqref{eq:upper}. 
Below we will consider only the case where \eqref{eq:cond-2} holds. 
We consider terms $\abs{\trace(\ugamma_k \beta_{k+1})}$ that is above $1/(1+\delta)$ of the average and define the average of these terms to be
\begin{equation}
g(n) := \frac{1}{n_0} \sum_{l: \abs{\trace(\ugamma_l \beta_{l+1})} > \frac{f(n)}{1+\delta}, \,l \leq n-1} \abs{\trace(\ugamma_l\beta_{l+1})},
\end{equation}
where $n_0 = \abs{\left\{l: \abs{\trace(\ugamma_l \beta_{l+1})} > \frac{f(n)}{1+\delta}, \,l \leq n-1\right\}}$ is the number of terms that is above $1/(1+\delta)$ of the average. 
Clearly, 
\begin{equation}
(n-1-n_0) \cdot \frac{f(n)}{1+\delta} + n_0 g(n) \geq f(n)(n-1)
~~\Rightarrow~~ 
n_0 \geq \frac{\delta (n-1)f(n)}{(1+\delta) g(n) - f(n)}. \label{eq:n0-lower}
\end{equation}

Note that $\trace(\ugamma_k) = \trace(\ubeta_k) = \trace(\beta_k) = \trace(\ubeta_k Z) = \trace(\beta_k Z) = 0$ for all $k$. For all $k$, $\beta_k$ are the same and we assume $\beta_k = b \tilde{X}$, where $b = \sqrt{\trace(\beta_k^2)}$ and $\tilde{X} = \frac{\beta_k}{\sqrt{\trace(\beta_k^2})}$. Let $\tilde{Y} = \frac{1}{2i}[\tilde{X},Z]$. Assume $\ubeta_k|_{\mU_k = \id} = \ub_1 \tilde{X} + \ub_2 \tilde{Y}$, $\ugamma_k = r_{1,k} \tilde{X} + r_{2,k} \tilde{Y} + r_{3,k} Z$. Then $r_{1,k} = \frac{1}{2}\trace(\ugamma_k \tilde{X}) = \frac{1}{2b}\trace(\ugamma_k\beta_{k+1})$. We have 
\begin{equation}
\frac{1}{4}\norm{\ugamma_k}_1^2 =  \frac{1}{2}\trace(\ugamma_k^2) = r_{1,k}^2 + r_{2,k}^2 + r_{3,k}^2,
\end{equation}
where $\norm{\cdot}_1$ is the trace norm, and 
\begin{equation}
\ugamma_{k+1} = \mU_{k+1}\circ\mE_\theta(\ugamma_k) + \ubeta_{k+1} = \mU_{k+1}\left(\left((1-2p) r_{1,k} + \ub_1\right) \tilde{X} + \left((1-2p) r_{2,k} + {\ub_2}\right) \tilde{Y} + r_{3,k} Z \right),
\end{equation}
Using the data-processing inequality for trace norm, we have 
\begin{align}
\frac{1}{4}\norm{\ugamma_{k+1}}_1^2 \leq \frac{1}{4}\norm{\ugamma_{k+1}}_1^2\big|_{\mU_{k+1} = \id}  
&= \frac{1}{2}\trace(\ugamma_{k+1}^2)\big|_{\mU_{k+1} = \id} 
=  ((1-2p)r_{1,k} + \ub_1)^2 + ((1-2p)r_{2,k} + \ub_2)^2 + r_{3,k}^2 
\\&\leq \frac{1}{4}\norm{\ugamma_{k}}_1^2 + ((1-2p)|r_{1,k}| + \abs{\ub_1})^2 - r_{1,k}^2 +  ((1-2p)^2|r_{2,k}| + \abs{\ub_2})^2 - r_{2,k}^2 
\\&\leq \frac{1}{4}\norm{\ugamma_{k}}_1^2 + \frac{\ub_1^2+\ub_2^2}{4p(1-p)}.
\end{align}
We maximize over $r_{1,k}$ and $r_{2,k}$ in the last step. 
Since $((1-2p)x+\ub)^2 - x^2$ is a function of $x$ that decreases when $x \geq \frac{1-2p}{4p(1-p)}\ub$ for any $\ub>0$, we have for all $l$ such that $|r_{1,l}| = \frac{1}{2b}\abs{\trace(\ugamma_l \beta_{l+1})} > \frac{f(n)}{2b(1+\delta)} \geq \frac{1-2p}{4p(1-p)}\abs{\ub_1}$, 
\begin{equation}
\begin{split}
&\quad \sum_{l: \abs{\trace(\ugamma_l \beta_{l+1})} > \frac{f(n)}{1+\delta}} \frac{1}{4}\norm{\ugamma_{l+1}}_1^2 - \frac{1}{4}\norm{\ugamma_{l}}_1^2\\
&\leq \Bigg(\sum_{l: \abs{\trace(\ugamma_l \beta_{l+1})} > \frac{f(n)}{1+\delta}}  ((1-2p)|r_{1,k}| + \abs{\ub_1})^2 - r_{1,k}^2\Bigg) + n_0\frac{\ub_2^2}{4p(1-p)}\\
&\leq -n_0\frac{p(1-p)g(n)^2}{b^2} - n_0\frac{(1-2p)g(n)\abs{\ub_1}}{b} + n_0\ub_1^2 
+ n_0\frac{\ub_2^2}{4p(1-p)},
\end{split}
\end{equation}
where we use the Cauchy--Schwarz inequality in the last step. 
Then we consider dividing summation $\sum_{k=1}^{n-2}$ into two parts, $\sum_{k: \abs{\trace(\ugamma_k \beta_{k+1})} > \frac{f(n)}{1+\delta}} $ and $\sum_{k: \abs{\trace(\ugamma_k \beta_{k+1})} \leq \frac{f(n)}{1+\delta}}$, we have the following inequality: 
\begin{gather}
- \norm{\ugamma_{1}}_1^2 \leq \norm{\ugamma_{n-1}}_1^2 - \norm{\ugamma_{1}}_1^2 = \sum_{k=1}^{n-2} \norm{\ugamma_{k+1}}_1^2 - \norm{\ugamma_{k}}_1^2, \\
\Rightarrow~~ -(\ub_1^2+\ub_2^2) \leq - p(1-p) \frac{g(n)^2}{b^2} n_0 + (1-2p) \frac{g(n)\abs{\ub_1}}{b} n_0 + \ub_1^2n_0 + \frac{1}{4p(1-p)} ((n-1-n_0)\ub_1^2+(n-1)\ub_2^2),\nonumber \\
\Rightarrow~~ 
n_0 p(1-p) \left( \frac{g(n)}{b} - \frac{(1-2p) \abs{\ub_1}}{2 p(1-p)} \right)^2 \leq \left(\frac{1}{4p(1-p)}+2\right) (n-1)(\ub_1^2+\ub_2^2),
\label{eq:n0-upper}
\end{gather}
Note that $g(n) \geq f(n) > \frac{(1-2p) b}{2 p(1-p)}\abs{\ub_1}$ so the left-hand side is always positive. 
Combining \eqref{eq:n0-lower} (a lower bound on $n_0$) and \eqref{eq:n0-upper} (an upper bound on $n_0$), we have 
\begin{equation}
f(n) \leq \frac{(1+\frac{1}{\delta})g(n)}{\frac{1}{\delta}+ \frac{ p(1-p) \left( {g(n)} - \frac{(1-2p) b\abs{\barbelowsmall{b}_1}}{2 p(1-p)} \right)^2 }{\left(\frac{1}{4p(1-p)} + 2 \right)(b^2(\barbelowmedium{b}_1^2+\barbelowmedium{b}_2^2))} }.
\end{equation}
Maximizing the denominator over all $g(n) \geq 0$ and let $\delta = 1$, we have 
\begin{equation}
f(n) \leq \frac{2b}{p}  \sqrt{\left(\frac{(1-2p) b\abs{\ub_1}}{2 (1-p)}\right)^2  +  \left(\frac{1}{4(1-p)^2} + \frac{2p}{(1-p)} \right)(b^2(\ub_1^2+\ub_2^2))}. 
\end{equation}
\eqref{eq:upper} is then proven. 
\end{proof}

\subsection{Generalization to bounded ancilla}

Here we will consider situations where bounded ancilla is allowed and the unital controls (and the initial state) can act on the global system of the probe and the ancilla (see \figeref{fig:strategies}). We have the following theorem. 

\begin{theorem}[QFI upper bound for estimating dephasing-class channels using sequential strategies with unital controls and bounded noiseless ancilla]  
\label{thm:dephasing-unital-bounded-ancilla}
    Consider estimating an unknown parameter $\theta$ in a dephasing-class channel $\mE_\theta$ of the most general form in \eqref{eq:one-parameter-dephasing-qubit-channel} using sequential strategies with $n_A$ noiseless ancillary qubits and unital controls $\{\mU_k\}_{k=1}^n$ acting across the probe system and the ancillary system, the channel QFI of 
    \begin{equation}
        \mA_\theta^{(n)} := \mU_n \circ (\mE_\theta \otimes \id_A) \circ \cdots \circ \mU_1 \circ (\mE_\theta \otimes \id_A)
    \end{equation}
    satisfies 
    \begin{equation}
    F(\mA_\theta^{(n)}) \leq F(\mA^{(n)}_\theta \otimes \id) \leq 2^{n_A}\mu_1 n,  
    \end{equation}
    where $\mu_1$ is a constant independent of $n,\,n_A$, and $\mu_1 = O(1/p_\theta)$ for small $p_\theta$.
    \end{theorem}

\begin{proof}
    The proof follows almost directly from previous discussions with a slight generalization. First, note that \lemmaref{lemma:channel-extension} can be extended to the bounded-ancilla-assisted sequential strategy with unital controls. (One can verify that the proof of \lemmaref{lemma:channel-extension} does not assume $\mE_\theta$ is a qubit channel, therefore one can replace $\mE_\theta$ with $\mE_\theta \otimes \id_A$ and derive a new upper bound accordingly). Choosing the Kraus operator of $\mU_k \circ (\mE_\theta \otimes \id_A)$ to be $\tK_{(a,b)_k} = C_{a_k}  (\tK_{b_k} \otimes \id_A)$, where $\{C_{a_k}\}_a$ are any Kraus operators of $\mU_k$ and $\{\tK_{b_k}\}_b$ are Kraus operators defined the same as in the above discussion, i.e., satisfying \eqref{eq:tk01}, \eqref{eq:tk0d}, \eqref{eq:tk1d} and \eqref{eq:h-choice}, we have 
    \begin{equation}
        F(\mA_\theta^{(n)}) \leq F(\mA_\theta^{(n)} \otimes \id) \leq \sum_{k=1}^{n} 4 \trace(\iota^{\mA}_{k-1} \alpha^{\mA}_k) + \sum_{k=1}^{n-1} 8\, \trace(\ugamma^{\mA}_k \beta^{\mA}_{k+1}),
    \end{equation}
    where for all $ 1 \leq k \leq n$, 
    \begin{gather}
            \iota^{\mA}_{k-1} = \mA^{(k-1)}_\theta(\id) = \mU_{k-1}\circ (\mE_\theta\otimes\id) \circ \cdots \mU_{1}\circ (\mE_\theta\otimes\id)(\id) = \id,\\
    \alpha^\mA_k = \alpha_k|_{\mU_k = \id}\otimes \id_A,\quad 
    \beta^\mA_k = \beta_k|_{\mU_k = \id}\otimes \id_A,\quad \ubeta^\mA_k = \mU_k(\ubeta_k|_{\mU_k = \id}\otimes \id_A),\quad\\
    \ugamma^\mA_{k} = \mU_{k}\circ(\mE_\theta \otimes \id_A)(\ugamma^\mA_{k-1}) + \ubeta^\mA_{k},\quad \ugamma_0 = 0,
    \end{gather}
    where the expressions of $\alpha_k,\beta_k,\ubeta_k$ can be found at \eqref{eq:alpha-dephasing-unital}, \eqref{eq:beta-dephasing-unital} and \eqref{eq:ubeta-dephasing-unital}. Our goal is to show 
    \begin{equation}
        \sum_{k=1}^{n-1} \abs{\trace(\ugamma^{\mA}_k \beta^{\mA}_{k+1})} \leq 2^{n_A} \mu n,
    \end{equation}
    for some $\mu$ independent of $n_A$. 
    To see this, we follow the same proof strategy as in the proof of \thmref{thm:dephasing-unital}. (Note that we will not spell out every step below to avoid repetitions from the previous proof, and we will only detail new steps). We adopt the same notations of $\tilde{X},\tilde{Y},b,\ub_1,\ub_2$. 
    We also assume analogously 
    \begin{equation}
    f(n) := \frac{\sum_{k=1}^{n-1} \abs{\trace(\ugamma^{\mA}_k \beta^{\mA}_{k+1})}}{n-1} >
    \frac{2^{n_A}(1-2p)b\abs{\ub_1}}{2p(1-p)} (1+\delta), 
    \end{equation}
    and 
\begin{equation}
g(n) := \frac{1}{n_0} \sum_{l: \abs{\trace(\ugamma^{\mA}_l \beta^{\mA}_{l+1})} > \frac{f(n)}{1+\delta}, \,l \leq n-1} \abs{\trace(\ugamma^{\mA}_l\beta^{\mA}_{l+1})},
\end{equation}
where $\delta = 1$ and $n_0 = \abs{\left\{l: \abs{\trace(\ugamma^{\mA}_l \beta^{\mA}_{l+1})} > \frac{f(n)}{1+\delta}, \,l \leq n-1\right\}}$ is the number of terms that is above $1/(1+\delta)$ of the average. 
    We let 
    \begin{equation}
        \gamma^\mA_{k} = r_{1,k} \tilde{X}\otimes \id_A + r_{2,k} \tilde{Y}\otimes \id_A + r_{3,k} Z\otimes \id_A + \sum_{i=0}^3\sum_{j=1}^{4^{n_A}-1} r_{ij,k} \tilde{\sigma}_i \otimes B_j,
    \end{equation}
    where $\{B_j\}_{j=0}^{4^{n_A}-1}$ is an orthogonal basis of Hermitian operators in the ancillary system (e.g. Pauli basis) satisfying $\trace(B_i B_j) = 2^{n_A}$, $B_0 = \id_A$, and $\tilde\sigma_{0,1,2,3}=\id,\tilde{X},\tilde{Y},Z$. $r_{00,k} = 0$ because $\gamma^\mA_{k}$ is traceless. Then 
    \begin{align}
        \gamma^\mA_{k+1} &= \mU_{k}\circ(\mE_\theta \otimes \id_A)(\ugamma^\mA_{k}) + \ubeta^\mA_{k+1} \\ &= \mU_{k} \Bigg( ((1-2p)r_{1,k}+\ub_1) \tilde{X}\otimes \id_A + ((1-2p)r_{2,k}+\ub_2) \tilde{Y}\otimes \id_A + r_{3,k} Z\otimes \id_A + \sum_{i=0}^3\sum_{j=1}^{4^{n_A}-1} s_{ij,k} \tilde{\sigma}_i \otimes B_j \Bigg),
    \end{align} 
    where $\abs{s_{ij,k}} \leq \abs{r_{ij,k}}$. Using the contractivity of the Hilbert-Schmidt norm under unital channels~\cite{perez2006contractivity}, we have for all $k$,
    \begin{align}
        \frac{1}{2^{1+n_A}}\trace((\gamma^\mA_{k+1})^2) &\leq \frac{1}{2^{1+n_A}}\trace((\gamma^\mA_{k+1})^2|_{\mU_k = \id}) \\
        &= ((1-2p)r_{1,k}+\ub_1)^2 + ((1-2p)r_{2,k}+\ub_2)^2 + r_{3,k}^2 + \sum_{i=0}^3\sum_{j=1}^{4^{n_A}-1} s_{ij,k}^2 \\
        &\leq \frac{1}{2^{1+n_A}}\trace((\gamma^\mA_{k})^2) + \frac{\ub_1^2+\ub_2^2}{4p(1-p)},
    \end{align}
    and for all $l$ such that $\abs{r_{1,l}} = \frac{1}{2^{1+n_A} b}\abs{\trace(\gamma^{\mA}_{l}\beta^{\mA}_{l+1})} > \frac{f(n)}{2^{1+n_A}b(1+\delta)}\geq \frac{1-2p}{4p(1-p)}\abs{\ub_1}$, 
            \begin{align}
        \frac{1}{2^{1+n_A}}\trace((\gamma^\mA_{l+1})^2) - \frac{1}{2^{1+n_A}}\trace((\gamma^\mA_{l})^2) \leq ((1-2p)\abs{r_{1,l}}+\ub_1)^2 - \abs{r_{1,l}}^2 + \frac{\ub_2^2}{4p(1-p)}. 
    \end{align}
    The rest of the proof follows identically from \thmref{thm:dephasing-unital}, replacing $\frac{1}{4}\norm{\ugamma_k}_1^2$ in the previous proof with $\frac{1}{2^{1+n_A}}\trace((\gamma^\mA_{k})^2)$ in the current proof.     
\end{proof}
    
\subsection{General CPTP controls}
\label{app:dephasing-non-unital}

Above, we have shown that the ancilla-free (or bounded-ancilla) sequential strategies with unital controls can only achieve the SQL, in contrast to the ancilla-assisted sequential strategy where the HL is recoverable when HNKS is satisfied. We use the channel extension method (\lemmaref{lemma:channel-extension}) to show that there is a suitable choice of Kraus representation of each $\mU_k \circ \mE_{\theta}$ such that $\beta_k$ and $\ubeta_k|_{\mU_k=\id}$ are perpendicular to noise (Pauli $Z$) which leads to a constant upper bound on the average of $\abs{\trace(\ugamma_k \beta_{k+1})}$ and the corresponding linear scaling of the QFI. For non-unital controls, the above technique does not directly generalize because
\begin{equation}
    \iota_{k-1} = \mE^{(k-1)}_\theta(\id) = \mC_{k-1}\circ\mE_\theta \circ \cdots \mC_{1}\circ \mE_\theta(\id) = \id
\end{equation}
no longer holds true and $\ubeta_k$ (\eqref{eq:ubeta-general}) no longer has a simple form. 

Luckily, our proof strategies can be still modified to prove a non-trivial bound with general quantum controls, as we will show below. 
\begin{theorem}[QFI upper bound for estimating dephasing-class channels using ancilla-free sequential strategies with CPTP controls]  
\label{thm:dephasing-non-unital}
    Consider estimating an unknown parameter $\theta$ in a dephasing-class channel $\mE_\theta$ of the most general form in \eqref{eq:one-parameter-dephasing-qubit-channel} using the ancilla-free sequential strategy and assume arbitrary CPTP controls $\mC_k$, the channel QFI of \eqref{eq:entire-channel} satisfies 
    \begin{equation}
    F(\mE_\theta^{(n)}) \leq F(\mE^{(n)}_\theta \otimes \id) \leq \mu_2 n^{3/2},
    \end{equation}
    where $\mu_2$ is a constant independent of $n$ and $\mu_2 = O(1/p_\theta)$ for small $p_\theta$.
    \end{theorem}
\thmref{thm:dephasing-non-unital} shows the QFI of $n$ concatenated dephasing-class channels is at most $O(n^{3/2})$, which has a clear separation from the HL $O(n^2)$. (However, $O(n^{3/2})$ may not be a tight bound and it still remains open whether there is a separation between ancilla-free sequential strategies with unital controls and that with non-unital controls.)

To prove \thmref{thm:dephasing-non-unital}, we apply the channel extension method (\lemmaref{lemma:channel-extension}) in two steps and a useful upper bound on the norm of $\alpha$ operators (\lemmaref{lemma:alpha-bound-non-unital}) that we previously proved in \appref{app:HNKS}. 
In the first step, we partition the entire quantum channel $\mE^{(n)}_\theta$ into the concatenation of several shorter segments such that each of them is \emph{not too} non-unital and prove each of them has a QFI upper bound of a cube-of-square-root scaling with respect to the number of concatenated $\mE_\theta$. In the second step, we calculate an upper bound on $F(\mE^{(n)}_\theta)$ using the channel extension method (\lemmaref{lemma:channel-extension}), where each segment functions as a single quantum channel and each $\norm{\alpha}$ is bounded through \lemmaref{lemma:alpha-bound-non-unital}. 

To this end, we first prove the following lemma, which shows the cube-of-square-root upper bound for channels that are not too non-unital. 
\begin{lemma}[QFI upper bound when controls are not too non-unital]
Consider the same setting as in \thmref{thm:dephasing-non-unital}. Assume $\abs{\trace(\iota_k Z)} < 1$ for all $k \leq n$. Then 
\begin{equation}
\label{eq:non-unital-partial}
F(\mE^{(n)}_\theta) \leq F(\mE^{(n)}_\theta \otimes \id) < \tilde{\mu}_2 n^{3/2},
\end{equation}
where $\tilde{\mu}_2$ is a constant independent of $n$ and $\tilde\mu_2 = O(1/p_\theta)$ for small $p_\theta$.
\label{lemma:non-unital-partial}
\end{lemma}

\begin{proof}
In \appref{app:dephasing-unital}, we chose $h_{00} = h_{11} = 0$ and $h_{01} = h_{10} = -\frac{(1-p)\trace(G_0 Z) + p \trace(G_1 Z)}{4\sqrt{p(1-p)}}$ in \eqsref{eq:alpha-general}{eq:ubeta-general} such that $\beta_k$ and $\ubeta_k|_{\mC_k=\id}$ are perpendicular to noise (Pauli $Z$) which allows us to derive the linear upper bound on QFI. For non-unital controls, the above is no longer true and we need to choose $h$ more carefully as below. 
Let 
\begin{gather}
z_k = \frac{\trace(\iota_{k}Z)}{2} ,\quad 
g_{+,k} 
= \frac{\trace(\iota_{k}G_+)}{2} - \frac{\trace(G_+ Z)}{2}z_k,
\quad g_{-,k}
= \frac{\trace(\iota_{k}G_-)}{2} - \frac{\trace(G_- Z)}{2}z_k.
\end{gather}
$\abs{z_k} < 1/2$ as assumed. $g_{\pm,k}$ are bounded by $\abs{g_{\pm,k}} \leq \norm{G_\pm}$. We will take $\Im[h_{01}] = 0$ for simplicity and we will take other parts of $h$ to be the ones that satisfy 
\begin{gather}
(1-p) h_{00} + p h_{11} = \frac{-g_{+,k-1}}{1-z_{k-1}^2},
\label{eq:h-cond-1}\\
2\sqrt{p(1-p)}\Re[h_{01}] + \frac{\trace(G_+ Z)}{2} = \frac{z_{k-1}g_{+,k-1}}{1-z_{k-1}^2}
,\label{eq:h-cond-2}\\
(1-p)h_{00}-p h_{11} =  
-  g_{-,k-1} 
- \frac{\trace(G_- Z)}{2} z_{k-1}.  \label{eq:h-cond-3}
\end{gather}
(Note that for unital controls, $\iota_{k} = \id$, $g_{+,k}=g_{-,k}=z_{k} = 0$, and the above choice of $h$ is exactly the one we made in \appref{app:dephasing-unital}).  
One can verify that the first two conditions (\eqref{eq:h-cond-1} and \eqref{eq:h-cond-2}) guarantee
\begin{equation}
    \trace(\ubeta_k)|_{\mC_k=\id} = \trace(Z\ubeta_k)|_{\mC_k=\id} = 0,
\end{equation}
and the third condition (\eqref{eq:h-cond-3}) minimizes $\trace(\iota_{k-1}\alpha_k)$ while the first two conditions are satisfied. Then 
\begin{gather}
\beta_k = G_+ - \frac{\trace(G_+ Z)}{2} Z + \frac{-g_{+,k-1}}{1-z_{k-1}^2} \id + \frac{z_{k-1}g_{+,k-1}}{1-z_{k-1}^2} Z,\\
\trace(X\ubeta_k)|_{\mC_k=\id} = \frac{1}{2}\trace(\iota_{k-1}\{X,G_-\}) - \left( g_{-,k-1} + \frac{\trace(G_- Z)}{2} z_{k-1}\right) \trace(\iota_{k-1} X), 
\label{eq:x-expression}\\
\trace(Y\ubeta_k)|_{\mC_k=\id} = \frac{1}{2}\trace(\iota_{k-1}\{Y,G_-\}) - \left( g_{-,k-1} + \frac{\trace(G_- Z)}{2} z_{k-1}\right) \trace(\iota_{k-1} Y) 
,\label{eq:y-expression}\\
\begin{split}
\frac{1}{2}\trace(\iota_{k-1}\alpha_k)  &= \frac{1-p}{2}\trace(G_0^2) + \frac{p}{2}\trace(G_1^2) + \frac{\dot p^2}{4p(1-p)} - h_{00}^2 - h_{11}^2 + h_{01}^2 + h_{01} \sqrt{p(1-p)} \trace((G_0+G_1)Z).
\end{split}
\end{gather}
where $h_{00,11,01}$ are chosen as in \eqsref{eq:h-cond-1}{eq:h-cond-3}. 

Clearly, $\trace(\iota_{k-1}\alpha_k) = O(1/p)$ for all $k$. 
\lemmaref{lemma:channel-extension} indicates 
    \begin{equation}
    {F(\mE^{(n)}_\theta\otimes \id)} \leq \sum_{k=1}^{n} 4 \trace(\iota_{k-1} \alpha_k) + \sum_{k=1}^{n-1} 8\, \trace(\ugamma_k \beta_{k+1}).
    \end{equation}
The first term above has at most a linear scaling with $n$. It is sufficient to show the second term is upper bounded by $\mu n^{3/2}$ for some $\mu = O(1/\sqrt{p})$.  
Assume $\ugamma_k = r_{1,k} X + r_{2,k} Y + r_{3,k} Z$ ($\ugamma_k$ is always traceless because $\ubeta_k$ is traceless). 
We have 
\begin{gather}
\frac{1}{4}\norm{\ugamma_k}_1^2 =  \frac{1}{2}\trace(\ugamma_k^2) = r_{1,k}^2 + r_{2,k}^2 + r_{3,k}^2,\\
\ubeta_k = \mC_k \left( \frac{\trace(X\ubeta_k)|_{\mC_k=\id}}{2} X + \frac{\trace(Y\ubeta_k)|_{\mC_k=\id}}{2} Y\right),\\
\begin{split}
\ugamma_{k+1} 
&= \mC_{k+1}\left(\left((1-2p) r_{1,k} + \frac{\trace(X\ubeta_k)|_{\mC_k=\id}}{2}\right) X + \left((1-2p)r_{2,k} + \frac{\trace(Y\ubeta_k)|_{\mC_k=\id}}{2}\right) Y + r_{3,k} Z \right),
\end{split}
\end{gather}
Using the data-processing inequality, we have 
\begin{align}
\frac{1}{4}\norm{\ugamma_{k+1}}_1^2 &\leq  \frac{1}{4}\norm{\ugamma_{k+1}}_1^2\big|_{\mC_{k+1} = \id}  = \frac{1}{2}\trace(\ugamma_{k+1}^2)\big|_{\mC_{k+1} = \id}\\ &=  \left((1-2p)r_{1,k} + \frac{\trace(X\ubeta_k)|_{\mC_k=\id}}{2}\right)^2 + \left((1-2p)r_{2,k} + \frac{\trace(Y\ubeta_k)|_{\mC_k=\id}}{2}\right)^2 + r_{3,k}^2 
\\&\leq \frac{1}{4}\norm{\ugamma_{k}}_1^2 
+ \frac{\trace(X\ubeta_k)^2\big|_{\mC_{k} = \id} + \trace(Y\ubeta_k)^2\big|_{\mC_{k} = \id}}{16p(1-p)}. 
\end{align}
Note that from \eqref{eq:x-expression} and \eqref{eq:y-expression}, 
\begin{align}
\trace(X\ubeta_k)^2 + \trace(Y\ubeta_k)^2\big|_{\mC_{k} = \id} 
&\leq \trace(\iota_{k-1}\{X,G_-\})^2 + \trace(\iota_{k-1}\{Y,G_-\})^2 + \trace(\iota_{k-1}G_+)^2(\trace(\iota_{k-1}X)^2 + \trace(\iota_{k-1}Y)^2)\nonumber\\
&\leq    8(\norm{G_-}^2 + \norm{G_+}^2). 
\end{align}
Then 
\begin{equation}
\frac{1}{4}\norm{\ugamma_{k}}_1^2 \leq k \frac{\norm{G_-}^2 + \norm{G_+}^2}{2p(1-p)}, 
\end{equation}
and 
\begin{align}
\abs{\trace(\ugamma_k\beta_{k+1})} 
&\leq \norm{\ugamma_k}_1 \norm{\beta_{k+1}}
\leq \sqrt{2k \frac{\norm{G_-}^2 + \norm{G_+}^2}{p(1-p)}}
\left(\norm{G_+} + \frac{|g_{+,k-1}|}{1-|z_{k-1}|} + \frac{|\trace(G_+ Z)|}{2}\right)\\
&\leq 4\sqrt{2k} \sqrt{\frac{\norm{G_-}^2 + \norm{G_+}^2}{p(1-p)}}\norm{G_+}. 
\end{align}
Finally, we have 
\begin{equation}
\sum_{k=1}^{n-1} \abs{\trace(\ugamma_k\beta_{k+1})} \leq 4\sqrt{2} n^{3/2} \sqrt{\frac{\norm{G_-}^2 + \norm{G_+}^2}{p(1-p)}}\norm{G_+},
\end{equation}
proving \eqref{eq:non-unital-partial} using \lemmaref{lemma:alpha-bound-non-unital}. 
\end{proof}

(Note that we proved above an upper bound of $O(n)/p + O(n^{3/2})/\sqrt{p}$ on the QFI, which is slightly more general than the statement of the lemma, but for simplicity we will not delve into the details here as the scaling with respect to $p$ is not our focus.)
With \lemmaref{lemma:non-unital-partial}, we are now ready to present the full proof of \thmref{thm:dephasing-non-unital}.

\begin{proof}[Proof of {\thmref{thm:dephasing-non-unital}}]
Let $\mE^{(\ell_1,\ell_2)}_\theta$ denotes $\mC_{\ell_2}\circ\mE_\theta \circ \cdots \mC_{\ell_1} \circ \mE_\theta$ 
where $\ell_2 \geq \ell_1$ and $\ell_{1,2}$ is the $\ell_{1,2}$-th control operation. For example, $\mE^{(1,n)}_\theta = \mE^{(n)}_\theta$. 
We will also let $\mE^{(\ell_1,\ell_2)}_\theta$ denote the identity channel and $K^{(\ell_1,\ell_2)}_{\theta,\vi}$ the identity operator when $\ell_2 < \ell_1$. We choose $n_0 < n_1 < n_2 < \cdots < n_m$ satisfying $n_0 = 0$ and $n_m = n$ such that the channel $\mE^{(n)}_\theta$ is partitioned into the composition of the following $m$ channels 
\begin{equation}
\mE^{(1,n_1)}_\theta, \quad 
\mE^{(n_1+1,n_2)}_\theta, \quad 
\cdots, \quad 
\mE^{(n_{m-1}+1,n_m)}_\theta,  
\end{equation}
i.e.,
\begin{equation}
\mE^{(n)}_\theta = \mE^{(n_{m-1}+1,n_m)}_\theta \circ \cdots \circ \mE^{(n_1+1,n_2)}_\theta \circ \mE^{(1,n_1)}_\theta,
\end{equation}
such that 
\begin{equation}
\abs{\frac{\trace(\mE^{(\ell_1,\ell_2)}(\id)Z)}{2} } < \frac{1}{2}
\end{equation}
whenever $\ell_1 = n_k + 1$ and $n_k + 1 \leq \ell_2 \leq n_{k+1}-1$ for all $k=0,\ldots, m-1$, and 
\begin{equation}
\abs{\frac{\trace(\mE^{(\ell_1,\ell_2)}(\id)Z)}{2} } \geq \frac{1}{2}
\end{equation}
whenever $\ell_1 = n_k + 1$ and $\ell_2 =n_{k+1}$ for all $k=0,\ldots, m-2$. 
Consider $\mE^{(1,n_{m-1})}$ as a composition of $\{\mE^{(n_k+1,n_{k+1})}\}_{k=0}^{m-2}$. Then we have from \lemmaref{lemma:channel-extension}, 
\begin{equation}
\label{eq:upper-m-segment}
{F(\mE^{(1,n_{m-1})}\otimes \id)} \leq \sum_{k=1}^{m-1} 4 \trace(\iota_{k-1} \alpha_k) + \sum_{k=1}^{m-2} 8\, \trace(\ugamma_k \beta_{k+1}),
\end{equation}
where
\begin{gather}
\alpha_k = \sum_{\vi} \dot \tK_{\vi}^{(n_{k-1}+1,n_{k})\dagger} \dot \tK_{\vi}^{(n_{k-1}+1,n_{k})\dagger} 
,\quad 
\beta_k = i \sum_{\vi}  \tK_{\vi}^{(n_{k-1}+1,n_{k})\dagger} \dot \tK_{\vi}^{(n_{k-1}+1,n_{k})} 
,\\
\ugamma_{k+1} = \mE^{(n_{k}+1,n_{k+1})}(\ugamma_k) + \ubeta_{k+1},\qquad \ugamma_0 = 0,\\   
\ubeta_k 
= \frac{1}{2}\bigg( i \sum_{\vi} \dot\tK_{\vi}^{(n_{k-1}+1,n_{k})} \iota_{k-1}  \tK_{\vi}^{(n_{k-1}+1,n_{k})\dagger} - i \sum_{\vi} \tK_{\vi}^{(n_{k-1}+1,n_{k})} \iota_{k-1} \dot\tK_{\vi}^{(n_{k-1}+1,n_{k})\dagger} \bigg) 
,\\  
\iota_{k-1} = \mE^{(1,n_{k-1})}(\id), \quad \iota_{0} = \id,
\end{gather}
and $\tK_{\vi}^{(n_{k-1}+1,n_{k})}$ can be any Kraus operators representing $\mE^{(n_{k-1}+1,n_{k})}$. (Note that in this proof, $\iota_{k},\alpha_k,\beta_k,\ubeta_k,\ugamma_k$ are defined and used differently from previous discussions, but we keep these notations for simplicity.)

Using \lemmaref{lemma:non-unital-partial} and \eqref{eq:sdp}, we know that there exists a Kraus representation $\{\tK_{\vi}^{(n_{k-1}+1,n_{k}-1)}\}$ of $\mE^{(n_{k-1} + 1,n_{k}-1)}$ such that 
\begin{equation}
\big\|\alpha^{(n_{k-1} + 1,n_{k}-1)}\big\| = \frac{F(\mE^{(n_{k-1} + 1,n_{k}-1)} \otimes \id)}{4} \leq \frac{\tilde{\mu}_2}{4} (n_{k} - n_{k-1} - 1)^{3/2}. 
\end{equation}
for all $1 \leq k \leq m$, 
where 
\begin{equation}
\alpha^{(n_{k-1} + 1,n_{k}-1)} = \sum_{\vi} \dot\tK^{(n_{k-1} + 1,n_{k}-1)\dagger }_\vi \dot\tK^{(n_{k-1} + 1,n_{k}-1)}_\vi.  
\end{equation}

First, let us choose the Kraus representation of $\mE^{(n_{k-1} + 1,n_{k})}$ to be 
\begin{equation}
\label{eq:kraus-choice-1}
\tK^{(n_{k-1} + 1,n_{k})}_{i,\vi'} = K_i \tK^{(n_{k-1} + 1,n_{k}-1)}_{\vi'}.
\end{equation}
$\{K_i\}$ is a Kraus representation of $\mC_{n_{k+1}} \circ \mE$ that can be written as $K_{i=(a,b)} = C_a K_b$ where $\{C_a\}$ is an arbitrary Kraus representation of $\mC_{n_{k+1}}$ and $\{K_b\}$ in the natural choice of Kraus representation from \eqref{eq:natural-kraus}. Specifically, we have 
\begin{equation}
    \sum_i \dot K_i^\dagger \dot K_i = \sum_b \dot K_b^\dagger \dot K_b =  \bigg(\frac{1-p}{2}\trace(G_0^2) + \frac{p}{2}\trace(G_1^2)\bigg) \id. 
\end{equation}
(Note that $\{K_i\}$ depends on $n_{k}$, but for simplicity we omit the superscript here.) Using the Kraus representation in \eqref{eq:kraus-choice-1}, we have 
\begin{equation*}
\begin{split}
\norm{\alpha_k}\big|_{\text{\eqref{eq:kraus-choice-1} holds}}&= \bigg\|\sum_\vi \dot\tK^{(n_{k-1} + 1,n_k)\dagger}_{\vi} \dot\tK^{(n_{k-1} + 1,n_k)}_{\vi}\bigg\| \\
&= \bigg\| \sum_{\vi'} \dot\tK^{(n_{k-1} + 1,n_k-1)\dagger}_{\vi'} \dot\tK^{(n_{k-1} + 1,n_k-1)}_{\vi'} + \sum_{i,\vi'} \tK^{(n_{k-1} + 1,n_k-1)\dagger}_{\vi}\dot K_i^\dagger K_i  \dot\tK^{(n_{k-1} + 1,n_k-1)}_{\vi} \\
&\quad  \;+ \sum_{i,\vi'} \dot\tK^{(n_{k-1} + 1,n_k-1)\dagger}_{\vi} K_i^\dagger \dot K_i  \tK^{(n_{k-1} + 1,n_k-1)}_{\vi} + \sum_{i,\vi'} \tK^{(n_{k-1} + 1,n_k-1)\dagger}_{\vi'} \dot K_i^\dagger \dot K_i \tK^{(n_{k-1} + 1,n_k-1)}_{\vi'} \bigg\|\\
&\leq \norm{\alpha^{(n_{k-1} + 1,n_k-1)}} + \Big\|\sum_i \dot K_i^\dagger \dot K_i\Big\|+ 2\sqrt{\norm{\alpha^{(n_{k-1} + 1,n_k-1)}}\Big\|\sum_i \dot K_i^\dagger \dot K_i\Big\|}\\
&\leq 2 \norm{\alpha^{(n_{k-1} + 1,n_k-1)}} + 2 \Big\|\sum_i \dot K_i^\dagger \dot K_i\Big\| \leq \frac{\tilde{\mu}_2}{2} (n_k - n_{k-1} - 1)^{3/2} + (1-p)\trace(G_0^2) + p\trace(G_1^2),
\end{split}
\end{equation*}
and 
\begin{equation*}
\begin{split}
\norm{\beta_k}\big|_{\text{\eqref{eq:kraus-choice-1} holds}} = \bigg\|\sum_\vi \dot\tK^{(n_{k-1} + 1,n_k)\dagger}_{\vi} \tK^{(n_{k-1} + 1,n_k)}_{\vi}\bigg\| \leq \sqrt{\norm{\alpha_k}}\big|_{\text{\eqref{eq:kraus-choice-1} holds}},
\end{split}
\end{equation*}
where we use $\norm{AB} \leq \norm{A}\norm{B}$ and $\norm{A} = \sqrt{\norm{A^\dagger A}}$ for any matrices $A$ and $B$.

Next, we will choose a different Kraus representation of $\mE^{(n_{k-1} + 1,n_k)}$ to prove the desired QFI upper bound on $F(\mE^{(1,n_{m-1})}\otimes \id)$. We will apply \lemmaref{lemma:alpha-bound-non-unital} to $\mE^{(n_{k-1} + 1,n_k)}$ and the Kraus representation in \eqref{eq:kraus-choice-1} for all $k=1,\ldots,m-1$. Note that 
\begin{equation}
\norm{\mE^{(n_{k-1} + 1,n_k)}(\id)  - \id} \geq \abs{\frac{\trace(\mE^{(n_{k-1} + 1,n_k)}(\id)Z)}{2} } \geq \frac{1}{2}. 
\end{equation}
\lemmaref{lemma:alpha-bound-non-unital} states that (for all $1 \leq k \leq m-1$) there exists a different Kraus representation of $\mE^{(n_{k-1} + 1,n_{k})}$ from \eqref{eq:kraus-choice-1} such that 
\begin{equation}
\label{eq:kraus-choice-2}
\beta_k = 0,  \; \text{and} \;
\norm{\alpha_k} \leq 2 \norm{\alpha_k}\big|_{\text{\eqref{eq:kraus-choice-1} holds}} + \xi(\sqrt{\norm{\alpha_k}},\frac{1}{2})\big|_{\text{\eqref{eq:kraus-choice-1} holds}} = \left(2 + \xi(1,\frac{1}{2})\right) \norm{\alpha_k}\big|_{\text{\eqref{eq:kraus-choice-1} holds}} . 
\end{equation}
Plugging the new Kraus representation that satisfies \eqref{eq:kraus-choice-2} into \eqref{eq:upper-m-segment}, we have 
\begin{align}
{F(\mE^{(1,n_{m-1})}\otimes \id)}
&\leq 4\sum_{k=1}^{m-1} \trace(\iota_{k-1} \alpha_k) \leq 8\sum_{k=1}^{m-1} \norm{\alpha_k} \leq 8\left(2 + \xi(1,\frac{1}{2})\right) \sum_{k=1}^{m-1} \norm{\alpha_k}\big|_{\text{\eqref{eq:kraus-choice-1} holds}} \nonumber \\
&\leq 8\left(2 + \xi(1,\frac{1}{2})\right)\sum_{k=1}^{m-1} \bigg( \frac{\tilde{\mu}_2}{2} (n_k - n_{k-1} - 1)^{3/2} + (1-p)\trace(G_0^2) + p\trace(G_1^2) \bigg) \nonumber \\
&\leq 8\left(2 + \xi(1,\frac{1}{2})\right) \left(
\frac{\tilde{\mu}_2}{2} \, n_{m-1}^{3/2} + (m-1)\big((1-p)\trace(G_0^2) + p\trace(G_1^2)\big)\right). \label{eq:gluing}
\end{align}

Finally, we complete the proof by taking the last segment $\mE^{(n_{m-1}+1,n_m)}$ into consideration. The proof will then be completed considering the following two cases: 
\begin{enumerate}[(1)]
    \item $\abs{\frac{\trace(\mE^{(n_{m-1}+1,n_{m})}(\id)Z)}{2} } \geq \frac{1}{2}$. Then the discussion above (e.g. \eqref{eq:gluing}) can be directly extended to include the $k=m$ case to show that 
    \begin{equation}
    {F(\mE^{(n)}\otimes \id)} = {F(\mE^{(1,n_{m})}\otimes \id)}\leq 8\left(2 + \xi(1,\frac{1}{2})\right) \left(
\frac{\tilde{\mu}_2}{2} \, n_{m}^{3/2} + m\big((1-p)\trace(G_0^2) + p\trace(G_1^2)\big)\right). 
    \end{equation}
    \item $\abs{\frac{\trace(\mE^{(n_{m-1}+1,n_{m})}(\id)Z)}{2} } < \frac{1}{2}$. Using \lemmaref{lemma:non-unital-partial}, we have 
    \begin{equation}
    {F(\mE^{(n_{m-1}+1,n_{m})}\otimes \id)}  \leq \tilde\mu_2 (n_{m}-n_{m-1})^{3/2}. 
    \end{equation}
Then using the chain rule of root QFI~\cite{yuan2017fidelity,katariya2020geometric}, namely, 
\begin{equation}
    F( (\mM_\theta \circ \mN_\theta) \otimes \id)^{1/2} \leq F(  \mN_\theta \otimes \id)^{1/2} + F( \mM_\theta \otimes \id)^{1/2}, 
\end{equation}
we have 
    \begin{align}
    {F(\mE^{(n)}\otimes \id)}
    &\leq \left( \sqrt{{F(\mE^{(1,n_{m-1})}\otimes \id )}} + \sqrt{{F(\mE^{(n_{m-1}+1,n_{m})}\otimes \id )}} \right)^2\nonumber\\
    &\leq 2 {F(\mE^{(1,n_{m-1})}\otimes \id )} + 2 {F(\mE^{(n_{m-1}+1,n_{m})}\otimes \id )}\nonumber\\
    &\leq 16\left(2 + \xi(1,\frac{1}{2})\right) \left(
\frac{\tilde{\mu}_2}{2} \, n_{m-1}^{3/2} + (m-1)\big((1-p)\trace(G_0^2) + p\trace(G_1^2)\big)\right) + 2\tilde\mu_2 (n_{m}-n_{m-1})^{3/2}\nonumber\\
    &\leq 16\left(2 + \xi(1,\frac{1}{2})\right) \left(\frac{{\tilde{\mu}_2}}{2} n^{3/2} + n\big((1-p)\trace(G_0^2) + p\trace(G_1^2)\big)\right), 
    \end{align}
    which completes the proof. 
\end{enumerate}
\end{proof}

\section{Dephasing-class channels: Achieving the SQL}
\label{app:sql}

In \appref{app:dephasing-unital}, we showed the ancilla-free sequential strategy with unital controls can at most achieve a QFI of $O(n)$ (even when the HNKS condition is satisfied). In this appendix, we present a specific protocol that uses unitary controls and achieves the linear scaling of QFI for channels satisfying the RGNKS condition (i.e., either $G_0$ or $G_1$ is not proportional to a Pauli-Z operator), proving the upper bound is tight. Note that although the SQL, i.e., a QFI of $\Theta(n)$, can be trivially achieved when quantum measurement can be performed in each of the $n$ steps, it is previously unknown whether the SQL is achievable with only unitary controls and a single measurement in the end. In particular, as we will see later in \appref{app:dephasing-2}, when the RGNKS condition is violated (when both $G_0$ and $G_1$ are proportional to the Pauli-Z operator), the QFI is at most a constant when only unitary controls is available, in contrast to the case where repeated measurements are possible. 

Specifically, here we again consider the most general form of one-parameter dephasing-class channel (\eqref{eq:one-parameter-dephasing-qubit-channel}), i.e., 
\begin{equation}
\mE_\theta(\rho) = (1-p_\theta) e^{-i G_0 \theta} \rho e^{i G_0 \theta} + p_\theta  Z e^{-i G_1 \theta} \rho e^{-i G_1 \theta} Z, 
\end{equation}
and we assume the RGNKS condition is satisfied, i.e., either $G_0$ or $G_1$ contains a non-Pauli-Z component and one of $\trace(G_0 X)$, $\trace(G_0Y)$, $\trace(G_1Y)$ and $\trace(G_1Y)$ must be non-zero, and quantum controls are unitary operations $\mC_k(\cdot) = \mU_k(\cdot) = U_k(\cdot)U_k^\dagger$. We will prove a lower bound of $\Omega(n)$ on the channel QFI $F(\mE^{(n)}_\theta)$ by presenting a specific choice of $\{U_k\}_{k=1}^n$ and an initial state $\rho_0$ that achieves the SQL.

\begin{theorem}[Achieving the SQL using unitary controls for dephasing-class channels]  
\label{thm:dephasing-lower}
    Consider estimating an unknown parameter $\theta$ in a dephasing-class channel $\mE_\theta$ of the most general form in \eqref{eq:one-parameter-dephasing-qubit-channel} using the ancilla-free sequential strategy. Assume either $G_0$ or $G_1$ is not proportional to a Pauli-Z operator, i.e., the RGNKS condition is satisfied. Then for each $n$, there exists a unitary control sequence $\{(\mU_k)_n\}_{k=1}^n$ and an initial state $\rho_0$ such that 
    \begin{equation}
    F(\mE^{(n)}_\theta(\rho_0) = F((\mU_n)_n \circ \mE_{\theta} \circ (\mU_{n-1})_n \circ \mE_{\theta} \circ \cdots \circ (\mU_{1})_n \circ \mE_{\theta}(\rho_0) = \Theta(n), 
    \end{equation}
    where we use $(\cdot)_n$ to denote the dependence on $n$. The control sequences are different for different $n$.  
    \end{theorem}

\begin{proof}
Since we consider only unitary controls, it is convenient to represent the evolution of quantum states using Bloch sphere representation. To be more specific, for $k = 0,\ldots , n$,  let 
\begin{equation}
\rho_k = \mU_k \circ \mE_{\theta} \circ \mU_{n-1} \circ \mE_{\theta} \circ \cdots \circ \mU_{1} \circ \mE_{\theta}(\rho_0) =: \frac{\id + x_k X + y_k Y + z_k Z}{2},
\end{equation}
where the Bloch vector $\vv_k := (x_k,y_k,z_k)$ fully represents the quantum state $\rho_k$ at the $k$-th step. The iteration relation, for $k=0,\ldots,n$, 
\begin{equation}
    \rho_k = \mU_k \circ \mE_\theta (\rho_{k-1})
\end{equation}
can be translated to an iteration relation on the Bloch vector 
\begin{equation}
\label{eq:pauli-transfer}
    \vv_k = O_k \left( (1-p) e^{R_0 \theta} + p \begin{pmatrix}-1 & 0 & 0 \\ 0 & -1 & 0 \\ 0 & 0 & 1 \end{pmatrix}e^{R_1 \theta}\right) \vv_{k-1}, 
\end{equation}
where $O_k$ is a rotation matrix representing $\mU_k$ and 
\begin{equation}
    R_{0,1} = \begin{pmatrix}
        0 & -\trace(G_{0,1} Z) & \trace(G_{0,1} Y) \\
        \trace(G_{0,1} Z) & 0 & -\trace(G_{0,1} X) \\
        -\trace(G_{0,1} Y) & \trace(G_{0,1} X) & 0 
    \end{pmatrix}. 
\end{equation}
Note that quantum channel acts through standard matrix multiplication in \eqref{eq:pauli-transfer}, where this matrix representation is usually called \emph{Pauli transfer matrix} (It is originally a four-by-four matrix defined by $\frac{1}{2}\trace(\sigma_i\mN(\sigma_j))$ for Pauli basis $\sigma_{0,1,2,3} = \id,X,Y,Z$ representing any qubit channel $\mN(\cdot)$ but here we consider only a three-by-three matrix in basis $\sigma_{1,2,3}$ because every channel is unital in this discussion.)
Moreover, the derivative of $\vv_k$ at $\theta = 0$ satisfies
\begin{equation}
\dot\vv_k = O_k D \vv_{k-1} + O_k M \dot\vv_{k-1},   
\end{equation}
where $\dot\vv_{0} = 0$, 
\begin{equation}
    M = \begin{pmatrix}
    1-2p & 0 & 0 \\
    0 & 1-2p & 0 \\
    0 & 0 & 1
    \end{pmatrix},\quad 
    D = 
    \begin{pmatrix}
        -2\dot p & -\trace(G_- Z) & \trace(G_- Y) \\
        \trace(G_- Z) & -2\dot p & -\trace(G_- X) \\
        -\trace(G_+ Y) & \trace(G_+ X) & 0 
    \end{pmatrix}
\end{equation}
Assume $O_k = O$ for all $k$, i.e., unitary controls are the same each step. Then we can write 
\begin{equation}
\label{eq:SQL-iteration}
    \vv_n =  (O M)^n \vv_0,\quad 
    \dot\vv_n =  \sum_{k=0}^{n-1} (O M)^{n-k-1} OD (O M)^k \vv_0. 
\end{equation}

Since we assume $G_0$ and $G_1$ are not both proportional to Pauli-Z operators, $\trace(G_0 X)$, $\trace(G_0 Y)$, $\trace(G_1 X)$ and $\trace(G_1 Y)$ cannot be simultaneously zero. We first consider the case where $\trace(G_0 X) \neq 0$. (The proof when $\trace(G_0 Y) \neq 0$ follows analogously.).  Then we assert the following ansatz of $O = O_k$ can achieve the SQL: 
\begin{equation}
\label{eq:SQL-ansatz}
O = \begin{pmatrix}
1 & 0 & 0 \\
0 & \cos\varphi & -\sin\varphi \\
0 & \sin\varphi & \cos\varphi 
\end{pmatrix},
\end{equation}
which corresponds to the unitary control $U_k = \exp(-i\frac{\varphi}{2}X)$ and $\varphi = \sqrt{w/n}$ for some real number $w > 0$. To compute the output state and its derivative with respect to $\theta$ (\eqref{eq:SQL-iteration}) with the given ansatz \eqref{eq:SQL-ansatz}, we first compute the eigendecomposition of $O M = S\Lambda S^{-1}$, where 
{\small
\begin{equation}
    S =  \begin{pmatrix}
        1 & 0 & 0 \\
        0 & 1 & -\frac{\varphi}{2p} + O(\varphi^3) \\
        0 & -\frac{1-2p}{2p}\varphi + O(\varphi^3) & 1
    \end{pmatrix} ,\quad 
    \Lambda = \begin{pmatrix}
        1-2p & 0 & 0 \\
        0 & 1-2p + \frac{(1-p)(1-2p)}{2p}\varphi^2 + O(\varphi^3) & 0 \\
        0 & 0 & 1 - \frac{(1-p)\varphi^2}{2p} + O(\varphi^3) \\
    \end{pmatrix} . 
\end{equation}}
Then 
\begin{gather}
    (OM)^n = S \Lambda^n S^{-1} = 
    \exp\left(-\frac{1-p}{2p}w\right) \begin{pmatrix}
        0 & 0 & 0 \\
        0 & 0 & -\frac{\sqrt{w}}{2p\sqrt{n}} \\
        0 & \frac{(1-2p)\sqrt{w}}{2p\sqrt{n}} & 1 
    \end{pmatrix} + O(1/n), \\
    \begin{split}\sum_{k=0}^{n-1} (O M)^{n-k-1} OD (O M)^k &= \sum_{k=0}^{n-1} S \Lambda^{n-k-1} S^{-1} O D S \Lambda^k S^{-1} 
    \\& = \exp\left(-\frac{1-p}{2p}w\right) \frac{(1-p)\trace(G_0X)\sqrt{w n}}{p} \begin{pmatrix}
        0 & 0 & 0 \\
        0 & 0 & 0 \\
        0 & 0 & 1 
    \end{pmatrix} + O(1). 
    \end{split}
\end{gather}
We have the output state and its derivative equal to 
\begin{equation}
\label{eq:derivative-rho}
    \rho_n = \frac{\id + \exp\left(-\frac{p}{2(1-p)}w\right) z_0 Z}{2} + O(1/\sqrt{n}), \quad\dot\rho_n = \frac{\exp\left(-\frac{p}{2(1-p)}w\right) \frac{p\trace(G_0X)\sqrt{w n}}{1-p} z_0 Z}{2} + O(1). 
\end{equation}
The QFI of $\rho_n$ can be computed as 
\begin{equation}
\label{eq:control-QFI-1}
    F(\rho_n) = \frac{\frac{(1-p)^2}{p^2}w}{z_0^{-2}\exp\left(\frac{1-p}{p}w\right)-1} \trace(G_0 X)^2  n + O(\sqrt{n}). 
\end{equation}
Taking $z_0 = 1$ and $w$ a small number, $F(\rho_n)$ can be arbitrarily close to $\frac{1-p}{p} \trace(G_0 X)^2 n$ asymptotically. When $p$ is small, $F(\rho_n)$ can be very large and goes beyond $n F(\mE_\theta)$ which is the optimal QFI allowed using $n$ repeated measurements.   

Above, we show the SQL is achievable when $\trace(G_0 X) \neq 0$. The same argument works analogously when $\trace(G_0 Y) \neq 0$ in which case we can choose the control operator to be $U_k = \exp(-i\frac{\varphi}{2}Y)$ where $\varphi = \sqrt{w/n}$ for some real number $w > 0$. When $\trace(G_1 X) \neq 0$ or $\trace(G_1 Y) \neq 0$, we can choose instead the control operator to be $U_k = \exp(-i\frac{\varphi}{2}X)Z$ or $U_k = \exp(-i\frac{\varphi}{2}Y)Z$ where $\varphi = \sqrt{w/n}$ for some real number $w > 0$. Below, we calculate the QFI in the case where $\trace(G_1 X) \neq 0$ and $U_k = \exp(-i\frac{\varphi}{2}X)Z$ and the calculation for the $\trace(G_1 Y) \neq 0$ case follow analogously. Here 
\begin{equation}
O = \begin{pmatrix}
-1 & 0 & 0 \\
0 & -\cos\varphi & -\sin\varphi \\
0 & -\sin\varphi & \cos\varphi 
\end{pmatrix},
\end{equation}
and $O M = S \Lambda S^{-1}$ where 
{\small
\begin{equation}
    S =  \begin{pmatrix}
        1 & 0 & 0 \\
        0 & 1 & -\frac{\varphi}{2(1-p)} + O(\varphi^3) \\
        0 & \frac{1-2p}{2(1-p)}\varphi + O(\varphi^3) & 1
    \end{pmatrix} ,\quad 
    \Lambda = \begin{pmatrix}
        -1+2p & 0 & 0 \\
        0 & -1+2p + \frac{-p(1-2p)}{2(1-p)}\varphi^2 + O(\varphi^3) & 0 \\
        0 & 0 & 1 - \frac{p\varphi^2}{2(1-p)} + O(\varphi^3) \\
    \end{pmatrix} . 
\end{equation}}
Then 
\begin{gather}
    (OM)^n = S \Lambda^n S^{-1} = 
    \exp\left(-\frac{p}{2(1-p)}w\right) \begin{pmatrix}
        0 & 0 & 0 \\
        0 & 0 & -\frac{\sqrt{w}}{2(1-p)\sqrt{n}} \\
        0 & \frac{-(1-2p)\sqrt{w}}{2(1-p)\sqrt{n}} & 1 
    \end{pmatrix} + O(1/n), \\
    \begin{split}\sum_{k=0}^{n-1} (O M)^{n-k-1} OD (O M)^k &= \sum_{k=0}^{n-1} S \Lambda^{n-k-1} S^{-1} O D S \Lambda^k S^{-1} 
    \\& = \exp\left(-\frac{p}{2(1-p)}w\right) \frac{p\trace(G_1X)\sqrt{w n}}{1-p} \begin{pmatrix}
        0 & 0 & 0 \\
        0 & 0 & 0 \\
        0 & 0 & 1 
    \end{pmatrix} + O(1). 
    \end{split}
\end{gather}
Similar to \eqref{eq:control-QFI-1}, we will have 
\begin{equation}
\label{eq:control-QFI-2}
    F(\rho_n) = \frac{\frac{p^2}{(1-p)^2}w}{z_0^{-2}\exp\left(\frac{p}{1-p}w\right)-1} \trace(G_1 X)^2  n + O(\sqrt{n}). 
\end{equation}
Taking $z_0 = 1$ and $w$ a small number, $F(\rho_n)$ can be arbitrarily close to $\frac{p}{1-p} \trace(G_1 X)^2 n$ asymptotically. 
\end{proof}

\section{Dephasing-class channels: Unachievability of the SQL}
\label{app:dephasing-2}

We showed in \appref{app:dephasing} two types of QFI upper bounds for estimating an unknown parameter in dephasing-class channels using unital and CPTP controls. In \appref{app:sql}, we showed the upper bound of linear scaling is achievable using unitary controls for dephasing-class channels satisfying the RGNKS condition. Here we show that in contrast to the previous case, when the RGNKS condition is violated, the QFI is at most a constant using unital controls. 

\begin{theorem}[QFI upper bound for estimating dephasing-class channels using ancilla-free sequential strategies with unital controls]  
\label{thm:dephasing-upper-2}
 Consider estimating an unknown parameter $\theta$ in a dephasing-class channel $\mE_\theta$ of the most general form in \eqref{eq:one-parameter-dephasing-qubit-channel} using the ancilla-free sequential strategy and assume arbitrary unital controls $\mU_k$. Assume $G_0,G_1 \propto Z$, i.e., the RGNKS condition is violated. Then the channel QFI of \eqref{eq:entire-channel} for any fixed $\mE_\theta$ satisfies 
    \begin{equation}
    \label{eq:unital-dephasing-II}
    \sup_{\{\mU_k\}_{k=1}^n} F(\mE^{(n)}_\theta) = \sup_{\{\mU_k\}_{k=1}^n} F(\mU_n\circ\mE_\theta\circ\cdots\circ\mU_1\circ\mE_\theta) = O(1). 
    \end{equation}
\end{theorem}

\begin{proof}
    For $k = 0,\ldots, n$, let 
\begin{equation}
\rho_k = \mU_k \circ \mE_{\theta} \circ \mU_{k-1} \circ \mE_{\theta} \circ \cdots \circ \mU_{1} \circ \mE_{\theta}(\rho_0) =: \frac{\id + x_k X + y_k Y + z_k Z}{2}. 
\end{equation}
Then we have
\begin{equation}
    \vv_k = T_k M \vv_{k-1},\quad 
    \dot\vv_k = T_k D \vv_{k-1} + T_k M \dot\vv_{k-1},
\end{equation}
where $T_k$ describes $\mU_k$ by $\mC_k\left(\frac{1}{2}(\id + \vw\cdot\vsig)\right) = \frac{1}{2}(\id + (T_k\vw)\cdot\vsig)$ and
\begin{equation}
    M = \begin{pmatrix}
        1-2p & 0 & 0 \\
        0 & 1-2p & 0\\
        0 & 0 & 1
    \end{pmatrix},\quad 
    D = \begin{pmatrix}
        -2\dot p & - \trace(G_- Z) & 0\\
        \trace(G_- Z) & -2\dot p & 0 \\
        0 & 0 & 0
    \end{pmatrix},
\end{equation}
where we use the assumption that $G_{0,1} \propto Z$. We also assume $T_0 = \id$. Note that $\norm{T_k} \leq 1$ in order for the channel to be CPTP. We have 
\begin{align}
    \norm{\dot\vv_k}^2 &= \norm{T_k D \vv_{k-1}+ T_k M \dot\vv_{k-1}}^2 \leq \norm{D \vv_{k-1}+M \dot\vv_{k-1}}^2 \\
    &= \left(\left((1-2p)\dot x_{k-1} - \trace(G_- Z) y_{k-1} - 2\dot p x_{k-1} \right)^2 + \left((1-2p)\dot y_{k-1} + \trace(G_- Z) x_{k-1}  - 2\dot p y_{k-1}\right)^2 + \dot z_{k-1}^2\right).\nonumber 
\end{align}
Note that $\norm{\dot\vv_{k-1}}^2 = \dot x_{k-1}^2 + \dot y_{k-1}^2 + \dot z_{k-1}^2$. 
Then we have 
\begin{align}
    \norm{\dot\vv_k}^2 - \norm{\dot\vv_{k-1}}^2 
    & \leq \frac{(- \trace(G_- Z) y_{k-1} - 2\dot p x_{k-1} )^2 + (\trace(G_- Z) x_{k-1} - 2\dot p y_{k-1})^2}{4p(1-p)}\\
    & \leq \frac{(\trace(G_- Z)^2+ 4\dot p^2) (x_{k-1}^2 + y_{k-1}^2)}{2p(1-p)}.\label{eq:iteration-1}
\end{align}
Then note that $\norm{\vv_{k-1}}^2 = x_{k-1}^2 + y_{k-1}^2 + z_{k-1}^2$, $\norm{M\vv_{k-1}}^2 \leq (1-2p)^2x_{k-1}^2 + (1-2p)^2y_{k-1}^2 + z_{k-1}^2$, and $\norm{\vv_{k}} = \norm{T_k M\vv_{k-1}} \leq \norm{M\vv_{k-1}}$. We have  
\begin{equation}
     {x_{k-1}^2 + y_{k-1}^2} = \frac{\norm{\vv_{k-1}}^2 - \norm{M \vv_{k-1}}^2}{1-(1-2p)^2} \leq \frac{\norm{\vv_{k-1}}^2 - \norm{\vv_{k}}^2}{1-(1-2p)^2}.
\end{equation}
Then from \eqref{eq:iteration-1}, we have 
\begin{equation}
    \norm{\dot\vv_k}^2 - \norm{\dot\vv_{k-1}}^2 \leq \frac{\trace(G_- Z)^2+ 4\dot p^2}{8p^2(1-p)^2} \left(\norm{\vv_{k-1}}^2 - \norm{\vv_{k}}^2\right). 
\end{equation}
Summing the above inequality from $k=1$ to $n$, we have 
\begin{equation}
    \norm{\dot\vv_n}^2\leq \frac{\trace(G_- Z)^2+ 4\dot p^2}{8p^2(1-p)^2} \left(\norm{\vv_{0}}^2 - \norm{\vv_{n}}^2\right) \leq \frac{\trace(G_- Z)^2+ 4\dot p^2}{8p^2(1-p)^2} \left(1 - \norm{\vv_{n}}^2\right). 
\end{equation}
When $\norm{\vv_n} = 1$, $\dot\vv_n = 0$ and we have $F(\rho_n) = 0$. Otherwise, the QFI satisfies (let the diagonalization of $\rho_n$ be $\lambda_0\ket{\psi_0}\bra{\psi_0} + \lambda_1 \ket{\psi_1}\bra{\psi_1}$)
\begin{align}
    F(\rho_n) 
    &= \frac{\abs{\bra{\psi_0}\dot\rho_n\ket{\psi_0}}^2}{\lambda_0} + \frac{\abs{\bra{\psi_1}\dot\rho_n\ket{\psi_1}}^2}{\lambda_1} + 2 \abs{\bra{\psi_0}\dot\rho_n\ket{\psi_1}}^2\\
    &\leq \norm{\dot\vv_n}^2 \left( \frac{1}{\frac{1+\norm{\vv_n}}{2}} + \frac{1}{\frac{1-\norm{\vv_n}}{2}} + 2 \right) \leq \frac{8\norm{\dot\vv_n}^2}{1 - \norm{\vv_n}^2} \leq  \frac{\trace(G_- Z)^2+ 4\dot p^2}{p^2(1-p)^2} . \label{eq:qfi-upper-purity}
\end{align}
\eqref{eq:unital-dephasing-II} is then proven. 
\end{proof}

\thmref{thm:dephasing-upper-2} shows a gap between dephasing-class channels that satisfy the RGNKS condition and that violate the RGNKS condition. Specifically, when RGNKS is satisfied, the SQL is attainable using unitary controls; but when RGNKS is violated, the QFI has a constant upper bound using unital controls. 

One might wonder whether general CPTP controls allows the SQL to be achieved even in the second case. We don't have a definite answer to this question. However, below we will present a lemma showing that even with CPTP controls, the derivative of the output quantum state when RGNKS fails has at most a constant scaling, i.e. $\norm{\dot\vv_n} \leq \text{constant}$, in contrast to the case when RGNKS holds and the derivative operator grows with a square-root scaling (as in \eqref{eq:derivative-rho}). Note that this observation itself does not directly imply the constant scaling of QFI for dephasing-class channels violating the RGNKS condition, but it implies that if there exist a set of CPTP controls that achieve the SQL then the control must have a unique feature that makes the Bloch vector of the output state asymptotically converges to the Bloch surface (see \eqref{eq:qfi-upper-purity}, the QFI has a constant upper bound if $\norm{\vv_n}$ is bounded away from $1$). 

\begin{proposition}[Constant scaling of the derivative of output states from dephasing-class channels using ancilla-free sequential strategies]
\label{prop:dephasing-derivative}
Consider estimating an unknown parameter $\theta$ in a dephasing-class channel $\mE_\theta$ of the most general form in \eqref{eq:one-parameter-dephasing-qubit-channel} using the ancilla-free sequential strategy and assume arbitrary CPTP controls $\mC_k$. Assume $G_0,G_1 \propto Z$, i.e., the RGNKS condition is violated. Then there exists a constant $c$ such that 
\begin{equation}
\norm{\dot\rho_{\theta,n}} \leq c,
\end{equation}
where $\rho_{\theta,n} = \mC_k \circ \mE_{\theta} \circ \mC_{k-1} \circ \mE_{\theta} \circ \cdots \circ \mC_{1} \circ \mE_{\theta} (\rho_0)$ for any initial state $\rho_0$. 
\end{proposition}

\begin{proof}
For $k = 0,\ldots, n$, let 
\begin{equation}
\rho_k = \mC_k \circ \mE_{\theta} \circ \mC_{k-1} \circ \mE_{\theta} \circ \cdots \circ \mC_{1} \circ \mE_{\theta}(\rho_0) =: \frac{\id + x_k X + y_k Y + z_k Z}{2}. 
\end{equation}
Then we have
\begin{equation}
    \vv_k = \vt_k + T_k M \vv_{k-1},\quad 
    \dot\vv_k = T_k D \vv_{k-1} + T_k M \dot\vv_{k-1},
\end{equation}
where $(T_k,\vt_k)$ describes $\mC_k$ by $\mC_k\left(\frac{1}{2}(\id + \vw\cdot\vsig)\right) = \frac{1}{2}(\id + (\vt_k + T_k\vw)\cdot\vsig)$ and
\begin{equation}
    M = \begin{pmatrix}
        1-2p & 0 & 0 \\
        0 & 1-2p & 0\\
        0 & 0 & 1
    \end{pmatrix},\quad 
    D = \begin{pmatrix}
        -2\dot p & - \trace(G_- Z) & 0\\
        \trace(G_- Z) & -2\dot p & 0 \\
        0 & 0 & 0
    \end{pmatrix},
\end{equation}
where we use the assumption that $G_{0,1} \propto Z$. We also assume $T_0 = \id$ and $\vt_0 = 0$.  Note that $\norm{\vt_k  + T_k \vw} \leq \norm{\vw}$ for any unit vector $\vw$ in order for $(\vt_k,T_k)$ to describe a valid quantum channel. Note that $\norm{T_k} \leq 1$. We have 
\begin{align}
    \norm{\dot\vv_k}^2 &= \norm{T_k D \vv_{k-1}+ T_k M \dot\vv_{k-1}}^2 \leq \norm{T_k}^2\norm{D \vv_{k-1}+M \dot\vv_{k-1}}^2 \\
    &= \norm{T_k}^2\left(\left((1-2p)\dot x_{k-1} - \trace(G_- Z) y_{k-1} - 2\dot p x_{k-1} \right)^2 + \left((1-2p)\dot y_{k-1} + \trace(G_- Z) x_{k-1}  - 2\dot p y_{k-1}\right)^2 + \dot z_{k-1}^2\right).\nonumber 
\end{align}
Note that $\norm{\dot\vv_{k-1}}^2 = \dot x_{k-1}^2 + \dot y_{k-1}^2 + \dot z_{k-1}^2$. 
Then we have 
\begin{align}
    \norm{\dot\vv_k}^2 - \norm{T_k}^2\norm{\dot\vv_{k-1}}^2 
    & \leq \norm{T_k}^2 \frac{(- \trace(G_- Z) y_{k-1} - 2\dot p x_{k-1} )^2 + (\trace(G_- Z) x_{k-1} - 2\dot p y_{k-1})^2}{4p(1-p)}\\
    & \leq \norm{T_k}^2 \frac{(\trace(G_- Z)^2+ 4\dot p^2) (x_{k-1}^2 + y_{k-1}^2)}{2p(1-p)}.\label{eq:iteration-1-0}
\end{align}

To prove $\norm{\vv_k} \leq c$ for some constant $c$, we use 
\begin{equation}
     {x_{k-1}^2 + y_{k-1}^2} = \frac{\norm{\vv_{k-1}}^2 - \norm{M \vv_{k-1}}^2}{1-(1-2p)^2} \leq \frac{\norm{\vv_{k-1}}^2 - \norm{\vv_k}^2 + 1 - \norm{T_{k-1}}^2}{1-(1-2p)^2},\label{eq:iteration-2} 
\end{equation}
which holds for general CPTP controls. 
In the last step, we use \eqref{eq:qubit-channel-ineq} in \propref{prop:bloch} (a mathematical property of qubit channels that we will prove later) which leads to 
\begin{equation}
    \norm{\vv_{k}}^2 - \norm{M \vv_{k-1}}^2 = \norm{\vt_k + T_kM\vv_{k-1}}^2 - \norm{M \vv_{k-1}}^2 \leq 1 - \norm{T_k}^2. 
\end{equation}
Combining \eqref{eq:iteration-1-0} and \eqref{eq:iteration-2}, we have
\begin{gather}
    \norm{\dot\vv_{k}}^2 - \norm{T_k}^2 \norm{\dot \vv_{k-1}}^2 \leq \norm{T_k}^2 (\norm{\vv_{k-1}}^2 - \norm{\vv_k}^2 + 1 - \norm{T_{k-1}}^2) \, \frac{\trace(G_- Z)^2+ 4\dot p^2}{8p^2(1-p)^2},\\
    \Rightarrow~~ \norm{T_n}^2\cdots\norm{T_{k+1}}^2\norm{\dot\vv_{k}}^2 -  \norm{T_n}^2\cdots\norm{T_{k+1}}^2\norm{T_k}^2 \norm{\dot \vv_{k-1}}^2 \nonumber\\ 
    \quad \quad \leq  \norm{T_n}^2\cdots\norm{T_{k+1}}^2\norm{T_k}^2 (\norm{\vv_{k-1}}^2 - \norm{\vv_k}^2 + 1 - \norm{T_{k-1}}^2) \, \frac{\trace(G_- Z)^2+ 4\dot p^2}{8p^2(1-p)^2}, \\
    \Rightarrow~~ \norm{\dot \vv_{n}}^2 \leq \left( - \norm{T_n}^2 \norm{\vv_n}^2 + \sum_{k=1}^{n} \norm{T_n}^2\cdots\norm{T_{k+1}}^2\norm{T_k}^2 (1 - \norm{T_{k-1}}^2)(1 + \norm{\vv_{k-1}}^2) \right)\frac{\trace(G_- Z)^2+ 4\dot p^2}{8p^2(1-p)^2} \\
     \leq \norm{T_n}^2 \frac{\trace(G_- Z)^2+ 4\dot p^2}{4p^2(1-p)^2} \leq \frac{\trace(G_- Z)^2+ 4\dot p^2}{4p^2(1-p)^2} =: c, \quad \;\;\,
\end{gather}
where in the last step we sum the inequality over $k = 1$ to $n$.
\end{proof}

Below, we provide the mathematical property of qubit quantum channels that was used above in the proof of \propref{prop:dephasing-derivative}. Note that this property is a simple mathematical relation that applies to any qubit quantum channels and might have other applications beyond quantum metrology. 
\begin{proposition}[Property of Pauli transfer matrices of qubit channels]
\label{prop:bloch}
For any qubit quantum channel $\mE$ represented by $(\vt,T)$ through  
\begin{equation}
    \mE\left(\frac{1}{2}(\id + \vw\cdot\vsig)\right) = \frac{1}{2}(\id + (\vt + T\vw)\cdot\vsig), \quad \forall \vw \in \bR^3, 
\end{equation}
we have 
\begin{equation}
\label{eq:qubit-channel-ineq-0}
\norm{\vt}^2 \leq (1 - \sigma_{\min}(T)^2)(1 - \norm{T}^2),
\end{equation} 
where $\sigma_{\min}(\cdot)$ represents the smallest singular value of $T$. 
As a corollary, we have for all $\vw \in \bR^3$, 
   \begin{equation}
   \label{eq:qubit-channel-ineq}
       \norm{\vt + T\vw}^2 \leq 1 - \norm{T}^2 + \norm{\vw}^2. 
   \end{equation}
\end{proposition}

\begin{proof}
    We first prove \eqref{eq:qubit-channel-ineq-0}. To simplify calculations, we first note that $T$ can be decomposed into $T = R_1 \Lambda R_2$ such that $S$ and $R_{1,2}$ are rotation matrices (orthogonal matrices with determinant equal to 1) and $\Lambda$ is a diagonal matrix whose diagonal elements are $(\lambda_1,\lambda_2,\lambda_3)$. ($\lambda_{1,2,3}$ can be negative and we assume $\abs{\lambda_1} \geq \abs{\lambda_2} \geq \abs{\lambda_3}$ and $1 > \abs{\lambda_1} > 0$ without loss of generality.) Then there always exist unitaries $U$ and $V$ such that  
    \begin{equation}
        U \mE  \left(V\frac{1}{2}(\id + \vw\cdot\vsig)V^\dagger\right) U^\dagger = \frac{1}{2}\left(\id + (R_1^T \vt + R_1^T T R_2^T \vw)\cdot\vsig\right) = \frac{1}{2}\left(\id + (R_1^T \vt + \Lambda \vw)\cdot\vsig\right). 
    \end{equation}
    Note that \eqref{eq:qubit-channel-ineq-0} is equivalent to $\norm{R_1^T \vt}^2 = \norm{\vt}^2 \leq (1 - \lambda_3^2)(1 - \lambda_1^2) = (1 - \sigma_{\min}(\Lambda)^2)(1 - \norm{\Lambda}^2)$. Thus it will be sufficient to prove \eqref{eq:qubit-channel-ineq-0} for channel $\mE'(\cdot) = U\mE(V(\cdot)V^\dagger)U^\dagger$ represented by $(\vs,\Lambda)$, where we let $\vs = R_1^T \vt$ for simplicity. To obtain the relation between $\vs$ and $\Lambda$ such that $\mE'$ is a CPTP map, we write down the Choi representation of $\mE'$ which gives 
    \begin{align}
        \Sigma = (\mE'\otimes \id)\left(\sum_{i=0,1} \ket{i}\ket{i}\sum_{j=0,1} \bra{j}\bra{j}\right) &= \frac{1}{2}(\mE'\otimes \id)\left(\id\otimes \id + X\otimes X^T + Y\otimes Y^T + Z\otimes Z^T\right) \\
        &= \frac{1}{2}\left((\id+\vs\cdot\vsig)\otimes \id + \lambda_1 X\otimes X^T + \lambda_2 Y\otimes Y^T + \lambda_3 Z\otimes Z^T\right).
    \end{align}
    $\Sigma$ being postivie semidefinite is a necessary and sufficient condition for $\mE'$ to be a CPTP map. 
    Calculations show that when we fix $\lambda_{1,2,3}$ and view $\Sigma$ as a function of $s_{1,2,3}$, the characteristic polynomial of $\Sigma$ is 
    \begin{equation}
    \label{eq:cptp-condition}
        \det(\Sigma - \kappa \id) = \kappa^4 - 4\kappa^3 + c_2(\norm{\vs}) \kappa^2 + c_1(\norm{\vs}) \kappa + c_0(\norm{\vs}) - 4 \sum_{i=1}^3\lambda_i^2 s_i^2,
        \end{equation}
    where we didn't write out specific expressions of $c_{0,1,2}$ which are functions of $\norm{\vs}$ only. Our goal is to show for any $(\vs,\Lambda)$ such that $\Sigma$ is positive semidefinite, $\norm{\vs}^2 \leq (1 - \lambda_3^2)(1 - \lambda_1^2)$. To simplify the proof further, we observe from \eqref{eq:cptp-condition} that for any $\vs$, if $(\vs,\Lambda)$ is a CPTP map, $((0,0,\norm{\vs}),\Lambda)$ must also be a CPTP map. The reason is $\det(\Sigma - \kappa \id)$ will only be increased by a non-negative constant number when we replace $\vs$ with $(0,0,\norm{\vs})$, and thus eigenvalues of $\det(\Sigma - \kappa \id)$ stay non-negative after the replacement if they are non-negative before the replacement. Moreover, when $\vs = (0,0,\norm{\vs})$, we have the following analytical solution of $\det(\Sigma - \kappa \id)$ which gives 
    \begin{equation}
        \kappa = 1 - \lambda_3 \pm \sqrt{\norm{\vs}^2 + (\lambda_1-\lambda_2)^2},\; 1 + \lambda_3 \pm \sqrt{\norm{\vs}^2 + (\lambda_1+\lambda_2)^2}. 
    \end{equation}
    (Note that when the argument in $\sqrt{(\cdot)}$ is negative, the corresponding solution should be deleted from the above list.) From the discussion above, we know that all solutions of $\kappa$ must be non-negative. Specifically, $\kappa_1 = 1 - \lambda_3 - \sqrt{\norm{\vs}^2 + (\lambda_1-\lambda_2)^2}$ and $\kappa_2 = 1 + \lambda_3 - \sqrt{\norm{\vs}^2 + (\lambda_1+\lambda_2)^2}$ must be non-negative, which leads to 
    \begin{equation}
        \norm{\vs}^2 \leq \min\{(1-\lambda_3)^2 - (\lambda_1-\lambda_2)^2,(1+\lambda_3)^2 - (\lambda_1+\lambda_2)^2\}. 
    \end{equation}
    We will prove \eqref{eq:qubit-channel-ineq-0} by considering two cases, (i) $\lambda_3 \geq \lambda_1^2$. In this case, $(1-\lambda_3)^2 - (\lambda_1-\lambda_2)^2$ is always no larger than $(1+\lambda_3)^2 - (\lambda_1+\lambda_2)^2$ and then we have $\norm{\vs}^2 \leq (1-\lambda_3)^2 \leq (1-\lambda_3^2)(1-\lambda_1^2)$. (ii) $\lambda_3 < \lambda_1^2$. In this case, $\min\{(1-\lambda_3)^2 - (\lambda_1-\lambda_2)^2,(1+\lambda_3)^2 - (\lambda_1+\lambda_2)^2\}$ is that largest when $\lambda_2 = \lambda_3/\lambda_1$ and taking $\lambda_2 = \lambda_3/\lambda_1$, we have $\norm{\vs}^2 \leq 1 + \lambda_3^2 - \lambda_1^2 - (\lambda_3/\lambda_1)^2 \leq (1-\lambda_3^2)(1-\lambda_1^2)$.  

Now we prove \eqref{eq:qubit-channel-ineq} using \eqref{eq:qubit-channel-ineq-0}.
As shown above, there always exists unitaries $U$ and $V$ such that $\mE'(\cdot) = U\mE(V(\cdot)V^\dagger)U^\dagger$ is represented by $(\vs,\Lambda)$ and $\vs=R_1^T$ and $\Lambda = R_1^T T R_2^T$ is diagonal for some rotation matrices $R_{1,2}$ and $\abs{\lambda_1}\geq\abs{\lambda_2}\geq\abs{\lambda_3}$. We assume $\abs{\lambda_1} < 1$ as otherwise $\vt = 0$ and the inequality trivially holds. It is sufficient to prove \eqref{eq:qubit-channel-ineq} for $(\vs,\Lambda)$, i.e., for any $\vw \in \bR^3$, 
\begin{equation}
\label{eq:qubit-channel-ineq-simp}
    \norm{\vs+\Lambda \vw}^2 \leq 1 - \lambda_1^2 + \norm{\vw}^2. 
\end{equation}
Furthermore, \eqref{eq:qubit-channel-ineq-simp} is true because 
\begin{align}
    \norm{\vs+\Lambda \vw}^2 &= \sum_{i=1}^3 (s_i + \lambda_i w_i)^2 \leq \sum_{i=1}^3 \left( (1-\lambda_i^2) + (1-(1-\lambda_i^2))\right)\left(\frac{s_i^2}{1-\lambda_i^2} + w_i^2\right)\\
    &= \sum_{i=1}^3 \frac{s_i^2}{1-\lambda_i^2} + \norm{\vw}^2 \leq \frac{\norm{\vs}^2}{1-\lambda_3^2} + \norm{\vw}^2 \leq 1 - \lambda_1^2 + \norm{\vw}^2,
\end{align}
where we use the Cauchy-Schwarz inequality in the first inequality and \eqref{eq:qubit-channel-ineq-0} in the last inequality. 
\end{proof}

\section{Strictly Contractive Channels: Unachievability of the SQL}
\label{app:contractive}

Here we show a constant QFI upper bound in strategies (d) (see \figdref{fig:strategies}) when the one-parameter channel is strictly constractive. We will first prove \propref{prop:contraction} which implies the contraction coefficient of strictively contractive channels with respect to QFI is smaller than one and then use the relation between the QFI and the Bures distance to prove our final result. Note that \propref{prop:contraction} itself is a mathematical property of QFI and might have applications beyond quantum channel estimation. 

\begin{proposition}[Contraction coefficient of qubit channels with respect to QFI] Given a strictly contractive qubit channel $\mN(\cdot)$. 
\begin{equation}
    \eta(\mN):= \sup_{\sigma_\theta} \frac{F(\mN(\sigma_\theta))}{F(\sigma_\theta)} \leq \max_{\text{\rm Hermitian}\,A:\trace(A)=0}\frac{\norm{\mN(A)}_1}{\norm{A}_1}< 1, 
\end{equation}
where $\sigma_\theta$ is an arbitrary one-parameter quantum state. 
In other words, the contraction coefficient with respect to QFI is no larger than the contraction coefficient with respect to the trace distance. 
\label{prop:contraction}
\end{proposition}

\begin{proof}
To prove this lemma, we use the fact that the QFI is equal to the optimal classical FI optimized over projective measurements. In particular, for qubits, there exist~\cite{braunstein1994statistical} an optimal binary measurement defined by measurement operators $B = \ket{\phi_0}\bra{\phi_0}$ and $\id - B = \ket{\phi_1}\bra{\phi_1}$ that satisfies 
\begin{align}
    F(\mN(\sigma_\theta)) 
    &= F(\mN(\sigma_\theta),\{B,\id-B\}) \\
    &= \trace(B \mN(\dot\sigma_\theta))^2 \left( \frac{1}{\trace(B \mN(\sigma_\theta))} + \frac{1}{1 - \trace(B \mN(\sigma_\theta))}\right)\\
    &= F(\sigma_\theta,\{\mN^\dagger(B),\id-\mN^\dagger(B)\}) \leq F^{\bP}(\sigma_\theta,\{\mN^\dagger(B),\id-\mN^\dagger(B)\}),
\end{align}
where $\mN^\dagger$ is the dual map of $\mN$ and $F^\bP(\sigma_\theta,\{M_i\})$ is called quantum preprocessing-optimized FI~\cite{len2021quantum,zhou2023optimal} defined by
\begin{equation}
    F^\bP(\sigma_\theta,\{M_i\}) = \sup_{\text{CPTP:}\,\mR} F(\mR(\sigma_\theta),\{M_i\}). 
\end{equation}
Refs.~\cite{zhou2023optimal} showed that for binary measurements on a qubit and any $\sigma_\theta$
\begin{equation}
    \frac{F^\bP(\sigma_\theta,\{M,\id-M\})}{F(\sigma_\theta)} \leq 1 - \big(\sqrt{m_1m_2} + \sqrt{(1-m_1)(1-m_2)} \big)^2,
\end{equation}
where $m_{1,2}$ are eigenvalues of $M$. In our case, $M = \mN^\dagger(\ket{\phi_0}\bra{\phi_0})$. Up to unitary rotations before and after $\mN$ (we can see below unitary rotations before and after $\mN$ doesn't affect our derivations), we can assume 
\begin{equation}
    \mN\left(\frac{\id + \vw \cdot \vsig}{2}\right) = \frac{\id + (\vt + T\vw) \cdot \vsig}{2},\quad T = \begin{pmatrix}
        \lambda_1 & 0 & 0 \\
        0 & \lambda_2 & 0\\
        0 & 0 & \lambda_3 \\
    \end{pmatrix},
\end{equation}
for all $\vw$ and $\abs{\lambda_1} \geq \abs{\lambda_2} \geq \abs{\lambda_3}$ and $ \abs{\lambda_1} = \max_{\text{\rm Hermitian}\,A:\trace(A)=0}\frac{\norm{\mN(A)}_1}{\norm{A}_1}$. Let $\ket{\phi_0}\bra{\phi_0} = \frac{\id + \vw_0 \cdot \vsig}{2}$ where $\vw_0$ is a unit vector, we have 
\begin{equation}
    \mN^\dagger\left(\frac{\id + \vw_0 \cdot \vsig}{2}\right) = \frac{(1 + \vt\cdot\vw_0)\id + (T\vw_0)\cdot \vsig}{2}. 
\end{equation}
It implies $\delta := \abs{m_1 - m_2} \leq \abs{\lambda_1}$. Without loss of generality, assume $m_1 > m_2$ and $m_1 = m_2 + \delta$ we write 
\begin{equation}
    1 - \big(\sqrt{m_1m_2} + \sqrt{(1-m_1)(1-m_2)} \big)^2 = 1 - \big(\sqrt{(m_2+\delta)m_2} + \sqrt{(1-m_2)(1-m_2-\delta)}\big)^2,
\end{equation}
\sloppy and $0\leq m_2 \leq 1 - \delta$. Note that $\frac{\partial}{\partial m_2} \sqrt{(m_2+\delta)m_2} + \sqrt{(1-m_2)(1-m_2-\delta)} = \frac{1}{2}\bigg(\sqrt{\frac{m_2}{m_2+\delta}}+\sqrt{\frac{m_2+\delta}{m_2}}-\sqrt{\frac{1-m_2}{1-m_2-\delta}}-\sqrt{\frac{1-m_2-\delta}{1-m_2}}\bigg)$ is monotonically decreasing in $m_2 \in [0,1-\delta]$. Then the minimum value of $\sqrt{m_1m_2} + \sqrt{(1-m_1)(1-m_2)}$ as a function of $m_2$ must be reached at either $m_2 = 0$ or $m_2 = 1- \delta$. It indicates that $\sqrt{m_1m_2} + \sqrt{(1-m_1)(1-m_2)} \leq \sqrt{1 - \delta}$. Finally, combining all discussions above, we have 
\begin{align}
    \frac{F(\mN(\sigma_\theta))}{F(\sigma_\theta)} &\leq \frac{F^\bP(\sigma_\theta,\{M_i\})}{F(\sigma_\theta)} \leq 1 - (\sqrt{m_1m_2} + \sqrt{(1-m_1)(1-m_2)})^2 \\ &\leq 1 - (\sqrt{1-\delta})^2 = \delta \leq \abs{\lambda_1} = \max_{\text{\rm Hermitian}\,A:\trace(A)=0}\frac{\norm{\mN(A)}_1}{\norm{A}_1}. 
\end{align}
The lemma is then proven. 
\end{proof}

\begin{theorem}[QFI upper bound for estimating strictly contractive channels using ancilla-free sequential strategies]  
\label{thm:contractive-upper}
 Consider estimating an unknown parameter $\theta$ in a strictly contractive channel $\mE_\theta$ using the ancilla-free sequential strategy and assume arbitrary general CPTP controls $\mC_k$. Then the channel QFI of \eqref{eq:entire-channel} for any fixed $\mE_\theta$ satisfies 
    \begin{equation}
    \sup_{\{\mC_k\}_{k=1}^n} F(\mE^{(n)}_\theta) = \sup_{\{\mC_k\}_{k=1}^n} F(\mC_n\circ\mE_\theta\circ\cdots\circ\mC_1\circ\mE_\theta) \leq \frac{F(\mE_\theta)}{(1-\sqrt{\eta(\mE_\theta)})^2} = O(1). 
    \end{equation}
    Additionally, the constant scaling of the QFI is achievable (by letting $\mC_{n-1}= \trace(\cdot)\,\psi $ to be a replacement channel such that $F(\mE_\theta(\psi)) > 0$).
\end{theorem}

\begin{proof}
When $F(\mE_\theta) = 0$, we must have $\dot\mE_\theta = 0$ and then the theorem trivially holds. Therefore, we assume $F(\mE_\theta) > 0$ throughout our discussion. 
Let the initial quantum state be $\rho_0$. The state at the $k$-th step is
\begin{equation}
    \rho_{\theta,k} = \mC_k\circ\mE_\theta\circ\cdots\circ\mC_1\circ\mE_\theta(\rho_0). 
\end{equation}
Let $\mE_\theta$ be the strictly contractive channel under our consideration. It was assumed (without loss of generality) that $\theta = 0$ is the true value of $\theta$ and we also assumed $\mE_\theta$ is strictly contractive in a neighborhood of $0$. Using \propref{prop:contraction}, we can assume there exists $\epsilon,\epsilon' > 0$ such that the QFI contraction coefficient of $\mE_\delta$ is smaller than $1-\epsilon'$ when $\delta \in [-\epsilon,\epsilon]$, i.e., $F(\mE_{\delta}(\rho_\theta))/F(\rho_\theta) \leq 1-\epsilon'$ for any $\rho_\theta$ as a function of the unknown parameter $\theta$ and any $\delta \in [-\epsilon,-\epsilon]$.

Then for any $n > 1$, we have 
\begin{equation}
\label{eq:qfi-bures}
    F(\rho_{\theta,n}) = \,\lim_{d\theta \rightarrow 0} \frac{d^2_{\rm B}(\rho_{\theta+d\theta,n},\rho_{\theta-d\theta,n})}{(d\theta)^2}, \text{~~and equivalently~~} \sqrt{F(\rho_{\theta,n})} = \,\lim_{d\theta \rightarrow 0} \frac{d_{\rm B}(\rho_{\theta+d\theta,n},\rho_{\theta-d\theta,n})}{|d\theta|}
\end{equation}
which is a mathematical relation between the QFI and the Bures distance~\cite{zhou2019exact,hubner1992explicit}. 
Using the triangular inequality of the Bures distance, we have 
\begin{equation}
\label{eq:triangular-bures}
    d_{\rm B}(\rho_{\theta+d\theta,n},\rho_{\theta-d\theta,n})|_{\theta = 0} \leq d_{\rm B}(\rho^{[n]},\rho^{[n-1]}) + d_{\rm B}(\rho^{[n-1]},\rho^{[n-2]}) + \cdots + d_{\rm B}(\rho^{[1]},\rho^{[0]})
\end{equation}
where 
\begin{equation}
    \rho^{[k]}:= \mC_n\circ\mE_{d\theta}\circ\cdots\circ\mC_{k+1}\circ\mE_{d\theta}\circ\mC_{k}\circ\mE_{-d\theta}\circ\cdots\circ\mC_1\circ\mE_{-d\theta},
\end{equation}
and $\rho^{[0]} = \rho_{d\theta,n}$, and $\rho^{[n]} = \rho_{-d\theta,n}$. Here $d\theta$ is an arbitrary number in $[-\epsilon,\epsilon]$. 
Furthermore, for any $b > 0$, when $d\theta$ is small enough, it holds that 
\begin{align}
    &\quad\; d^2_{\rm B}(\rho^{[k+1]},\rho^{[k]})\\ 
    &= \left(\lim_{\delta\rightarrow 0}\frac{d^2_{\rm B}\big(\mC_{n}\circ\mE_{d\theta} \circ\cdots\circ\mE_{d\theta} \circ\mC_{k+1}\circ\mE_{-\delta}(\rho_{0,k}),\mC_{n}\circ\mE_{d\theta} \circ\cdots \circ \mE_{d\theta} \circ \mC_{k+1}\circ\mE_{\delta}(\rho_{0,k})\big)}{\delta^2}\right)(d\theta)^2 + o(d\theta^2)\\
    &= F( (\mC_{n}\circ\mE_{d\theta} \circ\cdots \circ\mC_{k+2}\circ\mE_{d\theta})\circ(\mC_{k+1}\circ\mE_{\theta})(\rho_{0,k})) (d\theta)^2 + o(d\theta^2)\\
    &\leq (1-\epsilon')^{n-k-1} F(\mE_{\theta}) (d\theta)^2 + o((d\theta)^2) \leq (1-\epsilon')^{n-k-1} F(\mE_{\theta}) (1+b) (d\theta)^2,
    \label{eq:qfi-exp-decay}
\end{align}
where in the first step we take the Taylor expansion of the $k$-step $\mE_{d\theta}$ over $d\theta$ and we replace $d\theta$ with $\delta$ in the second line to distinguish it from other $d\theta$'s that were not involved in the Taylor expansion. In the second step, we simply use the relation between the QFI and the Bures distance. Note that the notation $F( (\mC_{n}\circ\mE_{d\theta} \circ\cdots \circ\mC_{k}\circ\mE_{d\theta})\circ\mC_{k}\circ\mE_{\theta}(\rho_{0,k}))$ means the QFI of the quantum state in the parenthesis when $\theta$ is viewed as the parameter to be estimated (treating $d\theta$ as a constant).  In the third step, we use the data-processing inequality and the fact that the QFI contraction coefficient of $\mE_{d\theta}$ when $d\theta \in [-\delta,\delta]$ is no larger than $1-\epsilon'$, as discussed above. In the final step, we note that for any $b > 0$, we can always pick $d\theta$ small enough (the choice of $d\theta$ may depends on $n$ which is a fixed finite constant here) such that for any $k$, $o((d\theta)^2)$ in the above inequality is smaller than $b (1-\epsilon')^{n-k} F(\mE_{\theta}) d\theta^2$. Combining \eqref{eq:qfi-exp-decay} with \eqref{eq:triangular-bures}, we have 
\begin{align}
    d_{\rm B}(\rho_{\theta+d\theta,n},\rho_{\theta-d\theta,n})|_{\theta = 0} 
    &\leq \sum_{k=0}^{n-1} d_{\rm B}(\rho^{[n]},\rho^{[n-1]}) \leq \sum_{k=0}^{n-1} \sqrt{(1-\epsilon')^{n-k-1} F(\mE_{\theta}) (1+b) (d\theta)^2} \\ 
    &= \frac{1-\sqrt{1-\epsilon'}^n}{1-\sqrt{1-\epsilon'}} \sqrt{(1+b)F(\mE_\theta)} \abs{d\theta} \leq \frac{\sqrt{(1+b)F(\mE_\theta)} }{1-\sqrt{1-\epsilon'}} \abs{d\theta}. 
\end{align}
Taking $d\theta \rightarrow 0$, we have 
\begin{equation}
    F(\rho_{\theta,n}) \leq \frac{(1+b)F(\mE_\theta) }{(1-\sqrt{1-\epsilon'})^2}. 
\end{equation}
In particular, since $b$ can be any positive number, and $1-\epsilon'$ can be arbitrarily close to the QFI contraction coefficient of $\mE_{\theta}$, we have 
\begin{equation}
    F(\rho_{\theta,n}) \leq \frac{F(\mE_\theta)}{(1-\sqrt{\eta(\mE_\theta)})^2}, 
\end{equation}
proving the theorem. 
\end{proof}

Note that the proof of \thmref{thm:contractive-upper} is general and does not require $\mE_\theta$ to be a qubit channel. In fact, it applies to any qudit channel $\mE_\theta$ satisfying $\eta(\mE_\theta) < 1$. However, whether \propref{prop:contraction} can be generalized to qudit cases is unknown.


\end{document}